%% file: main.tex
\documentclass{article}
\usepackage[utf8]{inputenc}
\usepackage[T1]{fontenc}
\usepackage{mathptmx}

\usepackage{pdfpages}
\usepackage{multicol}
\usepackage{proof}
\usepackage{bussproofs}
\usepackage{longtable}
\usepackage{amssymb}
\usepackage{float}
\usepackage{array}
\usepackage{caption}
\usepackage{stmaryrd}
\usepackage{marginnote}
\usepackage{colortbl}
\usepackage{mathtools}
\usepackage{multirow}
\usepackage{amsthm}
\usepackage{txfonts}
\usepackage{tikz}
\usepackage{geometry}
\usepackage{authblk}

\usetikzlibrary{arrows,fit,calc}
\usepackage{tablefootnote}
\usepackage{footnote}
\usepackage{hyperref}    
\makesavenoteenv{longtable} 

\EnableBpAbbreviations
\input{macros}
\title{Questions as cognitive filters\thanks{This project has received funding from the European Union’s Horizon 2020 research and innovation programme under the Marie Sk\l odowska-Curie grant agreement No 101007627. The research of the second author is supported by the NWO grant KIVI.2019.001 awarded to Alessandra Palmigiano.  The research of the fourth author was supported by the Luxembourg National Research Fund (FNR) (INTER/DFG/23/17415164/LODEX) and by the Key Project of the Chinese Ministry of Education (22JJD720021).}}
\author{Willem Conradie}
\author[2]{Krishna Manoorkar}
\author[2,3]{Alessandra Palmigiano}
\author[6]{Apostolos Tzimoulis}
\author[4,5]{Nachoem Wijnberg}
\affil[1]{School of Mathematics, University of the Witwatersrand, Johannesburg}
\affil[2]{School of Business and Economics, Vrije Universiteit Amsterdam}
\affil[3]{Department of Pure and Applied Mathematics, University of Johannesburg}
\affil[4]{Faculty of Economics and Business, University of Amsterdam}
\affil[5]{College of Business and Economics, University of Johannesburg}
\affil[6]{University of Luxembourg, Luxembourg}

\date{June 2025}

\begin{document}

\maketitle
\begin{abstract}
In this paper, we develop a logico-algebraic framework for modeling decision-making through deliberation in multi-agent settings. The central concept in this framework is that of {\em interrogative agendas}, which represent the cognitive stances of agents regarding which features should be considered relevant in the final decision. We formalize an agent’s interrogative agenda as an equivalence relation that identifies outcomes differing only in aspects the agent deems irrelevant. Moreover, we characterize the sublattices of the resulting lattice that correspond to relevant interrogative agendas for deliberation scenarios governed by different “winning rules.” We then introduce a two-sorted logico-algebraic structure—comprising the lattice of relevant interrogative agendas and the Boolean algebras of agent coalitions—to model the interaction between agents and agendas during deliberation. Finally, we discuss which interaction conditions can and cannot be defined within this framework.
\end{abstract}
\section{Introduction}
In the literature, there has been a growing interest in formally investigating the mechanisms of deliberation with the tools of social choice \cite{list2018democratic, dryzek2003social,adachi2024impossibility}, also in combination with  formal argumentation in multi-agent systems \cite{ganzer2019combining, bodanza2017collective, peirera2017handling}. This interest has been further fueled by the AI-driven, emerging opportunity, and urgent need, to develop  methods and tools for enhancing the engagement and participation  of individuals and social groups in democratic processes  \cite{helbing2023democracy,novelli2025testing,dobbe2021hard}. 

Agentic AI \cite{acharya2025agentic,wang2024survey} is another fast-growing area in which  formal tools for representing and analysing the  dynamics of deliberation are of paramount importance.   
Deliberative AI systems, inspired by philosophical theories of rational agency \cite{bratman1987intention} and formalized in architectures such as belief-desire-intention (BDI) software models \cite{rao1995bdi, meyer2015bdi} and its extensions such as BOID \cite{broersen2001boid,broersen2002goal} or A-BDI \cite{yu2025bdi}, enable agents to weigh competing objectives, manage uncertainty, and plan  strategies. Endowing artificial agents with the ability to represent multi-agent deliberation or to participate in deliberation with other (human) agents is particularly critical 
in areas such as autonomous systems \cite{chopra2018handbook,hu2024automated}, legal and normative reasoning \cite{peirera2017handling,markovich2021new,fang2023ai}, and AI governance \cite{boella2003norm,cranefield2019incorporating}.
Precisely in these contexts, formal models for deliberation will also be key to increase explainability. 

The present paper contributes to this line of research  by introducing a framework for representing and reasoning about the dynamics of decision-making via deliberation in a multi-agent setting. 
Rather than analysing the dynamics of deliberation in terms of e.g.~the various arguments offered by the decision-makers \cite{kok2011formal, rahwan2009argumentation}, or  in terms of preference or judgment aggregation \cite{list2007deliberation}, in the present approach, the dynamics of deliberation are analysed in terms
of the {\em cognitive stance}, or {\em attitude}, of each decision-making agent,  as expressed through the  aspects, features, issues, or parameters that each agent considers important for the decision-making task at hand.
The dynamics of deliberation are then represented by the way in  which each agent's cognitive attitude  changes during the deliberation process, and is then aggregated into a common  attitude, which determines the final decision.

In order to represent the cognitive stance of an agent or a group of agents, we take inspiration from a strand of literature at the intersection of philosophical logic and formal epistemology, which
has studied {\em questions} 
not only as a key function of natural language \cite{groenendijk1984studies}, but also in the context of the formal study of agency and multi-agent interaction, in close connection e.g.~with the dynamics of knowledge acquisition \cite{hintikka2002interrogative}. Driven by the idea that asking questions is a basic epistemic action, in \cite{van2012toward}, a dynamic epistemic logic framework for {\em issue management} is introduced, which  includes  updates induced by questions being asked, thereby 
``making questions  first-class citizens
in dynamic epistemic logic''. In \cite{Baltag2018}, a closely related framework is introduced for exploring  the {\em epistemic
 potential} of a group of agents, in line with Hintikka's model of inquiry as a goal-directed process:

\begin{quote}
    [...] questions act as “epistemic filters”
that structure and limit what one can know. [...] 
An agent’s questions shape her
knowledge and guide her learning. But as a consequence, they also impose limits on the agent’s knowledge and
ability to learn. In “normal” situations, agents can only learn information that fits their agendas: according to this
view, all we (can) know are answers to our own questions. As a consequence, the agents’ divergent interrogative
agendas can limit a group’s epistemic potential. [...]
\end{quote}
The agents' {\em interrogative agendas} in the quote above are ``the fundamental questions that they aim to address and get resolved''.

In the present paper, we explore the role of interrogative agendas not just as epistemic filters, but, more broadly, as {\em cognitive} filters, especially in deliberation processes. Specifically, the interrogative agenda of a given agent is understood as the set of  issues that this agent considers relevant or important, relative to the specific decision-making at hand. Hence, in this paper, this notion is used to capture the general cognitive stance, or attitude, of an agent involved in a decision-making process,  which determines the agent's behaviour in the decision-making. In particular,   the outcome of a decision-making process via deliberation will be determined by the way the individual interrogative agendas will be aggregated into one common agenda, and the dynamics of the deliberation process will formally be represented by the way in which the cognitive attitudes of the various agents influence this aggregation process. 

Through  case studies, we illustrate how processes of decision-making via deliberation can be formalized using the framework introduced in the present paper, and their various outcomes can be analysed and  predicted.

The framework introduced in the present paper is 
designed according to   
the  {\em multi-type} methodology, introduced in \cite{frittella2016multi,multi-typePDL, frittella2016DEL}, and further developed in \cite{frittella2016inquisitive, greco2023linear,greco2017lattice,bilkova2018logic,greco2021semi,greco2019bilattice, greco2019proper, chen2022non, conradie2021modelling}.  
Multi-type languages make it possible to define syntactic and semantic environments in which different types of entities interact with each other; hence, constituents such as actions, agents, or resources can be represented---not as {\em parameters} in the generation of formulas, but---as first-class citizens of the framework, via {\em terms} of their own specific type. 
The present framework 
showcases  the potential of the multi-type approach as a general platform for the meta-design of tools for reasoning about {\em interaction}. These tools can be especially useful for developing the foundations of  explainable agentic AI \cite{chakrabartycausal}.

\paragraph{Structure of this paper.} In Section \ref{interrogative:sec:preliminaries}, we collect preliminaries on interrogative agendas and their formalization, algebras of equivalence relations, the modal logic approach to equivalence relations, and characterization of join and meet generators of the lattice of equivalence relations. In Sections \ref{interrogative:sec: hiring committee} and \ref{interrogative:sec:car}, we discuss  case studies of decision-making via deliberation, and  analyze them in terms of the interrogative agendas of the agents involved in these decision-making processes.  In Section \ref{interrogative:sec:interrogative_agendas_and_coalitions}, we introduce a multi-type framework for describing and reasoning about coalitions of agents, their agendas, and deliberation dynamics. In Section \ref{interrogative:sec:interaction conditions}, we explore the    expressivity  of this framework, and, in Section \ref{interrogative:sec:Formalizing deliberation}, we use it to  formally model and predict various possible outcomes of the case study of Section  \ref{interrogative:sec: hiring committee}. Finally, in Section \ref{interrogative:sec:conclusion}, we conclude and discuss directions for future research.

\section{Preliminaries}\label{interrogative:sec:preliminaries}
In this section, we discuss preliminaries on interrogative agendas, their formalization as equivalence relations, and the algebra of equivalence relations. 

\subsection{Interrogative agendas and their algebraic environment}\label{interrogative:ssec:Interrogative agendas and their logical formalizations}
  In epistemology and formal philosophy, an epistemic agent’s (or a group of epistemic agents’, e.g.~users’) interrogative agenda (or research agenda \cite{enqvist2012modelling}) indicates the set of questions they are interested in, or what they want to know relative to a certain circumstance (independently of whether they utter the questions explicitly). 
  In each context, interrogative agendas act as cognitive filters that block content which is considered irrelevant by the agent and let through (possibly partial) answers to the agent's interrogative agenda. Only the information the agent considers relevant is actually absorbed (or acted upon) by the agent and used e.g.~in their decision-making, in the formation of their beliefs, etc. 
Interrogative agendas can be organized in hierarchies, and this hierarchical structure serves to establish whether a given interrogative agenda subsumes another, and defines different notions of “common ground” among agendas.

As discussed in the introduction, interrogative agendas are in essence (conjunctions of) questions. An influential approach in logic \cite{groenendijk1984studies,van2012toward} represents questions as equivalence relations over a suitable set of possible worlds $W$ (representing the possible states of affairs relative to a given situation).
More specifically, each question is associated with an equivalence relation corresponding to a partition of $W$, such that  each block in the partition consists of possible worlds in $W$ which have the same answer to the given question. For example, let $W= \{w_1, w_2, w_3, w_4\}$ and $p$ be a proposition stating `John killed Alan'. If $V(p)=\{w_1, w_2\}$, i.e.~$p$ is true at $w_1$ and $w_2$,  and false at $w_3$ and $w_4$, then the question `Did John kill Alan?' is represented by the equivalence relation $e$ associated with the partition $\mathcal{E}_e = \{\{w_1, w_2\},\{w_3, w_4\} \}$ illustrated below:  

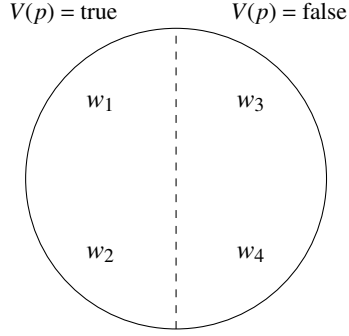
\begin{figure}[h]
\centering
\begin{tikzpicture}

\draw (0,0) circle (2cm);

\draw[dashed] (0,-2) -- (0,2);

\node at (-1,1) {$w_1$};
\node at (-1,-1) {$w_2$};

\node at (1,1) {$w_3$};
\node at (1,-1) {$w_4$};

\node at (-1.5, 2.2) {\small $V(p) = \text{true}$};
\node at (1.5, 2.2) {\small $V(p) = \text{false}$};


\end{tikzpicture}
\end{figure}

The equivalence relations on any set $W$ form a general (i.e.~not necessarily distributive) complete lattice $E(W)$ \cite{birkhoff1940lattice}. Lattices $E(W)$ serve as the algebraic environment  of interrogative agendas, as discussed in the next section. 

\subsection{Equivalence relations, partitions and preorders}\label{interrogative:ssec:equivalence relations and preorders}
An {\em equivalence relation} over a set $W$ is a binary relation $e\subseteq W\times W$ which is {\em reflexive} (i.e.~$(w,w)\in e$ for every $w\in W$), {\em symmetric} (i.e.~for all $w, u\in W$, if $(w, u)\in e$ then $(u, w)\in e$) and {\em transitive} (i.e.~for all $w, u, v\in W$, if $(w, u)\in e$ and $(u, v)\in e$ then $(w, v)\in e$). For any $w\in W$, we let $[w]_e: = \{u\in W\mid (w, u)\in e\}$. A partition  of $W$ is a collection $\mathcal{E}\subseteq \mathcal{P}(W)\setminus \{\varnothing\}$ such that $\bigcup \mathcal{E} = W$ and the elements of $\mathcal{E}$ are pairwise disjoint.  Each equivalence relation $e$ on $W$ gives rise to the partition  $\mathcal{E}_e: = \{[w]_e\mid w\in W\}$; conversely, every partition $\mathcal{E}$ on $W$ gives rise to the equivalence relation $e_{\mathcal{E}}$ such that if $u, v\in W$, then $(u, v)\in e_{\mathcal{E}}$ iff $u\in C$ and $v\in C$ for some $C\in \mathcal{E}$. Clearly, for every $e$ and $\mathcal{E}$,
\[e = e_{\mathcal{E}_e}\quad \mbox{ and }\quad \mathcal{E} = \mathcal{E}_{e_{\mathcal{E}}},\]
so that  equivalence relations and partitions can be identified.

    A {\em preorder} over a set $W$  is a reflexive and transitive binary relation ${\leq} \subseteq W \times W$. Every pre-order  $\leq$ on $W$ gives rise to an equivalence relation $e_\leq$ on $W$ defined as   $w_1 e_\leq w_2$ iff $w_1 \leq w_2$ and $w_2 \leq w$. Conversely, any equivalence relation $e$ on a preordered set $(W, \leq)$ gives rise to  a preorder $\leq_e$ on $W$ defined as follows:  $w \leq_e u$ iff for all $w'\in [w]_e$ some $u' \in [u]_e$ exists such that $w' \leq u'$. That is, $w \leq_e u$ iff for any $w'$ which is $e$-equivalent to $w$,   some  $u'$ exists which is $e$-equivalent to $u$ and  $w' \leq u'$. 
    
    If $e$ and $\leq$ are an equivalence relation and a pre-order on $W$, then $e$ is  {\em compatible with}   $\leq$ if  $e\, \circ \leq \circ\, e\subseteq {\leq}$, where $\circ$ denotes the operation of relational composition, and is {\em strongly compatible with} $\leq$ if $e$ is  compatible with $\leq$ and  $e_\leq \subseteq e$.  The reflexivity of $e$ implies that if $e$ is  compatible with   $\leq$, then  $e\, \circ \leq \circ\, e = {\leq}$.  By definition, if $\leq$ is a preorder, then $e_\leq$ is strongly compatible with $\leq$.

    

\begin{lemma}\label{interrogative:lem:preorder-equivalence back-forth}
For any  preorder  $\leq$ on $W$,
\begin{enumerate}
    \item  $\leq_{e_\leq} = {\leq}$.
    \item For   any equivalence relation $e$, if $e$ is strongly compatible with $\leq$, then $e_{\leq_e}=e$. 
\end{enumerate}
\end{lemma}
\begin{proof}
1. By definition,  $w \leq_{e_\leq} u$ iff for every $w' \in W$, if $(w',w)\in e_\leq$ then $w' \leq u'$ for some $u'\in W$ s.t.~$(u',u)\in e_\leq$. Hence, by instantiating $w':=w$, $w \leq u'$ for some $u'\in W$ s.t.~$(u',u)\in e_\leq$. By definition,  $(u,u')\in  e_\leq $ implies $u' \leq u$. Therefore, as $\leq$ is transitive, $w \leq u$, as required.

Conversely, let $w, w', u\in W$ s.t.~$w  \leq  u$ and $(w, w')\in  e_\leq$. By definition, $(w, w')\in  e_\leq$ implies $w' \leq w$.  Therefore, since $\leq$ is transitive, $w' \leq u$. Hence, for every such $w'$, choosing $u':=u$  ensures that $w' \leq u'$, which proves  $w \leq_{e_\leq} u$, as required. 

2. By definition, $(w, u)\in e_{\leq_e}$ iff $w\leq_e u$ and $u\leq_e w$. If $w, w', u\in W$ s.t.~$(w, u)\in  e$ and $(w', w)\in e$, then $(w', w)\in e$, $w \leq w$, and $(w, u)\in e$ imply $w' \leq u$  by the  assumption that $e$ is compatible with $\leq$. Hence, for every such $w'$, choosing $u':=u$  ensures that $w' \leq u'$, which proves  $w \leq_{e} u$.  By symmetry, $(w, u)\in  e$ implies $(u, w)\in  e$, which implies $u \leq_e w$ by analogous reasoning. Therefore, $(w, u)\in e_{\leq_e}$, as required.

Conversely, if $(w, u)\in e_{\leq_e}$, then, by definition, $w \leq_e u$ and $u \leq_e w$. By instantiating $w':=w$ in its defining clause,  $w \leq_e u$ implies $w \leq u'$ for some $u'\in W$ s.t.~$(u, u')\in e $, which implies  $w \leq u$, since $e\, \circ \leq \circ\, e\subseteq {\leq}$ by assumption. Similarly, one shows that $u \leq_e w$ implies $u \leq w$, which shows that $(w, u)\in e_\leq \subseteq e$, as required.  
\end{proof}

   \begin{definition}
    \label{interrogative:def:e prefers}
For every $w, u\in W$,  $e$ {\em prefers $u$ over $w$} iff $w\leq_e u$ and $u\not\leq_e w$.
\end{definition}

\subsection{Algebras of equivalence relations}\label{interrogative:ssec:Algebra of equivalence relations}

 For every set $W$, let $E(W)\subseteq\mathcal{P}(W\times W)$  denote the subposet of the equivalence relations on $W$ ordered by inclusion. This subposet is in fact a complete $\bigcap$-subsemilattice of $\mathcal{P}(W\times W)$, and hence is a complete lattice the top element of which is $\tau: = W\times W$, the bottom element is $\epsilon: = \Delta = \{(w, w)\mid w\in W\}$, and moreover for any $\mathcal{X}\subseteq E(W)$,
\[ \bigsqcap_{E(W)}\mathcal{X}: = \bigcap\mathcal{X}\quad \mbox{ and }\quad \bigsqcup_{E(W)}\mathcal{X}: = \mathsf{RST}(\bigcup\mathcal{X}),\]
where $\mathsf{RST}(R): = \bigcup_{n\in \omega} (R^{-1}\cup R)^{n}$ denotes the reflexive, symmetric and transitive closure\footnote{We let $R^{-1}: = \{(u, w)\in W\times W\mid (w, u)\in R\}$, and for every $n\in \mathbb{N}$,    $R^0: = \Delta$ and $R^{n+1}: = R^n\circ R$, where $S\circ T: = \{(w, v)\mid (w, u)\in S \mbox{ and } (u, v)\in T \mbox{ for some } u\in W\}$ for all $S, T\in\mathcal{P}(W\times W)$.} of any $R\in \mathcal{P}(W\times W)$.
Finite meets and joins in $E(W)$ are more easily visualized when we switch to partitions. Indeed, if $e_1$ and $e_2$ are represented as the partitions on the left-hand side of the picture below, then $e_1\sqcap e_2$ is represented by the partition obtained by {\em superposing} the lines of $e_1$ and $e_2$ (more formally, the partition capturing $e_1\sqcap e_2$ is $\{[w]_{e_1}\cap [w]_{e_2}\mid w\in W \}$), while  $e_1\sqcup e_2$ is represented as the partition obtained by {\em erasing}   the lines that $e_1$ and $e_2$ do not share (in other words, the partition capturing $e_1\sqcup e_2$ is the one obtained by taking the union of overlapping cells of the two partitions).
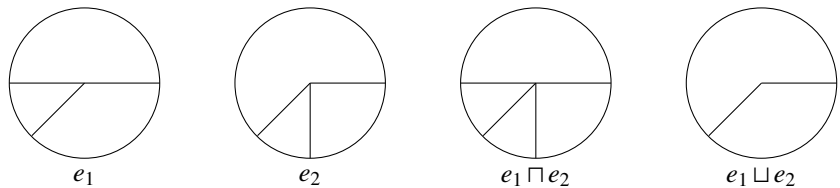
\begin{center}
\begin{tikzpicture}
\draw (0,0) circle (1cm);
\draw (3,0) circle (1cm);
\draw (6,0) circle (1cm);
\draw (9,0) circle (1cm);
\draw (-1,0) -- (1,0);
\draw (0,0) -- (-0.7,-0.7);
\draw (3,0) -- (4,0);
\draw (3,0) -- (3,-1);
\draw (3,0) -- (2.3,-0.7);
\draw (5,0) -- (7,0);
\draw (6,0) -- (5.3,-0.7);
\draw (6,0) -- (6,-1);
\draw (9,0) -- (10,0);
\draw (9,0) -- (8.3,-0.7);
\draw (0, -1.25) node {$e_1$};
\draw (3, -1.25) node {$e_2$};
\draw (6, -1.25) node {$e_1\sqcap e_2$};
\draw (9, -1.25) node {$e_1\sqcup e_2$};
\end{tikzpicture}
\end{center}
The lattice $E(W)$ is in general non-distributive, as is witnessed e.g.~when instantiating $W$ as the three-element set $\{a, b, c\}$,  or the four-element set $\{a, b, c, d\}$, with the corresponding $E(W)$ being the lattices represented as the  Hasse diagrams in Figure \ref{interrogative:fig: E(a, b, c) and E(a, b, c, d)}.
\begin{figure}
\begin{center}
\begin{tikzpicture}
\draw[very thick] (0, -1) -- (0, 1) --(-1, 0) -- (0, -1) -- (1, 0) -- (0, 1);
	
	\filldraw[black] (0,-1) circle (2 pt);
	\filldraw[black] (0, 1) circle (2 pt);
	\filldraw[black] (-1, 0) circle (2 pt);
	\filldraw[black] (1, 0) circle (2 pt);
	\filldraw[black] (0, 0) circle (2 pt);
		\draw (0, 1.3) node {$\{\{a, b, c\}\}$};
	\draw (0, -1.3) node {$\{\{a\}, \{b\}, \{c\}\}$};
	\draw (-1.9, 0) node {{\small{$\{\{a\}, \{b, c\}\}$}}};
    \draw (0.3, 0) node {{\small{$e_b$}}};
    \draw (1.9, 0) node {{\small{$\{\{a, b\}, \{c\}\}$}}};


    \filldraw[black] (8,2) circle (2 pt);

    \filldraw[black] (6,1) circle (2 pt); 
\filldraw[black] (8,1) circle (2 pt); 
\filldraw[black] (10,1) circle (2 pt); 

  \filldraw[black] (5,1) circle (2 pt); 
   \draw (4.7, 1.1) node {{\small{$e_{aa}$}}};
    \filldraw[black] (7,1) circle (2 pt); 
      \draw (6.7, 1.1) node {{\small{$e_{bb}$}}};
       \filldraw[black] (9,1) circle (2 pt); 
         \draw (9.3, 1.1) node {{\small{$e_{cc}$}}};
          \filldraw[black] (11,1) circle (2 pt); 
          \draw (11.3, 1.1) node {{\small{$e_{dd}$}}};
          \draw[very thick] (8, 2) -- (5, 1);
           \draw[very thick] (8, 2) -- (7, 1);
            \draw[very thick] (8, 2) -- (9, 1);
             \draw[very thick] (8, 2) -- (11, 1);

              \draw[very thick] (5, 1) -- (5.5, 0); 
              \draw[very thick] (5, 1) -- (7.5, 0);
              \draw[very thick] (5, 1) -- (9.5, 0);
	
              \draw[very thick] (7, 1) -- (5.5, 0); 
              \draw[very thick] (7, 1) -- (8.5, 0);
              \draw[very thick] (7, 1) -- (10.5, 0);

              \draw[very thick] (9, 1) -- (6.5, 0); 
              \draw[very thick] (9, 1) -- (7.5, 0);
              \draw[very thick] (9, 1) -- (10.5, 0);

              \draw[very thick] (11, 1) -- (6.5, 0); 
              \draw[very thick] (11, 1) -- (8.5, 0);
              \draw[very thick] (11, 1) -- (9.5, 0);
\filldraw[black] (5.5,0) circle (2 pt); 
 \draw (5.3, -0.1) node {{\small{$e_{ab}$}}};
\filldraw[black] (6.5,0) circle (2 pt); 
 \draw (6.3, -0.1) node {{\small{$e_{cd}$}}};
\filldraw[black] (7.5,0) circle (2 pt); 
 \draw (7.3, -0.1) node {{\small{$e_{ac}$}}};
\filldraw[black] (8.5,0) circle (2 pt); 
 \draw (8.8, -0.1) node {{\small{$e_{bd}$}}};
\filldraw[black] (9.5,0) circle (2 pt); 
 \draw (9.8, -0.1) node {{\small{$e_{ad}$}}};
\filldraw[black] (10.5,0) circle (2 pt); 
 \draw (10.8, -0.1) node {{\small{$e_{bc}$}}};
\filldraw[black] (8,-1) circle (2 pt);
\draw[very thick] (8, 2) -- (6, 1) --
	(5.5, 0) -- (8, -1) -- (6.5, 0) -- (6, 1);	
\draw[very thick] (8, 2) -- (8, 1) --
	(7.5, 0) -- (8, -1) -- (8.5, 0) -- (8, 1);	
\draw[very thick] (8, 2) -- (10, 1) --
	(9.5, 0) -- (8, -1) -- (10.5, 0) -- (10, 1);	
    	\draw (8, 2.3) node {$\{\{a, b, c, d\}\}$};
    \draw (8, -1.3) node {$\{\{a\}, \{b\}, \{c\}, \{d\} \}$};
\end{tikzpicture}
\end{center}
\caption{Hasse diagrams of the lattices $E(\{a, b, c\})$,  and  $E(\{a, b, c, d\})$. In the diagram on the left, $e_b$ corresponds to the partition $\{\{b\}, \{a, c\}\}$, and in the diagram on the right,  $e_{xy} = \{\{x\}, \{y\}, W\setminus\{x, y\}\}$ for all $x, y\in \{a, b, c, d\}$.}
\label{interrogative:fig: E(a, b, c) and E(a, b, c, d)}
\end{figure}
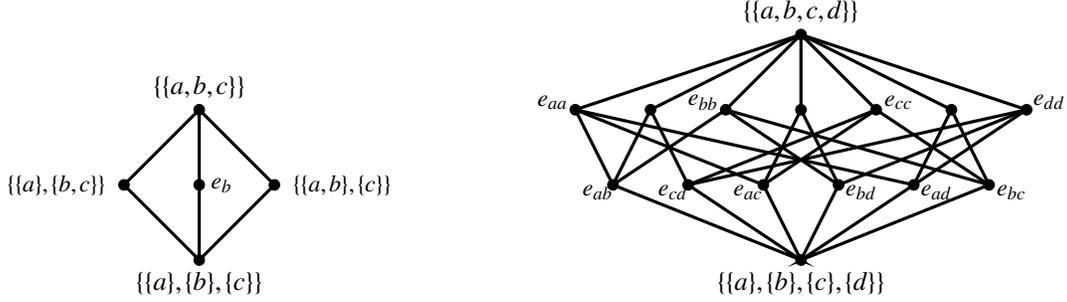

 Lattices of equivalence relations have been extensively studied \cite{birkhoff1935structure,Ore1943-ORETOE}. In particular, 
 every general lattice is a sublattice of $E(W)$ for some set $W$ \cite{whitman1946lattices}. This 
immediately implies that the negation-free fragment of classical propositional logic without the distributivity axioms (which 
we refer to as the basic nondistributive logic) is sound and complete w.r.t. the class of lattices of equivalence relations. Hence, 
the basic nondistributive logic can be regarded as the basic logic of interrogative agendas.

\subsection{The standard modal logic approach to equivalence relations} 
\label{interrogative:ssec:modal logic}
Representing equivalence relations as partitions is also useful to discuss two important constructions in modal logic. 
Each equivalence relation $e$ on $W$ induces the  semantic (normal) modal operators $\langle e\rangle, [e] : \mathcal{P}(W)\to \mathcal{P}(W)$  defined in the following, standard way. For any $X\subseteq W$, \[\langle e \rangle X: = e^{-1}[X]: = \{w\in W\mid (w, x)\in e\mbox{ for some } x\in X\} = \bigcup\{[x]_e\mid x\in X\}\] \[ [e]X: = (e^{-1}[X^c])^c: = \{w\in W\mid (w, x)\in e\mbox{ for all } x\in X\} = \bigcup\{[x]_e\mid x\in X \mbox{ and } [x]_e\subseteq X\},\]
where $(\cdot)^c$ denotes the relative complement. Thus, $\langle e \rangle$ maps $X$ to the union of the $e$-equivalence classes which have nonempty intersection with $X$, and $[e]$ maps $X$ to the union of the $e$-equivalence classes  contained in $X$, as illustrated in the following picture:
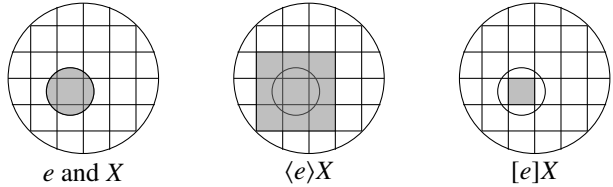
\begin{center}
\begin{tikzpicture}
\draw (0,0) circle (1cm);
\draw (-1,0) -- (1,0);
\draw (0,1) -- (0,-1);
\draw (0.7,0.7) -- (0.7,-0.7);
\draw (0.7,-0.7) -- (-0.7,-0.7);
\draw (-0.7,-0.7) -- (-0.7,0.7);
\draw (-0.7,0.7) -- (0.7,0.7);
\draw (0.35,0.95) -- (0.35,-0.95);
\draw (-0.35,0.95) -- (-0.35,-0.95);
\draw (0.95,0.35) -- (-0.95,0.35);
\draw (0.95,-0.35) -- (-0.95,-0.35);
\draw (3,0) circle (1cm);
\draw (2,0) -- (4,0);
\draw (3,1) -- (3,-1);
\draw (3.7,0.7) -- (3.7,-0.7);
\draw (2.3,0.7) -- (2.3,-0.7);
\draw (2.3,0.7) -- (3.7,0.7);
\draw (2.3,-0.7) -- (3.7,-0.7);
\draw (3.35,0.95) -- (3.35,-0.95);
\draw (2.65,0.95) -- (2.65,-0.95);
\draw (3.95,0.35) -- (2.05,0.35);
\draw (3.95,-0.35) -- (2.05,-0.35);
\draw (6,0) circle (1cm);
\draw (5,0) -- (7,0);
\draw (6,1) -- (6,-1);
\draw (6.7,0.7) -- (6.7,-0.7);
\draw (5.3,0.7) -- (5.3,-0.7);
\draw (5.3,0.7) -- (6.7,0.7);
\draw (5.3,-0.7) -- (6.7,-0.7);
\draw (6.35,0.95) -- (6.35,-0.95);
\draw (5.65,0.95) -- (5.65,-0.95);
\draw (6.95,0.35) -- (5.05,0.35);
\draw (6.95,-0.35) -- (5.05,-0.35);
\draw (0, -1.25) node {$e$ and $X$};
\draw (3, -1.25) node {$\langle e\rangle X$};
\draw (6, -1.25) node {$[e]X$};
\draw[fill=gray, opacity=0.5] (-0.175,-0.175) circle (9pt);
\draw (-0.175,-0.175) circle (9pt);
\draw (2.825,-0.175) circle (9pt);
\draw (5.825,-0.175) circle (9pt);
\draw[fill=gray, opacity=0.5] (2.3,-0.7) rectangle (3.35,0.35);
\draw[fill=gray, opacity=0.5] (5.65,-0.35) rectangle (6,0);
\end{tikzpicture}
\end{center}
Because $e$ is symmetric, $\langle e\rangle$ and $[e]$ are both dual to each other (i.e.~$[e]X: = (\langle e\rangle X^c)^c$) and form an adjoint pair,  i.e.~for all $X, Y\subseteq W$,
\begin{equation}\langle e\rangle X\subseteq Y\quad \mbox{ iff }\quad X\subseteq [e]Y.\end{equation}
  As usual, $\langle e\rangle$  preserves arbitrary joins in $\mathcal{P}(W)$ (hence also $\varnothing$ as the empty join) and $[e]$ preserves arbitrary meets in $\mathcal{P}(W)$ (hence also $W$ as the empty meet). The following properties  concern the interaction between the algebraic structure of  $E(W)$ and the semantic modal operators:
\begin{proposition}\label{semantic properties IA}
For any set $W$, any $X\subseteq W$, and $e_1, e_2\in E(W)$,
\begin{enumerate}
\item $\langle\tau\rangle X =\begin{cases}
  W & \text{if } X\neq \varnothing\\
  \varnothing & \text{if } X = \varnothing\\
  \end{cases}$
  \;\; and\;\; $[\tau] X =\begin{cases}
  \varnothing & \text{if } X\neq W\\
  W & \text{if } X = W;\\
  \end{cases}$
  \item $\langle \epsilon\rangle X =  X$ and $[\epsilon] X =  X$;
\item if $e_1\subseteq e_2$ then $\langle e_1\rangle X\subseteq \langle e_2\rangle X$. 

\item if $e_1\subseteq e_2$ then $[e_2]X\subseteq [e_1]X$. 
\end{enumerate}
\end{proposition}
\begin{proof}
Item 1 follows immediately from $[x]_{\tau} = W$ for every $x\in W$.
Item 2 follows immediately from $[x]_{\epsilon} = \{x\}$ for every $x\in W$. As to 3,
by assumption, $[x]_{e_1}\subseteq  [x]_{e_2}$ for all $x\in W$. Hence, $\langle e_1\rangle X = \bigcup\{[x]_{e_1}\mid x\in X\}\subseteq \bigcup\{[x]_{e_2}\mid x\in X\} = \langle e_2\rangle X$.
As to 4,
by assumption, $[x]_{e_1}\subseteq  [x]_{e_2}$ for all $x\in W$. Hence, if $x\in [e_2] X$, i.e.~$x\in X$ and $[x]_{e_2}\subseteq X$, then $[x]_{e_1}\subseteq X$, and hence $x\in [e_1] X$, as required.
\end{proof}

\begin{remark}
It is not true in general that
$\langle e_1\sqcup e_2\rangle X = \langle e_1\rangle X \cup \langle e_2\rangle X$. Indeed, let $W = A\cup B\cup C$ where $A = \{a_1, a_2, a_3\}$, $B=\{ b_1, b_2, b_3\}$, $C=\{ c_1, c_2, c_3\}$; consider $e_1 = \{A, B, C\}$ and $e_2 = \{D_1, D_2, D_3\}$ where $D_i = \{a_i, b_i, c_i\}$ for $1\leq i\leq 3$, and $X = \{a_2, b_3\}$. Clearly, $e_1\sqcup e_2 = \tau$, and hence $\langle e_1\sqcup e_2 \rangle X = W$ (cf.\ Proposition \ref{semantic properties IA}.1). However, $\langle e_1\rangle X = A\cup B$ and $\langle e_2\rangle X = D_2\cup D_3$. Hence, $\langle e_1\rangle X \cup \langle e_2\rangle X = (A\cup B)\cup (D_2\cup D_3)\neq W$.

This example is also a counterexample for $\langle e_1\sqcap e_2\rangle X = \langle e_1\rangle X \cap \langle e_2\rangle X$. Indeed, $e_1\sqcap e_2 = \epsilon$, and hence $\langle e_1\sqcap e_2 \rangle X = X$ (cf.\ Proposition \ref{semantic properties IA}.1); however, $\langle e_1\rangle X \cap \langle e_2\rangle X = (A\cup B)\cap (D_2\cup D_3)\neq X$.

This same example is a counterexample for $  [ e_1\sqcap e_2] X  = [e_1] X \cup [ e_2] X$. Indeed,  $[ e_1\sqcap e_2] X = [\epsilon] X = X$ (cf.\ Proposition \ref{semantic properties IA}.1); however, $X$ does not contain any equivalence cell of $e_1$ or of $e_2$, and hence $[e_i]X = \varnothing$ for $1\leq i\leq 2$, thus $[e_1] X \cup [ e_2] X = \varnothing\neq X$.

Finally, it is not true that  $[ e_1\sqcup e_2] X = [ e_1] X \cap  [e_2] X$.  Indeed, let $W$, $e_1$ and $e_2$ be as in the example above, and let $X: = A\cup D_1$. Then $[ e_1\sqcup e_2] X = [\tau] X = \varnothing$ (cf.\ Proposition \ref{semantic properties IA}.1); however, $[e_1] X = A$ and $[e_2]X = D_1$, hence $[ e_1] X \cap  [e_2] X = \{a_1\}\neq \varnothing$.
\end{remark}

\subsection{Generating the lattice of equivalence relations} We let $\jty(E(W))$ (resp.~$\mty(E(W))$) denote the set  of  completely join-irreducible (resp.~meet-irreducible) elements of $E(W)$.\footnote{For any bounded lattice $L$ and any $x\in L$, $x$ is {\em completely join-irreducible} if $x\neq \bot$, and  $x = \bigvee X$ for $X\subseteq L$ implies $x\in X$; dually, $x$ is {\em completely meet-irreducible} if $x\neq \top$, and  $x = \bigwedge X$ for $X\subseteq L$ implies $x\in X$.  }
\begin{proposition}
\label{interrogative:prop:charact meet-irr}
For any set $W$, if $|W|\geq 2$, then $e\in \mty(E(W))$ iff $e$ is identified by some  partition of the form $\mathcal{E}_{X}: = \{X, W\setminus X\}$ with $\varnothing \subsetneq X \subsetneq W$.
\end{proposition}
\begin{proof}
Let $e$ be identified by $\mathcal{E}_X$ as above. Then, by construction, $e\neq \tau$. To show that $e$ is completely meet-irreducible, it is enough to show that, for every $e'\in E(W)$, if $e\subsetneq e'$, then $e' = \tau$ (so we are actually showing something stronger, namely that $e$ is a coatom). If $(x_0, y_0)\in e'$  but $(x_0, y_0)\notin e$, then, modulo renaming of variables, we can assume that $x_0\in X$ and $y_0\in W\setminus X$. Then,  each $x\in X$ is $e$-equivalent, hence $e'$-equivalent, to $x_0$, which is $e'$-equivalent to $y_0$, which is $e$-equivalent, hence $e'$-equivalent, to any $y\in W\setminus X$. Since $e'$ is transitive, this shows that $e'$ is the total relation, as required.

Conversely, let $e$ be a completely meet-irreducible element of $E(W)$. Then $e\neq \tau$, hence some $x, y\in W$ exist such that $(x, y)\notin e$. Since $e$ is reflexive, it must be $x\neq y$. Let $X\subsetneq W$ and $Y\subsetneq W$ respectively denote the $e$-equivalence classes of  $x$ and $y$. To complete the proof, it is enough to show that  $(W\setminus X)\subseteq Y$. Let $z\in W\setminus X$, and assume for contradiction that $z\notin Y$, i.e.~$(z, y)\notin e$. Hence, $Z\neq X$ and $Z\neq X$, where $Z$ denotes the $e$-equivalence class of $z$. Then let $e_1, e_2\in E(W)$ be respectively identified by the  partitions $\mathcal{E}_1: = \{X\cup Y, Z\}\cup \{[w]_e\mid w\in W\setminus(X\cup Y\cup Z)\}$ and $\mathcal{E}_2: = \{X\cup Z, Y\}\cup \{[w]_e\mid w\in W\setminus(X\cup Y\cup Z)\}$. By construction, $e = e_1\sqcap e_2$; however, $e\neq e_1$ and $e\neq e_2$,  contradicting the meet-irreducibility of  $e$.
\end{proof}
Thus, the meet-irreducible elements  of the lattice $E(W)$ can be identified with the binary (i.e. yes/no) questions  on $W$.

We observe that, as an immediate consequence of the proposition above, the completely meet-irreducibles and the coatoms of $E(W)$ coincide. However, notice that being a coatom does not imply being  meet-prime.\footnote{For any bounded lattice $L$ and any $x\in L$, $x$ is {\em completely join-prime} if $x\neq \bot$, and  $x \leq \bigvee X$ for $X\subseteq L$ implies $x\leq x'$ for some $x'\in X$; dually, $x$ is {\em completely meet-prime} if $x\neq \top$, and  $\bigwedge X\leq x$ for $X\subseteq L$ implies $x'\leq x$ for some $x'\in X$.  } As an example, consider $E(\{a, b, c\})$  (cf.~Figure \ref{interrogative:fig: E(a, b, c) and E(a, b, c, d)}). Let $e_a, e_b, e_c$ be the three coatoms of $E(W)$. Then $e_a\sqcap e_c = \epsilon\leq e_b$; however, $e_a\nleq e_b$ and $e_c\nleq e_b$. 
\begin{proposition}
\label{interrogative:prop:charact join-irr}
For any set $W$, if $|W|\geq 2$, then  $e\in \jty(E(W))$ iff $e$ is identified by some partition of the form $\mathcal{E}_{xy}: = \{\{x, y\}\}\cup \{\{z\}\mid z\in W\setminus \{x, y\} \}$ with $x, y \in W$ such that $x\neq y$.
\end{proposition}
\begin{proof}
By construction, any $e$ corresponding to some $\mathcal{E}_{xy}$ as above is an atom, and hence a completely join-irreducible element of $E(W)$. Conversely,  let $e$ be a completely join-irreducible element of $E(W)$. Then $e\neq \epsilon$, hence some $x, y\in W$ exist such that $x\neq y$ and $(x, y)\in e$. To show that $e$ is identified by the partition $\mathcal{E}_{xy}: = \{\{x, y\}\}\cup \{\{z\}\mid z\in W\setminus \{x, y\} \}$, we need to show that (a) the $e$-equivalence class $X$ of $x$ cannot contain  three pairwise distinct elements, and (b) for any $z\in W\setminus \{x, y\}$, the $e$-equivalence class $Z$ of $z$ is a singleton. As to (a), assume for contradiction that $X$ contains  three pairwise distinct elements $x, y, z$. Hence, $X$ is both the union of two disjoint subsets $X_1$ and $X_2$ such that $x\in X_1$ and $y, z\in X_2$ and is the union of two disjoint subsets $Y_1$ and $Y_2$ such that $y\in Y_1$ and $x, z\in Y_2$. Then  let $e_1, e_2\in E(W)$ be  identified respectively by the  partitions $\mathcal{E}_1: = \{\{X_1, X_2\}\cup \{[w]_e\mid w\in W\setminus X\}$ and $\mathcal{E}_2: = \{Y_1, Y_2\}\cup \{[w]_e\mid w\in W\setminus X\}$. By construction,  $e = e_1\sqcup e_2$; however, $e\neq e_1$ and $e\neq e_2$,  contradicting the join-irreducibility of $e$. As to (b), by (a) and the assumptions, $X$ and $Z$ are disjoint and $X$ contains two distinct elements. If $Z$ contains some $z'\in W$ such that $z\neq z'$,  then $Z$ is  the union of two disjoint subsets $Z_1$ and $Z_2$ such that $z\in Z_1$ and $z'\in Z_2$. Then  let $e_1, e_2\in E(W)$ be respectively identified by the  partitions $\mathcal{E}_1: = \{\{Z_1, Z_2\}\cup \{[w]_e\mid w\in W\setminus Z\}$ and $\mathcal{E}_2: = \{\{x\}, \{y\}\}\cup \{[w]_e\mid w\in W\setminus X\}$. By construction,  $e = e_1\sqcup e_2$; however, $e\neq e_1$ and $e\neq e_2$, again contradicting the join-irreducibility of $e$.
\end{proof}
We observe that, as an immediate consequence of the proposition above, the completely join-irreducibles and the atoms of $E(W)$ coincide. However, notice that being an atom does not imply being join-prime. As an example, consider $E(\{a, b, c\})$  (cf.~Figure \ref{interrogative:fig: E(a, b, c) and E(a, b, c, d)}). Let $e_a, e_b, e_c$ be the three coatoms of $E(W)$. Then $e_b\leq e_a\sqcup e_c = \tau$; however, $e_b\nleq e_a$ and $e_b\nleq e_c$.

From the previous propositions, it follows that, for any finite set $W$ such that $|W| = n\geq 2$, the cardinality of the set $\mty(E(W))$ of  completely meet-irreducible elements of $E(W)$ is $2^{n-1}-1$, and the cardinality of the set $\jty(E(W))$ of completely join-irreducible elements of $E(W)$ is $\frac{n^2-n}{2}$. 

 Moreover, it also follows that every $e\in E(W)$ is both the join of the completely join-irreducible elements below it and the meet of the completely meet-irreducible elements above it. In other words, $E(W)$ is a perfect lattice  dual to the {\em RS-polarity} (cf.~\cite[Definition 2.12]{gehrke2006generalized}) $(\jty(E(W)), \mty(E(W)), \leq)$.

 \section{Deliberation in a hiring committee}
\label{interrogative:sec: hiring committee}

Let us introduce the role of interrogative agendas in deliberation through 
the following scenario, from  \cite[Example 3]{Baltag2018}:
\begin{quote}
  John and
Mary are two candidates for an open position [...]. Both candidates claim
to work at the intersection of Philosophy and Logic. Both candidates have asked their supervisors to write letters
of reference for them and it now happens that John has better letters of reference (by far) than Mary. Our hiring
committee consists of two members, Alan and Betty. Alan is our philosophy expert, who doesn't understand much
formal logic. He knows that John's philosophical writing is slightly better than Mary. But Alan can't judge the
logic proofs of the candidates. Betty is a formal logician, she knows that Mary's logical work is really great, it fully
backs her philosophical claims: indeed, Mary's work is much better than John's (John's proofs are full of mistakes).
Both Alan and Betty can see that John has much better references than Mary, so this is common knowledge in the
group $G = \{Alan, Betty\}$. We assume that these three facts together imply that Mary is the best candidate overall.
Moreover, this fact is distributed knowledge in the group $G = \{Alan, Betty\}$. But is this really the decision that
the committee will reach? [...] the answer to this will depend on the issues that the committee
members consider to be relevant for this hiring. The winner will be the candidate who performs better on (all and
only) the issues which are agreed to be relevant by the committee members.
\end{quote}
So, scoring\footnote{The binary score captures the gist of this deliberation scenario, while keeping the space of candidate profiles  relatively small. Later on in this section, different types of scores will also be considered, including the many-valued linear scores adopted  in \cite{Baltag2018}.} John and Mary over $0$ and $1$,  John does better than Mary on the letters of reference, they do  equally well on the philosophy parameter, and Mary does better than John on the logic parameter.
As stressed in \cite{Baltag2018}, key to this decision-making is which parameters Alan  and Betty  consider  {\em relevant}. 
We propose that this is in fact the {\em sole key} aspect,  
and specifically, that the essentials of this story are actually {\em independent} from what the agents know (either individually or as a group). That is, even under the much stronger assumption that the performances of the candidates on each parameter are common knowledge between Alan and Betty, they might still consider some parameters not relevant, and let their respective positions depend only on the parameters they consider relevant for the decision-making, i.e.~their interrogative agendas.

In this paper, we propose that the possible outcomes of this and other deliberation scenarios can be analyzed  in terms of how agents come up with a {\em coalition  agenda} starting from their individual interrogative agendas.\footnote{In \cite{Baltag2018},  outcomes, and hence interrogative agendas leading to them, are classified as good or bad; we do not take such a position. The only distinction we make is between interrogative agendas that lead to an outcome and interrogative agendas that do not (cf.~Definition \ref{interrogative:def:e prefers}). 
}
In what follows, we analyze this deliberation scenario using the tools introduced  in the previous section.

\paragraph{Total dominance and binary scores.} Let
\[W =\{(w_r, w_p, w_l)\mid w_p, w_r, w_l\in \{0, 1\}\}\]
be the space of possible binary evaluations, or profiles, over the set $X\coloneqq \{r,p,l\}$ of initial parameters.\footnote{In the literature  \cite{gardenfors2004conceptual}, $W$ is sometimes  referred to  as the {\em feature space}.} In this space, John and Mary can be identified by the triples $(1, 1, 0)$ and $(0, 1, 1)$ respectively, given their relative performances on each parameter.
The {\em total dominance} winning rule for this deliberation scenario is encoded by the following partial order on $W$: for any $w, u\in W$, $w\leq u$ iff $w_x\leq u_x$ for every $x\in \{r, p, l\}$. That is, one candidate {\em dominates} the other if the two candidates have different profiles, and the one of the second candidate is coordinatewise less than or equal to  the one of the first. 
The  poset  $P = (W, \leq )$ can be represented as follows:
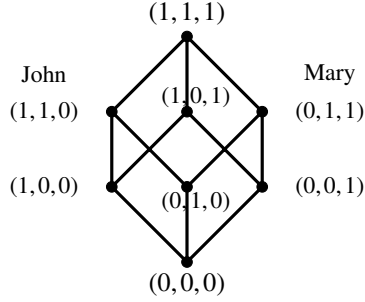
\begin{center}
\begin{tikzpicture}
\draw[very thick] (-1, 0) -- (-1, 1) --
	(0, 0) -- (1, 1) -- (1, 0) -- (0, 1) -- (-1, 0);
	\draw[very thick] (0, 2) -- (-1, 1);
\draw[very thick] (0, 2) -- (0, 1);
\draw[very thick] (0, 2) -- (1, 1);
	\draw[very thick] (0, -1) -- (-1, 0);
\draw[very thick] (0, -1) -- (0, 0);
\draw[very thick] (0, -1) -- (1, 0);
	\filldraw[black] (0,-1) circle (2 pt);
	\filldraw[black] (0, 2) circle (2 pt);
    \filldraw[black] (-1, 1) circle (2 pt);
	\filldraw[black] (1, 1) circle (2 pt);
	\filldraw[black] (0, 1) circle (2 pt);
	\filldraw[black] (-1, 0) circle (2 pt);
	\filldraw[black] (1, 0) circle (2 pt);
	\filldraw[black] (0, 0) circle (2 pt);
		\draw (0, 2.3) node {$(1, 1, 1)$};
	\draw (0, -1.3) node {$(0, 0, 0)$};
	\draw (-1.9, 0) node {{\small{$(1, 0, 0)$}}};
\draw (-1.9, 1) node {{\small{$(1, 1, 0)$}}};
    \draw (0.12, -0.2) node {{\small{$(0, 1, 0)$}}};
    \draw (0.12, 1.2) node {{\small{$(1, 0, 1)$}}};
    \draw (1.9, 0) node {{\small{$(0, 0, 1)$}}};
    \draw (1.9, 1) node {{\small{$(0, 1, 1)$}}};
     \draw (1.9, 1.5) node {{\small{Mary}}};
      \draw (-1.9, 1.5) node {{\small{John}}};

\end{tikzpicture}
\end{center}

In the above scenario, each committee member will individually judge candidates only on the basis of the parameters that they consider relevant. Thus,  Alan prefers John over Mary since John dominates Mary on all parameters relevant to Alan, while Betty is undecided between John and Mary, since neither dominates the other on all the parameters relevant to Betty. 
\begin{definition}
For every $Y\subseteq X$, let $e_Y\subseteq 2^{X}\times 2^X$ be the equivalence relation defined as follows:\footnote{Hence, $e_\varnothing = \tau = W\times W$.}
\[e_Y = \{(w, u)\mid w_y = u_y \text{ for any } y\in  Y\}.\]
\end{definition}
Intuitively, $e_Y$ can be understood to represent  the question (interrogative agenda)  `what are the given candidate's scores on the parameters in $Y$?', because any two candidates have the same answer to this question iff their profiles are in the same $e_Y$-equivalence class.

Abusing notation, we  write $e_y$ in place of $e_{\{y\}}$ for any $y\in X$. Hence,   $e_y$  is the equivalence relation representing the question `what is the given candidate's score on parameter $y$?'. 
As we score on a binary scale, any such question corresponds to a binary partition of $W$. Therefore, by Proposition \ref{interrogative:prop:charact meet-irr},  $e_y$ is a meet-irreducible element of the lattice $E(W)$. 
In what follows, for any $y\in Y$, we refer to  $e_y$  as the {\em issue}\footnote{\label{interrogative:ftn:issues and parameters}In the present situation, in which total domination is the winning rule, issues (i.e.~meet-generators of the algebra of interrogative agendas) can be identified with individual parameters; however, in the next section, a deliberation scenario with 
 a different winning rule is discussed in which issues cannot be identified with individual parameters.} defined by $y$, and, given that parameters are pivotal to the whole deliberation process, rather than working on the whole of $E(W)$, we consider the complete sub-meet-semilattice $\mathbb{D}$ of $E(W)$ generated by $\{e_x\mid x\in X\}$  as the  `algebra of interrogative agendas'.\footnote{\label{footn: D is BA} It is straightforward to see that the assignment $Y\mapsto e_Y$ induces an order-isomorphism between  the poset $(\mathcal{P}(X), \supseteq)$ and $\mathbb{D}$, which is  hence a Boolean algebra.}

By definition, 
$e_Y = \bigsqcap_{y\in Y}e_y$ for any $Y\subseteq X$,  
and $e_{Y}$ is the kernel of the  projection $\pi_Y: 2^X\to 2^Y$. Hence, in our example, $e_r$, $e_p$, and $e_l$ 
  are the kernels of the projections of the  profiles on the first, second, and third coordinates, respectively.
Then, Alan and Betty's agendas are   $e_\aga = e_p \rand e_r$, and  $e_\agb = e_l \rand e_r$, respectively, and are represented   in Figure \ref{interrogative:fig:agendas}\footnote{The way the story is written does not actually exclude the possibility that in fact Betty's interrogative agenda is  the one referred to as $e_\aga\rand e_\agb$ in Figure \ref{interrogative:fig:agendas}. However, representing Betty's interrogative agenda as we did is consistent with the subsequent analysis in \cite[Section 3.2]{Baltag2018}, where Betty does not score the philosophy parameter, which we translate by not including the philosophy issue in Betty's interrogative agenda. Notice, however, that even including the philosophy issue in Betty's interrogative agenda, our conclusions would stand: namely, Betty's initial agenda would not yield an unequivocal decision, and neither would taking the meet of Alan and Betty's interrogative agendas.}.

\begin{figure}[H]
\centering
\begin{tikzpicture}
\draw (0,0) circle (1cm);
\draw (0,1) -- (0,-1);
\draw (-0.7,-0.7) -- (0.7,0.7);

\draw (3,0) circle (1cm);
\draw (3,1) -- (3,-1);
\draw (2.3,0.7) -- (3.7,-0.7);

\draw (6,0) circle (1cm);
\draw (6,1) -- (6,-1);
\draw (5.3,-0.7) -- (6.7,0.7);
\draw (5.3,0.7) -- (6.7,-0.7);

\draw (9,0) circle (1cm);
\draw (9,1) -- (9,-1);

\draw (12,0) circle (1cm);
\draw (11.3,-0.7) -- (12.7,0.7);
\draw (11.3,0.7) -- (12.7,-0.7);
\draw (0, -1.25) node {$e_\aga$};
\draw (3, -1.25) node {$e_\agb$};
\draw (6, -1.25) node {$e_\aga\rand e_\agb$};
\draw (9, -1.25) node {$e_\aga\ror e_\agb$};
\draw (12, -1.25) node {$e'_\aga\rand e'_\agb$};
\end{tikzpicture}
\captionof{figure}{Representations of agendas}\label{interrogative:fig:agendas}
\end{figure}
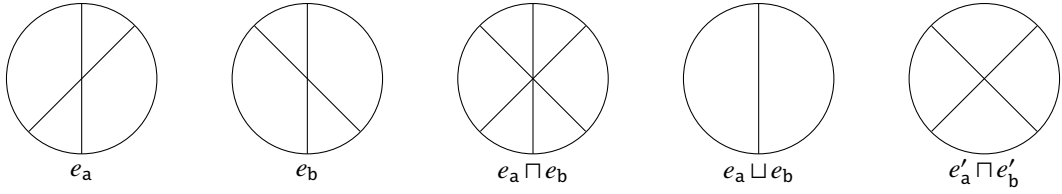

In the following picture, the  partitions of the poset $P$ of  profiles induced by Alan's and Betty's agendas  are represented (red for Alan and blue for Betty), as well as  their projections: 
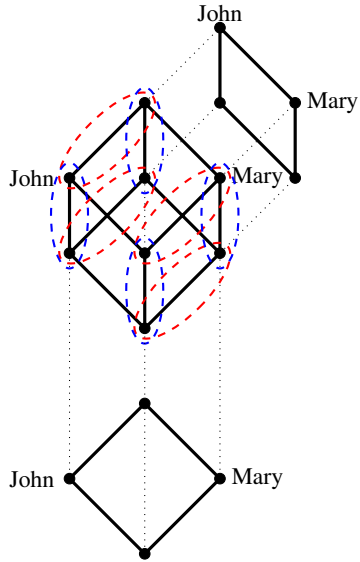
\begin{center}
\begin{tikzpicture}
\draw[very thick] (-1, 0) -- (-1, 1) --
	(0, 0) -- (1, 1) -- (1, 0) -- (0, 1) -- (-1, 0);
	\draw[very thick] (0, 2) -- (-1, 1);
\draw[very thick] (0, 2) -- (0, 1);
\draw[very thick] (0, 2) -- (1, 1);
	\draw[very thick] (0, -1) -- (-1, 0);
\draw[very thick] (0, -1) -- (0, 0);
\draw[very thick] (0, -1) -- (1, 0);
	\filldraw[black] (0,-1) circle (2 pt);
	\filldraw[black] (0, 2) circle (2 pt);
    \filldraw[black] (-1, 1) circle (2 pt);
	\filldraw[black] (1, 1) circle (2 pt);
	\filldraw[black] (0, 1) circle (2 pt);
	\filldraw[black] (-1, 0) circle (2 pt);
	\filldraw[black] (1, 0) circle (2 pt);
	\filldraw[black] (0, 0) circle (2 pt);
	
     \draw (1.5, 1) node {{\small{Mary}}};
      \draw (-1.5, 1) node {{\small{John}}};

              \draw[rotate around={ -45: (0.5, 0.5)},dashed, thick, red](0.5, 0.5) ellipse (9pt and 24pt);
              \draw[rotate around={ -45: (-0.5, 1.5)},dashed, thick, red](-0.5, 1.5) ellipse (9pt and 24pt);
              \draw[rotate around={ -45: (-0.5, 0.5)},dashed, thick, red](-0.5, 0.5) ellipse (9pt and 24pt);
              \draw[rotate around={ -45: (0.5, -0.5)},dashed, thick, red](0.5, -0.5) ellipse (9pt and 24pt);

               \draw[dashed, thick, blue](0, -0.5) ellipse (7pt and 20pt);
               \draw[dashed, thick, blue](0, 1.5) ellipse (7pt and 20pt);
               \draw[dashed, thick, blue](-1, 0.5) ellipse (7pt and 20pt);
               \draw[dashed, thick, blue](1, 0.5) ellipse (7pt and 20pt);

         \draw[very thick] (2, 1) -- (2, 2) -- (1, 3) -- (1, 2)-- (2,1);
         \filldraw[black] (2, 1) circle (2 pt);
         \filldraw[black] (2, 2) circle (2 pt);
         \filldraw[black] (1, 3) circle (2 pt);
         \filldraw[black] (1, 2) circle (2 pt);
         \draw (1, 3.2) node {{\small{John}}};
         \draw (2.5, 2) node {{\small{Mary}}};
         \draw[dotted] (0, 2) -- (1,3);
         \draw[dotted] (1, 1) -- (2,2);
         \draw[dotted] (1, 0) -- (2,1);
         \draw[dotted] (0, 1) -- (1,2);

         \draw[very thick] (0, -2) -- (1, -3) -- (0, -4) -- (-1, -3)-- (0,-2);
         \filldraw[black] (0, -2) circle (2 pt);
         \filldraw[black] (0, -4) circle (2 pt);
         \filldraw[black] (1, -3) circle (2 pt);
         \filldraw[black] (-1, -3) circle (2 pt);
         \draw[dotted] (0, 1) -- (0,-4);
         \draw[dotted] (-1, 0) -- (-1,-3);
         \draw[dotted] (1, 0) -- (1,-3);
         \draw (1.5, -3) node {{\small{Mary}}};
      \draw (-1.5, -3) node {{\small{John}}};
\end{tikzpicture}
\end{center}
Notice that in the quotient space obtained by taking the projection of $W$ according to Alan's interrogative agenda, the point corresponding to John is above the point corresponding to Mary in the quotient order, whereas in the quotient space arising from Betty's agenda, John and Mary remain incomparable. This fact encodes the information that indeed Alan's agenda yields a decision, whereas Betty's does not.

\begin{definition}
  \label{interrogative:def:quotient space and order}  

For any  agenda $e_Y$,
the  {\em  quotient space} of $P = (W, \leq)$ associated with $e_Y$ is  the poset $P_Y = (W_Y, \leq_Y $), where $W_Y$ is the projection of $W$ onto the coordinates corresponding to the parameters in $Y$, and for any $w, u \in W_Y$, $w \leq_Y u$ iff $u$ dominates $w$ on each parameter in $Y$.\footnote{\label{interrogative:ftn:leqY on W}In fact, an order $\leq_Y$ can be defined not only on $W_Y$ but also on $W$, by stipulating that, for all $u, w\in W$, $w\leq_Y u$ iff $w_y\leq u_y$ for all $y\in Y$. The two versions of $\leq_Y$ can be identified, since  $w\leq_Y u$ iff $\pi_Y(w)\leq_Y\pi_Y(u)$ for all $u, w\in W$.}  For all $w, u\in 2^X$ and any $Y\subseteq X$,  
we say that $e_Y$  {\em prefers $u$ to $w$} if $\pi_{Y}(w) \leq_Y \pi_{Y}(u)$ and $\pi_{Y}(w) \neq \pi_{Y}(u)$. More in general, we say that $e_Y$ {\em leads to a decision} between $u$ and $w$ whenever either $\pi_{Y}(w) \leq_Y \pi_{Y}(u)$ or $\pi_{Y}(u) \leq_Y \pi_{Y}(w)$.
\end{definition}
The following lemma shows that Definition \ref{interrogative:def:e prefers} simplifies to the definition given above when $e: = e_Y$, and describes $\leq_Y$ in terms of the relation $e_Y$ and the order relation of the poset $P$. 

\begin{lemma}
\label{interrogative:lemma:preference}
For all $w, u\in W$ and any $Y\subseteq Y'\subseteq X$,  
\begin{enumerate}
\item $e_{\leq_Y} = e_Y$ and ${\leq_{e_Y}}  ={\leq_Y}$;
    \item $e_{Y}\circ {\leq_{Y'}}\circ e_Y = {\leq_{Y}}$. Hence, for $Y': = X$,  $ {\leq_Y} =  e_Y \, \circ  {\leq}  \circ\, e_Y$.
\item $e_{Y'}$ is  semi-compatible with $\leq_Y$.
\end{enumerate}
\end{lemma}
\begin{proof}
1. By definition,  $(w, u)\in e_{\leq_Y}$ iff $w\leq_Y u$ and $u\leq_Y w$, iff $w_y = u_y$ for every $y\in Y$, iff $(w, u)\in e_{Y}$.  Hence,  $e_{\leq_Y} = e_Y$, which implies, by Lemma \ref{interrogative:lem:preorder-equivalence back-forth}.1, ${\leq_Y} = {\leq_{e_{\leq_Y}} }= {\leq_{e_Y}}$.

2. Let $ w, w', u, u'\in W$ s.t.~$(w, w')\in e_{Y}$, $w'\leq_{Y'} u'$ and $(u, u')\in e_{Y}$.  Then $w_y = w'_y$ and $u'_y = u_y$ for every $y\in Y$, and $w'_y\leq u'_y$ for every $y\in Y'$. Hence, $w_y\leq u_y$ for every $y\in Y$, which shows $e_{Y}\circ {\leq_{Y'}}\circ e_{Y}\subseteq {\leq_Y}$.

Conversely, if $w \leq_Y u$, then let $w'$ and $u'$ be the profiles such that $w'_y=w_y$ and $u'_y=u_y$ for all $y \in Y$, and $w'_x=0 = u'_x$ for $x \in X\setminus Y$. By construction,  $(w, w')\in e_Y $, $w'\leq u'$ and $(u', u)\in  e_Y$, 
witnessing   
$w (e_Y\, \circ \leq \circ\, e_Y) u$, as required.


3. Let $ w, w', u, u'\in W$ s.t.~$(w, w')\in e_{Y'}$, $w'\leq_Y u'$ and $(u, u')\in e_{Y'}$.  Then $w_y = w'_y$ and $u'_y = u_y$ for every $y\in Y'$, and $w'_y\leq u'_y$ for every $y\in Y$. Hence, $w_y\leq u_y$ for every $y\in Y$, which shows $e_{Y'}\circ {\leq_Y}\circ e_{Y'}\subseteq {\leq_Y}$. 
\end{proof}

\begin{remark}
   The  identities of the lemma above imply that corresponding identities  $\langle\leq_Y\rangle U= \langle e_Y \rangle \langle  \leq_{Y'}  \rangle \langle e_Y \rangle U $ and $\langle e_{Y'} \rangle \langle  \leq_Y  \rangle \langle e_{Y'} \rangle U \subseteq \langle\leq_{Y}\rangle U$ hold as well for every $U\subseteq W$ and all $Y\subseteq Y'\subseteq X$, where, for any  $R\subseteq W\times W$, the symbol $\langle R \rangle$ denotes the standard diamond operator associated  with  $R$ (cf.~Section \ref{interrogative:ssec:modal logic}).  These observations will play out in future work, when further expanding the present framework along the lines of the standard modal logic approach outlined in Section \ref{interrogative:ssec:modal logic}. 
\end{remark}


Notice that $e_{\aga}\rand e_{\agb}$, i.e.~the  interrogative agenda which includes all the parameters considered relevant by at least one agent, does  {\em not} yield  a decision,  given that each candidate does better than the other on one of the relevant parameters. 
In general, 
\begin{remark}
\label{interrogative:rem:distributed unanimity}
For any $\mathcal{X}\subseteq \mathcal{P}(X)$, the agenda $\bigsqcap\{e_Y\mid Y\in \mathcal{X}\} = e_{\bigcup\mathcal{X}}$ leads to a decision 
iff every $e_Y$, for each $Y\in\mathcal{X}$,   leads to the same decision. 
Thus, if the winning rule is total dominance, the attempt at generating the coalition agenda by including all the parameters considered relevant by some  agent will not lead to a decision 
in any non-trivial deliberation scenario. 

\end{remark}

However, each of the two remaining agendas in Figure \ref{interrogative:fig:agendas}  leads to a decision. Agenda $e_{\aga}\ror e_{\agb}$ leads to choosing John over Mary, and $e'_{\aga}\rand e'_{\agb}$ leads to choosing Mary over John. While $e_{\aga}\ror e_{\agb}$ can be accounted for simply by  taking the join of $e_{\aga}$ and $e_{\agb}$,  reaching $e'_{\aga}\rand e'_{\agb}$ from the initial positions requires some more sophistication than just taking the join or the meet of the original agendas.  In what follows, we describe a process that  leads  to a span of agendas which includes  such an agenda.

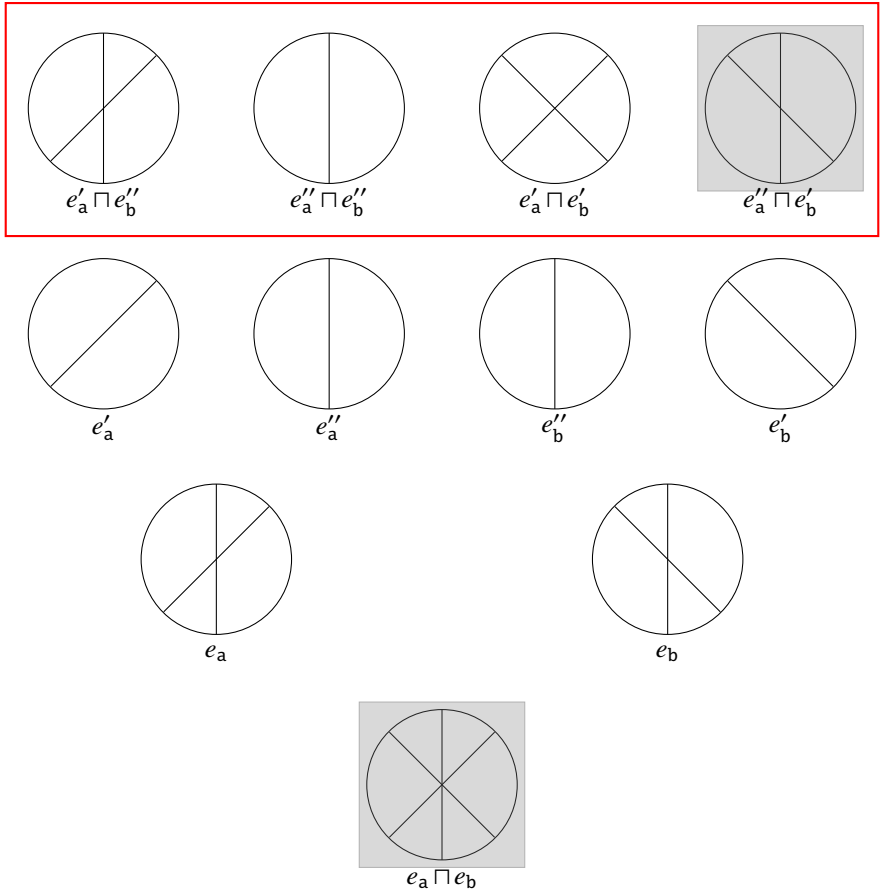
\begin{figure}[h]
    \centering
\begin{tikzpicture}
\draw[thick, red] (-1.3,-1.7) rectangle (10.3,1.4);
\draw (0,0) circle (1cm);
\draw (0,1) -- (0,-1);
\draw (-0.7,-0.7) -- (0.7,0.7);
\draw (0, -1.25) node {$e'_\aga\rand e''_\agb$};

\draw (3,0) circle (1cm);
\draw (3,1) -- (3,-1);
\draw (3, -1.25) node {$e''_\aga\rand e''_\agb$};

\draw (6,0) circle (1cm);
\draw (5.3,-0.7) -- (6.7,0.7);
\draw (5.3,0.7) -- (6.7,-0.7);
\draw (6, -1.25) node {$e'_\aga\rand e'_\agb$};

\draw (9,0) circle (1cm);
\draw (9,1) -- (9,-1);
\draw (9.7,-0.7) -- (8.3,0.7);
\draw[fill=gray, opacity=0.3] (7.9,-1.1) rectangle (10.1,1.1);
\draw (9, -1.25) node {$e''_\aga\rand e'_\agb$};


\draw (0,-3) circle (1cm);
\draw (-0.7,-3.7) -- (0.7,-2.3);
\draw (0, -4.25) node {$e'_\aga$};

\draw (3,-3) circle (1cm);
\draw (3,-2) -- (3,-4);
\draw (3, -4.25) node {$e''_\aga$};

\draw (6,-3) circle (1cm);
\draw (6,-2) -- (6,-4);
\draw (6, -4.25) node {$e''_\agb$};

\draw (9,-3) circle (1cm);
\draw (9.7,-3.7) -- (8.3,-2.3);
\draw (9, -4.25) node {$e'_\agb$};


\draw (1.5,-6) circle (1cm);
\draw (1.5,-7) -- (1.5,-5);
\draw (0.8,-6.7) -- (2.2,-5.3);
\draw (1.5, -7.25) node {$e_\aga$};

\draw (7.5,-6) circle (1cm);
\draw (7.5,-7) -- (7.5,-5);
\draw (8.2,-6.7) -- (6.8,-5.3);
\draw (7.5, -7.25) node {$e_\agb$};


\draw (4.5,-9) circle (1cm);
\draw (4.5,-8) -- (4.5,-10);
\draw (3.8,-8.3) -- (5.2,-9.7);
\draw (5.2,-8.3) -- (3.8,-9.7);
\draw[fill=gray, opacity=0.3] (3.4,-10.1) rectangle (5.6,-7.9);
\draw (4.5, -10.25) node {$e_\aga\rand e_\agb$};
\end{tikzpicture}
\caption{Possible coalition agendas in $C$}
    \label{interrogative:fig:coarsening}
\end{figure}

For any set of parameters $Y$, we define the set $\mathsf{Crs}_1(e_Y)$ (the {\em coarsenings} of $e_Y$)  to be the set of agendas of the form $e_Z$, where $Z$ is obtained by removing one parameter from $Y$.  In this case study,  $\mathsf{Crs}_1(e_\aga) = \{e'_\aga, e''_\aga\}$ and  $\mathsf{Crs}_1(e_\agb) = \{e'_\agb, e''_\agb\}$.
Then we consider the set $C$ of all the agendas of the form $\bigsqcap_{\mathtt{i}\neq \mathtt{j}}e'_{\mathtt{i}\mathtt{j}}$ where $e'_{\mathtt{i}\mathtt{j}}\in \mathsf{Crs}_1(e_{\mathtt{j}})$ for each $\mathtt{j}$ and $\mathtt{i}$ such that $\mathtt{j}\neq \mathtt{i}$. Intuitively,  this is understood as the process in which each agent $\mathtt{i}$ chooses an element in  $\mathsf{Crs}_1(e_\texttt{j})$ 
for every $\mathtt{j}\neq \mathtt{i}$, thereby vetoing one issue from the agenda of each other agent (as to the case study, the elements of $C$ are those in the uppermost row of the picture above), and then considering all the parameters in the resulting agendas. The  deliberation will then restrict itself to one of the elements of $C$, provided some yield a decision; if none of them  does, the process is repeated  starting from  the coarsenings in which exactly  two parameters are removed, and so on.

All agendas in $C$ for this case study are illustrated in Figure \ref{interrogative:fig:coarsening}. All these agendas except $e''_\aga\rand e'_\agb$ lead to a decision. The agenda $e'_\aga\rand e''_\agb$ coincides with Alan's initial agenda, and reflects a situation in which Betty accepts Alan's input while Alan rejects Betty's input; the agenda $e''_\aga\rand e''_\agb$ coincides with  $e_\aga\ror e_\agb$, and  reflects a situation in which each agent rejects the other agent's input, and each of them chooses the coarsening of the other agent's agenda that is most similar to one's own; the agenda $e'_\aga\rand e'_\agb$   reflects a situation in which each agent accepts the other agent's input, and each of them chooses the coarsening of the other agent's agenda that is most different to one's own, i.e.~those issues which he/she would {\em not} independently consider relevant. 

\begin{figure}
\begin{center}
\begin{tikzpicture}[scale=0.6]

\draw (3,3) circle (1cm);

\draw[very thick] (3,1.2) -- (3,1.8);
\draw[very thick] (2.1,2.1) -- (0.9,0.9);
\draw[very thick] (3.9,2.1) -- (5.1,0.9);

\draw (0,0) circle (1cm);
\draw (-0.7,-0.7) -- (0.7,0.7);

\draw (3,0) circle (1cm);
\draw (3,1) -- (3,-1);

\draw[thick, red] (1.9,-1.1) rectangle (4.1,1.1);

\draw (6,0) circle (1cm);
\draw (5.3,0.7) -- (6.7,-0.7);

\draw (10.5, 0.6) node {{\small{References}}};
\draw (8, 0.1) node {{\small{John}}};
\draw (8, -0.6) node {{\small{Mary}}};
\draw (10.5, 0.1) node {{\small{1}}};
\draw (10.5, -0.6) node {{\small{0}}};
\draw (7.3,1) -- (12.1,1);
\draw (7.3,1) -- (7.3,-0.9);
\draw (7.3,-0.9) -- (12.1,-0.9);
\draw (9,1) -- (9,-0.9);
\draw (12.1,1) -- (12.1,-0.9);
\draw (7.3,-0.25) -- (12.1,-0.25);
\draw (7.3,0.35) -- (12.1,0.35);
\filldraw [red, opacity = 0.2] (7.3, -0.25) rectangle (12.1, 0.35);

\draw[very thick] (0,-1.8) -- (0,-1.2);
\draw[very thick] (6,-1.8) -- (6,-1.2);

\draw[very thick] (2.1,-2.1) -- (0.9,-0.9);
\draw[very thick] (3.9,-2.1) -- (5.1,-0.9);
\draw[very thick] (0.9,-2.1) -- (2.1,-0.9);
\draw[very thick] (5.1,-2.1) -- (3.9,-0.9);

\draw (0,-3) circle (1cm);
\draw (0,-2) -- (0,-4);
\draw (-0.7,-3.7) -- (0.7,-2.3);
\draw (0, -4.5) node {$e_\aga$};

\draw (3,-3) circle (1cm);
\draw (2.3,-3.7) -- (3.7,-2.3);
\draw (3.7,-3.7) -- (2.3,-2.3);

\draw[fill=yellow, opacity=0.3] (1.9,-4.1) rectangle (4.1,-1.9);
\draw[very thick, yellow] (1.9,-4.1) rectangle (4.1,-1.9);

\draw (6,-3) circle (1cm);
\draw (6,-2) -- (6,-4);
\draw (6.7,-3.7) -- (5.3,-2.3);
\draw (6, -4.5) node {$e_\agb$};

\draw[very thick] (2.1,-5.1) -- (0.9,-3.9);
\draw[very thick] (3.9,-5.1) -- (5.1,-3.9);
\draw[very thick] (3,-4.8) -- (3,-4.2);

\filldraw [yellow] (7.3, -3.9) rectangle (13.7, -3.25);
\draw (10.5, -2.4) node {{\small{Philosophy}}};
\draw (13, -2.4) node {{\small{Logic}}};
\draw (8, -2.9) node {{\small{John}}};
\draw (8, -3.6) node {{\small{Mary}}};
\draw (10.5, -2.9) node {{\small{1}}};
\draw (10.5, -3.6) node {{\small{1}}};
\draw (13, -2.9) node {{\small{0}}};
\draw (13, -3.6) node {{\small{1}}};
\draw (7.3,-2) -- (13.7,-2);
\draw (7.3,-2) -- (7.3,-3.9);
\draw (7.3,-3.9) -- (13.7,-3.9);
\draw (13.7,-2) -- (13.7,-3.9);
\draw (9,-2) -- (9,-3.9);
\draw (12.1,-2) -- (12.1,-3.9);
\draw (7.3,-3.25) -- (13.7,-3.25);
\draw (7.3,-2.6) -- (13.7,-2.6);
\filldraw [yellow, opacity = 0.2] (7.3, -3.9) rectangle (13.7, -3.25);


\draw (3,-6) circle (1cm);
\draw (3,-5) -- (3,-7);
\draw (2.3,-5.3) -- (3.7,-6.7);
\draw (3.7,-5.3) -- (2.3,-6.7);

\draw[fill=gray, opacity=0.3] (1.9,-7.1) rectangle (4.1,-4.9);

\draw (7.5, -5.3) node {{\small{References}}};
\draw (10.5, -5.4) node {{\small{Philosophy}}};
\draw (13, -5.4) node {{\small{Logic}}};
\draw (5.5, -5.9) node {{\small{John}}};
\draw (5.5, -6.6) node {{\small{Mary}}};
\draw (7.5, -5.9) node {{\small{1}}};
\draw (7.5, -6.6) node {{\small{0}}};
\draw (10.5, -5.9) node {{\small{1}}};
\draw (10.5, -6.6) node {{\small{1}}};
\draw (13, -5.9) node {{\small{0}}};
\draw (13, -6.6) node {{\small{1}}};
\draw (4.9,-5) -- (13.7,-5);
\draw (4.9,-5) -- (4.9,-6.9);
\draw (4.9,-6.9) -- (13.7,-6.9);
\draw (13.7,-5) -- (13.7,-6.9);
\draw (4.9,-5.6) -- (13.7,-5.6);
\draw (4.9,-6.25) -- (13.7,-6.25);
\draw (6.2,-5) -- (6.2,-6.9);
\draw (9,-5) -- (9,-6.9);
\draw (12.1,-5) -- (12.1,-6.9);
\draw[fill=gray, opacity=0.3] (6.2,-6.9) rectangle (9,-5.6);
\draw[fill=gray, opacity=0.3] (12.1,-6.9) rectangle (13.7, -5.6);
\end{tikzpicture}
\end{center}
\caption{The algebra $\mathbb{D}$ of interrogative agendas, and their associated decisions (or lack thereof).}
\end{figure}
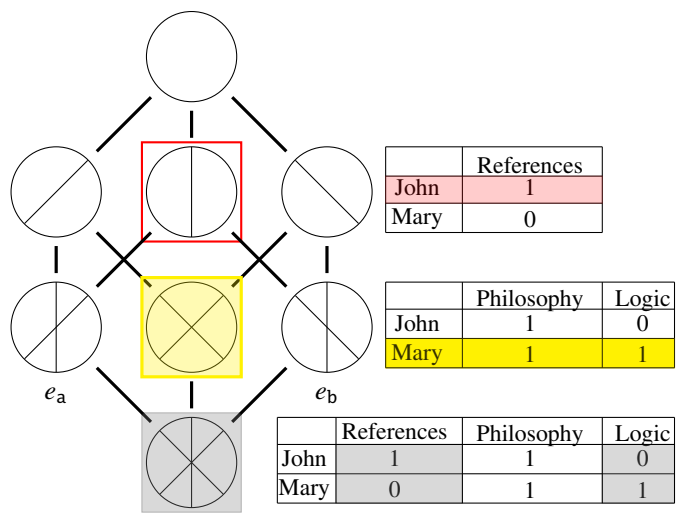


Summing up, starting with $n = |X|$ parameters, each scored in $\{0, 1\}$, and assuming the total dominance winning rule, the poset $P = (W, \leq)$ based on the feature space $W$ is the $n$-dimensional cube, i.e.~the Boolean algebra with $n$ atoms, and the interrogative agendas $e_Y$ 
defined by choosing a given  subset $Y\subseteq X$ of  parameters as relevant correspond  bijectively to kernels of projections over that subset of parameters (the identity relation and the total relation are the kernels of the projections over all parameters and over the empty set of  parameters, respectively). Hence,  the algebra $\mathbb{D}$ of such interrogative agendas is a Boolean algebra, and, being meet-generated by  $\{e_y\mid y\in X\}$, it
is isomorphic to $P$ with the reverse order.  

\paragraph{Total dominance  and many-valued linearly ordered scores.}
In the discussion above, we assumed that  parameters are scored on a binary scale. However, in many situations, it is desirable to score parameters  over a many-valued scale. For example, we may want to evaluate candidates based on how strong their reference letters are, or what their h-index is, not just whether they have one or not. Accordingly, the previous setting can be generalized to one in which every parameter $x\in X$ is associated with a linearly ordered range of values $\mathbb{L}_x = (L_x, \leq_x)$ which has a maximal element $1_x$ and a minimal element $0_x$. Note that each $\mathbb{L}_x$ is a bounded distributive lattice. Consequently, the set of  profiles\footnote{In Figure \ref{fig:linear mv}, we assume that John and Many are scored according to the following table:

\begin{tabular}{|r|c|c|c|}
\hline
& References & Philosophy& Logic\\
\hline
John & 1 &1& 0\\
\hline
Mary &0&1/2&1\\
\hline
\end{tabular}
} becomes $W := \prod_{x\in X}\mathbb{L}_x$, and $P = (W, \leq)$ is a bounded distributive lattice. Notice also that the statement and proof of Lemma \ref{interrogative:lemma:preference} hold verbatim also if $W := \prod_{x\in X}\mathbb{L}_x$; hence, assuming again that total dominance be the winning rule,  the preference order associated with each interrogative agenda $e_Y$ simplifies to the  order $\leq_Y$ on the quotient space (cf.~Definition \ref{interrogative:def:quotient space and order}).

\usetikzlibrary{3d, positioning}

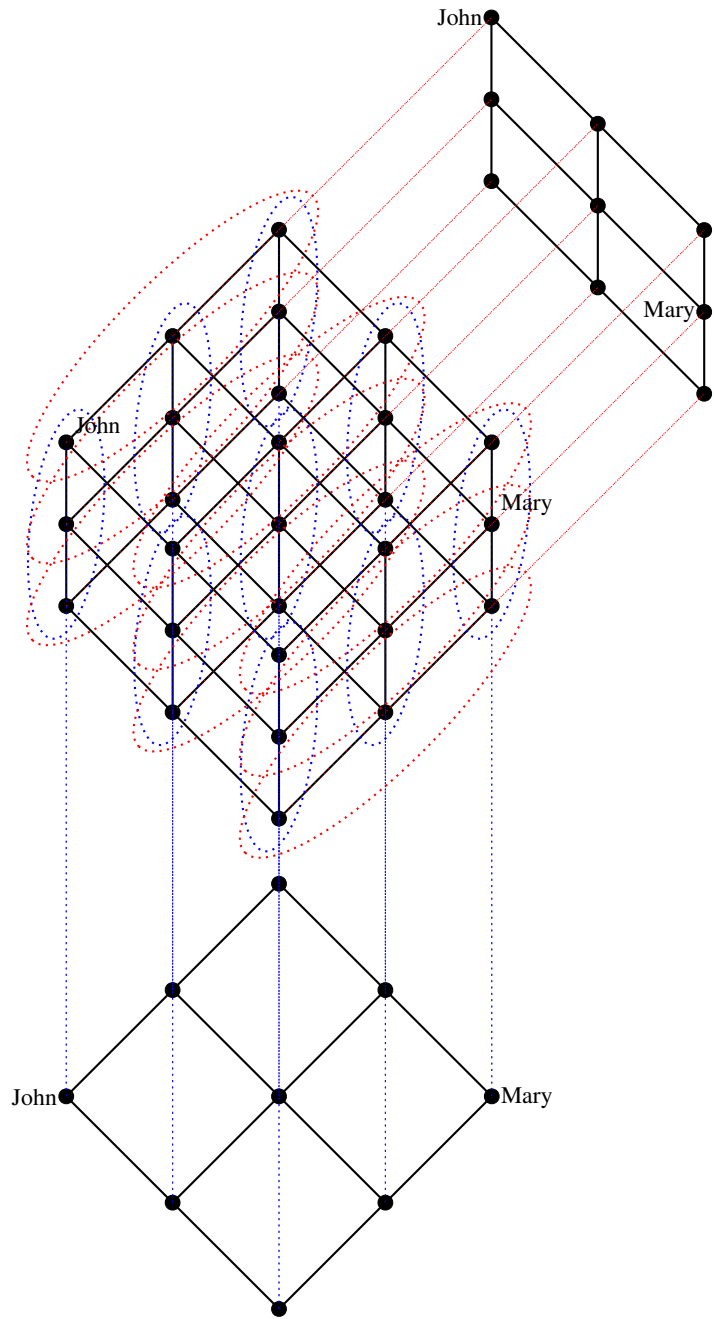
\begin{figure}
\begin{center}
    \begin{tikzpicture}[scale=4]
  \begin{scope} [rotate =225]

\def\step{0.5}

\foreach \x in {0,1,2}
  \foreach \y in {0,1,2}
    \foreach \z in {0,1,2} {
      \filldraw[black] (\x*\step,\y*\step,\z*\step) circle (0.7pt);
    }

\foreach \y in {0,1,2}
  \foreach \z in {0,1,2} {
    \draw[black, thick] (0,\y*\step,\z*\step) -- (2*\step,\y*\step,\z*\step);
  }
\foreach \x in {0,1,2}
  \foreach \z in {0,1,2} {
    \draw[black, thick] (\x*\step,0,\z*\step) -- (\x*\step,2*\step,\z*\step);
  }
\foreach \x in {0,1,2}
  \foreach \y in {0,1,2} {
    \draw[black, thick] (\x*\step,\y*\step,0) -- (\x*\step,\y*\step,2*\step);
  }

\coordinate (XYorigin) at (0,0,-3);

\foreach \x in {0,1,2} {
  \draw[black, thick] ($ (XYorigin) + (\x*\step,0,0) $) -- ++(0,2*\step,0);
}
\foreach \y in {0,1,2} {
  \draw[black, thick] ($ (XYorigin) + (0,\y*\step,0) $) -- ++(2*\step,0,0);
}

\foreach \x in {0,1,2}
  \foreach \y in {0,1,2} {
    \filldraw[black] ($ (XYorigin) + (\x*\step,\y*\step,0) $) circle (0.7pt);
  }

\foreach \x in {0,1,2}
  \foreach \y in {0,1,2}
    \foreach \z in {0,1,2} {
      \draw[blue, dotted, thin] (\x*\step,\y*\step,\z*\step) -- ($ (XYorigin) + (\x*\step,\y*\step,0) $);
    }

\coordinate (YZorigin) at (-1.0,0,0);

\foreach \y in {0,1,2} {
  \draw[black, thick] ($ (YZorigin) + (0,\y*\step,0) $) -- ++(0,0,2*\step);
}
\foreach \z in {0,1,2} {
  \draw[black, thick] ($ (YZorigin) + (0,0,\z*\step) $) -- ++(0,2*\step,0);
}

\foreach \y in {0,1,2}
  \foreach \z in {0,1,2} {
    \filldraw[black] ($ (YZorigin) + (0,\y*\step,\z*\step) $) circle (0.7pt);
  }

\foreach \x in {0,1,2}
  \foreach \y in {0,1,2}
    \foreach \z in {0,1,2} {
      \draw[red, dotted, thin] (\x*\step,\y*\step,\z*\step) -- ($ (YZorigin) + (0,\y*\step,\z*\step) $);
    }

\foreach \x in {0,1,2} {
  \foreach \y in {0,1,2} {
    \pgfmathsetmacro{\cx}{\x*\step}
    \pgfmathsetmacro{\cy}{\y*\step}
    \pgfmathsetmacro{\cz}{\step} 
    \draw[blue, dotted, thick, rotate around y=0] 
      plot[variable=\t,domain=0:360,smooth,samples=50]
      ({\cx + 0.18*cos(\t)}, {\cy}, {\cz + 0.66*sin(\t)});
  }
}

\foreach \y in {0,1,2} {
  \foreach \z in {0,1,2} {
    \pgfmathsetmacro{\cx}{\step} 
    \pgfmathsetmacro{\cy}{\y*\step}
    \pgfmathsetmacro{\cz}{\z*\step}
    \draw[red, dotted, thick]
      plot[variable=\t,domain=0:360,smooth,samples=50]
      ({\cx + 0.66*cos(\t)}, {\cy + 0.18*sin(\t)}, {\cz});
  }
}

\node[above right, font=\small] at (0,1,0.5) {Mary};
\node[above right, font=\small] at (1,0, 1) {John};

\node[left, font=\small] at (1, 0,-3) {John};
\node[right,  font=\small] at (0,1,-3) {Mary};

\node[left, font=\small] at (-1,1,0.5) {Mary};
\node[left, font=\small] at (-1.0,0,1) {John};

\end{scope}
\end{tikzpicture}
\end{center}
\caption{Profile space $P = (W\leq)$ when  $r, l, p$ are evaluated on a three-element chain, and its projections induced by Alan's and Betty's agendas.}
\label{fig:linear mv}
\end{figure}

%
 Notice that, for any $y\in X$,  the issue $e_y$  gives rise to a partition of $W$ into as many equivalence classes as there are elements in $L_x$, and hence, unlike the case of binary scores, $e_y$ is not necessarily a meet-irreducible element of $E(W)$. However, as in the previous case, we still have $e_Y=\bigsqcap_{y \in Y} e_y$ for every $Y\subseteq X$.  Thus,  the sublattice of $E(W)$ meet-generated by $\{e_x\mid x\in X\}$ is still isomorphic to the Boolean algebra 
 $(\mathcal{P}(X), \supseteq)$.


\paragraph{Total dominance and non-linearly ordered  scores.}
In a wide range of situations, it may be desirable that (some) parameters be  scored  on a partial order which is not linearly ordered (see, e.g.~\cite{nolt2022incomparable} for a book-length treatment of this issue). For example, continuing with the deliberation scenario discussed above, while we may want to rank reference letters 
 e.g.~on a $4$-valued scale $0$, $\frac{1}{3}$, $\frac{2}{3}$ and $1$, we might also want to accommodate the possibility that  some candidates miss the letters of reference. Hence,  
we can include a fifth value $u$ (for `\emph{unknown}') that we place strictly in between $0$ and $1$, since not having letters of reference is certainly worse  than having supportive letters and  better than having unambiguously negative ones, and   
yet, without any bases for placing in precisely between these two extreme values,  is also incomparable with all the other intermediate values.  
In this case, the set $W: = \Pi_{x\in X}\mathbb{L}_x$ of profiles becomes a general (non-distributive) lattice, as illustrated in Figure \ref{fig:non-linear W}, which represents the situation in which $p$ and $l$ are scored on the three-element chain, and $r$ is scored on the lattice $N_5$ described above.  The statement and proof of Lemma \ref{interrogative:lemma:preference} hold verbatim also under these assumptions; hence, with total dominance as the winning rule,  the preference order associated with each interrogative agenda $e_Y$ again simplifies to the  order $\leq_Y$ on the quotient space (cf.~Definition \ref{interrogative:def:quotient space and order}).
%
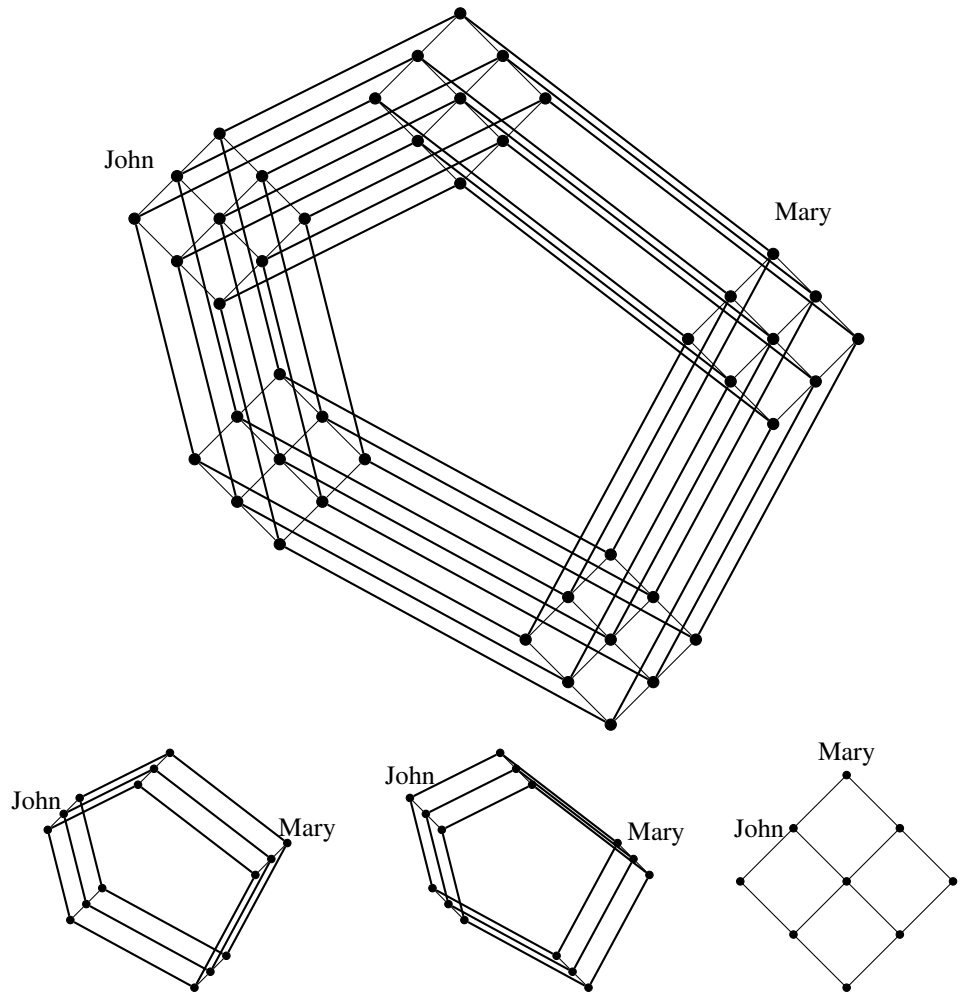
\begin{figure}[H]
\begin{center}
\begin{tikzpicture}[scale=0.8]

\newcommand{\drawRotatedGrid}[3]{%
  \begin{scope}[shift={(#1)}, rotate=45]
    \foreach \i in {0,1,2} {
      \foreach \j in {0,1,2} {
        \pgfmathsetmacro{\x}{\i}
        \pgfmathsetmacro{\y}{\j}
        \node[circle, draw, fill=black, inner sep=1.5pt] (p#2\i\j) at (\x, \y) {};
      }
    }
    \foreach \i in {0,1} {
      \foreach \j in {0,1,2} {
        \pgfmathtruncatemacro{\ip}{\i+1}
        \draw (p#2\i\j) -- (p#2\ip\j);
      }
    }
    \foreach \i in {0,1,2} {
      \foreach \j in {0,1} {
        \pgfmathtruncatemacro{\jp}{\j+1}
        \draw (p#2\i\j) -- (p#2\i\jp);
      }
    }
  \end{scope}
}

\drawRotatedGrid{(0,0)}{0}{}
\drawRotatedGrid{(-5.5,3)}{a}{}
\drawRotatedGrid{(-6.5,7)}{b}{}
\drawRotatedGrid{(2.7,5)}{u}{}
\drawRotatedGrid{(-2.5,9)}{1}{}


\node at (-8,9.4) {John};
\node at (3.2, 8.5) {Mary};

\newcommand{\connectMatchingPoints}[2]{%
  \foreach \i in {0,1,2} {
    \foreach \j in {0,1,2} {
      \draw[thick] (p#1\i\j) -- (p#2\i\j);
    }
  }
}

\connectMatchingPoints{0}{a}
\connectMatchingPoints{0}{u}
\connectMatchingPoints{a}{b}
\connectMatchingPoints{b}{1}
\connectMatchingPoints{u}{1}

\end{tikzpicture}

\begin{tabular}{ccc}
 \begin{tikzpicture}[scale=0.3]

\newcommand{\drawRotatedGrid}[3]{%
  \begin{scope}[shift={(#1)}, rotate=45]
    \foreach \i in {0,1,2} {
      \foreach \j in {0} {
        \pgfmathsetmacro{\x}{\i}
        \pgfmathsetmacro{\y}{\j}
        \node[circle, draw, fill=black, inner sep=1pt] (p#2\i\j) at (\x, \y) {};
      }
    }
    \foreach \i in {0,1} {
      \foreach \j in {0} {
        \pgfmathtruncatemacro{\ip}{\i+1}
        \draw (p#2\i\j) -- (p#2\ip\j);
      }
    }
    \foreach \i in {0,1,2} {
      \foreach \j in {0} {
        \pgfmathtruncatemacro{\jp}{\j+1}
      }
    }
  \end{scope}
}

\drawRotatedGrid{(0,0)}{0}{}
\drawRotatedGrid{(-5.5,3)}{a}{}
\drawRotatedGrid{(-6.5,7)}{b}{}
\drawRotatedGrid{(2.7,5)}{u}{}
\drawRotatedGrid{(-2.5,9)}{1}{}


\node at (-7,8.3) {John};
\node at (5, 7) {Mary};

\newcommand{\connectMatchingPoints}[2]{%
  \foreach \i in {0,1,2} {
    \foreach \j in {0} {
      \draw[thick] (p#1\i\j) -- (p#2\i\j);
    }
  }
}

\connectMatchingPoints{0}{a}
\connectMatchingPoints{0}{u}
\connectMatchingPoints{a}{b}
\connectMatchingPoints{b}{1}
\connectMatchingPoints{u}{1}

\end{tikzpicture}    &  

\begin{tikzpicture}[scale=0.3]
  \newcommand{\drawAngledGrid}[2]{%
    \begin{scope}[shift={(#1)}, rotate=135]
      \foreach \i in {0,1,2} {
        \node[circle, draw, fill=black, inner sep=1pt] (p#2\i) at (\i,0) {};
      }
      \foreach \i [evaluate=\i as \ip using int(\i+1)] in {0,1} {
        \draw (p#2\i) -- (p#2\ip);
      }
    \end{scope}
  }

  \drawAngledGrid{(0,0)}{0}
  \drawAngledGrid{(-5.5,3)}{a}
  \drawAngledGrid{(-6.5,7)}{b}
  \drawAngledGrid{(2.7,5)}{u}
  \drawAngledGrid{(-2.5,9)}{1}


  \node at (-7.9,9.4)        {John};
  \node at ( 3,6.8)        {Mary};

  \newcommand{\connectLayers}[2]{%
    \foreach \i in {0,1,2} {
      \draw[thick] (p#1\i) -- (p#2\i);
    }
  }

  \connectLayers{0}{a}
  \connectLayers{0}{u}
  \connectLayers{a}{b}
  \connectLayers{b}{1}
  \connectLayers{u}{1}
\end{tikzpicture}
     & 
     \begin{tikzpicture}[scale=2]
\begin{scope}[rotate=45]
    
\foreach \x in {0,0.5,1} {
  \foreach \y in {0,0.5,1} {
    \node[circle, draw, fill=black, inner sep=1pt] (p\x\y) at (\x,\y) {};
  }
}

\foreach \y in {0,0.5,1} {
  \draw (0,\y) -- (1,\y);
}

\foreach \x in {0,0.5,1} {
  \draw (\x,0) -- (\x,1);
}

\node[left] at (0.5,1) {John};
\node[above] at (1,1) {Mary};
\end{scope}

\end{tikzpicture}
\end{tabular}
\end{center}

\caption{Profile space $P = (W, \leq)$ when the set of values on which  $r$ is scored is shaped as the lattice $N_5$ and $p$ and $l$ are scored on the three-element chain, and its  projections associated with $e_r\rand e_l$, $e_r\rand e_p$ and $e_p\rand e_l$. Only the projections of $P$ associated with $e_p\rand e_l$ (and $e_l$) lead to a decision (which is the same decision, to prefer Mary over John). }\label{fig:non-linear W}
\end{figure}

%
If the candidates are scored according to the following table (see Figure \ref{fig:non-linear W}):

\begin{center}
\begin{tabular}{|r|c|c|c|}
\hline
& References & Philosophy& Logic\\
\hline
John & 2/3 &1& 1/2\\
\hline
Mary &$u$&1&1\\
\hline
\end{tabular}
\end{center}


\noindent as in the previous cases, $e_{\mathtt{a}} = e_r\rand e_p$ (resp.~$e_{\mathtt{b}} = e_r\rand e_l$) is the kernel of the projection onto the first and second  (resp.~first and  third) coordinate, while   $e_\aga \ror e_\agb = e_r$ is the kernel of the projection onto the first coordinate. As the values $u$ and $2/3$ are incomparable, John and Mary are incomparable according to each of these agendas, and more in general, for any $v, w\in W$ s.t.~$\pi_x(v)$ and $\pi_x(w)$ are incomparable for some $x\in X$, no projection $e_Y$ s.t.~$x\in Y$ can lead to a decision under the total domination winning rule. Thus, unlike the previous cases, neither $e_{\mathtt{a}}$ nor  $e_\aga \ror e_\agb$ lead to a decision. 

This points towards an important difference between linear and non-linear scales. As  $e_x$  is associated with the kernel of the projection onto the $x$-coordinate for any $x\in X$, if every  $x$ is scored on a linear scale, all such projections  are totally ordered. Thus, for any two distinct profiles,  a feature $x$ exists such that $e_x$
leads to a decision between them. However, the example above shows that this does not need to be the case under the present assumptions; i.e., 
it is possible that for no $x\in X$ issue $e_x$ leads to a decision. However, as in the previous cases, we still have $e_Y =\bigsqcap_{y \in Y}e_y$. Therefore, the algebra of interrogative agendas of the form $e_Y$ is still a Boolean algebra isomorphic to the Boolean algebra $\mathcal{P}(X), \supseteq)$. 

\section{A different winning rule: Choosing a car}\label{interrogative:sec:car}

The winning rule adopted in the deliberation scenario of the previous section 
was the `total dominance' rule, which could be represented as the product order on the set of profiles. However, in many real-life situations, different winning rules are considered, and in particular the one stipulating that, for any two profiles $u$ and $w$,   $u$ wins over   $w$ when {\em the sum} of the scores of $u$ in the relevant parameters is greater than the corresponding sum of the scores of $w$. In the present section, we discuss another deliberation  scenario in which this rule, which we refer to as the {\em sum rule}, is adopted.

Alan and Betty are a couple in the process of buying a second-hand car. Both are interested in fuel consumption; Alan is also interested in safety and parking-aid system, Betty in the trunk capacity, and maximum speed. At the car dealer,  two cars $C_1$ and $C_2$ are in their price range; their scores (on a scale $\{0,1\}$) w.r.t.~the parameters are reported in the following table:

{\small
\begin{center}
\begin{tabular}{|r| c| c |c| c|c|}
\hline
&\multirow{2}{4 em }{Safety} & Fuel & Parking-aid & Trunk & Max \\
& & Consumption &  System &  Capacity &  Speed\\
\hline
$C_1$ & 1 & 0 & 1 & 0& 1 \\
\hline
 $C_2$ & 0 & 1 & 0 & 1 & 0\\
 \hline
\end{tabular}
\end{center}
}
The couple will choose the car which scores better w.r.t.~most parameters they agree to be relevant. Let $X\coloneqq\{s,f,p,t,m\}$
be the set of  parameters, and
\[W: = \{(w_s, w_f, w_p, w_t,w_m)\mid w_x\in \{0, 1\} \mbox{ for all } x\in \{s, f, p, t,m\}\},\] 
be the space of profiles. The profiles corresponding to $C_1$ and $C_2$  are $(1,0,1,0,1)$ and $(0,1,0,1,0)$ respectively. For any set of parameters $Y \subseteq X$ and any $w\in W$, we let the {\em sum-score} of $w$ on  $Y$ be $w_Y\coloneqq\sum_{y \in Y} w_y$. 
Then each $Y\subseteq X$ induces  the  preordered set  $P_Y: = (W, \leq_Y)$ such that $w \leq_Y u$ iff $ w_Y\leq  u_Y$ for all $w, u \in W$.  
This preorder is not a partial order (let alone a lattice), since it is not antisymmetric: for instance, for $Y = X$,  $w = (1, 0, 0, 0,0)\neq  (0, 1, 0,0,0) = u$; however, $w_X = 1 = u_X$. For every $Y\subseteq X$, we let $e^\Sigma_Y: = e_{\leq_Y}$ denote the equivalence relation induced by $\leq_Y$;\footnote{Hence, $e^\Sigma_{\varnothing} = \tau = W\times W$.} that is, $(w, u)\in e^\Sigma_Y$ iff $w\leq_Y u$ and $u\leq_Y w$, iff $\sum_{y\in Y} u_y = \sum_{y\in Y} w_y$ for all $u, w\in W$. Each preorder $\leq_Y$  has $|Y|+1$ $e^\Sigma_Y$-equivalence classes, each corresponding to 
a value of sum-score on $Y$ between $0$ and $|Y|$. Any profile in an equivalence class corresponding to a lower sum-score lies below any profile in an equivalence class corresponding to a higher  sum-score, as illustrated in Figure \ref{interrogative:fig:pre-order sum}. 
{\scriptsize
\begin{figure}
    \centering

\begin{tikzpicture}[
    node distance=0.5cm,
    every node/.style={draw, rectangle, minimum width=3.5cm, minimum height=1.2cm, align=center},
    -, thick
  ]

\node (w0) {$(0,0,0,0,0)$};
\node (w1) [above=of w0] {$\{w\in W\mid w_X = 1\}$};
\node (w2) [above=of w1] {$C_2\in \{w\in W\mid w_X = 2\}$};
\node (w3) [above=of w2] {$C_1\in \{w\in W\mid w_X = 3\}$};
\node (w4) [above=of w3] {$\{w\in W\mid w_X = 4\}$};
\node (w5) [above=of w4] {$(1,1,1,1,1)$};

\draw (w0) -- (w1);
\draw (w1) -- (w2);
\draw (w2) -- (w3);
\draw (w3) -- (w4);
\draw (w4) -- (w5);

\end{tikzpicture}
    \caption{structure of the  pre-order  $P_X$. }
    \label{interrogative:fig:pre-order sum}
\end{figure}
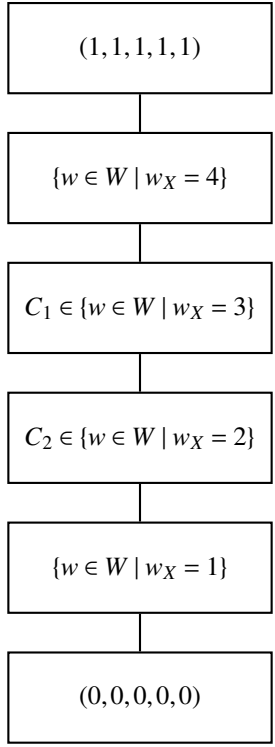
}

In the above scenario, each agent will individually judge the cars only based on  the parameters that they consider relevant.  
Using
the formal framework described in the preliminaries, for any $Y\subseteq X$, the equivalence relation $e^\Sigma_Y$ can be associated with the question (interrogative agenda)  `what is the given candidate's sum-score on $Y$?', as two candidates have the same answer to this question iff their profiles are in the same $e^\Sigma_Y$-equivalence class. Hence, in what follows, we  refer to  $e^\Sigma_Y$  as the agenda induced by $Y$, and we say that  $e^\Sigma_Y$ {\em prefers} $w$ over $u$ if  $u_Y< w_Y$. The following lemma shows that Definition \ref{interrogative:def:e prefers} simplifies to the latter definition when $e: = e^\Sigma_Y$.

\begin{lemma}
\label{interrogative:lemma:preference-sum}
For  any $Y\subseteq X$,  
\begin{enumerate}
\item 
${\leq_{e^\Sigma_Y}} = {\leq_Y}$.
   \item   $ {\leq_Y} =  e^\Sigma_Y \, \circ  {\leq}  \circ\, e^\Sigma_Y$, where $\leq$ is the pointwise order on $W$. 
\end{enumerate}
\end{lemma}
\begin{proof}
1. By definition,  
$e_{\leq_Y} = e^\Sigma_Y$, which implies, by Lemma \ref{interrogative:lem:preorder-equivalence back-forth}.1, ${\leq_Y} = {\leq_{e_{\leq_Y}}} = {\leq_{e^\Sigma_Y} }$.

2. Let $w, w', u, u'\in W$ s.t.~$(w, w')\in e^\Sigma_Y, (u, u')\in e^\Sigma_Y$  and $w' \leq u'$. Then $w_Y = w'_Y\leq u'_Y = u_Y$,  which shows  $w \leq_Y u$.

Conversely, if $w \leq_Y u$, then let $w'$ and $u'$ be the profiles such that $w'_y=w_y$ and $u'_y=u_y$ for all $y \in Y$, and  $w'_x=0 = u'_x$ for every $x \in X\setminus Y$. Let $u''$ be the profile  obtained by applying a permutation to the $Y$-indexed coordinates of  $u'$ so that, for every $y \in Y$,  $w'_y =1$ implies $u''_y =1$. Such a permutation exists because  $w'_Y =w_Y  \leq u_Y=u'_Y$. As the sum score is preserved under permutations,  $u''_Y =u'_Y=u_Y$, which shows that $(u, u'')\in e^\Sigma_Y$. Therefore, $w e^\Sigma_Y w'$, $w' \leq u''$ and $u'' e^\Sigma_Y u$ 
which shows 
$w (e^\Sigma_Y\, \circ \leq \circ\, e^\Sigma_Y) u$, as required. 
\end{proof}
\begin{remark}
   The  identity in item 2  of the lemma above implies that the corresponding identity  $\langle\leq_Y\rangle U= \langle e_Y \rangle \langle  \leq  \rangle \langle e_Y \rangle U $  holds as well for all $U\subseteq W$ and  $Y\subseteq X$, where, for any  $R\subseteq W\times W$, the symbol $\langle R \rangle$ denotes the standard diamond operator associated  with  $R$ (cf.~Section \ref{interrogative:ssec:modal logic}).  This observation will play out in future work, when further expanding the present framework along the lines of the standard modal logic approach outlined in Section \ref{interrogative:ssec:modal logic}. 
\end{remark}

\begin{remark}
For any  $Y\subseteq X$, if $e_Y$
prefers $w$ over $u$, then so does $e^\Sigma_Y$. 
 This is because if $u_y\leq w_y$ for every $y\in Y$ and $u_y<w_y$ for some $y\in Y$, then $u_Y = \sum_{y\in Y}u_y< \sum_{y\in Y}w_y = w_Y$. In the current scenario,  Alan and Betty's agendas are $e_\aga=e^\Sigma_A$ and $e_\agb=e^\Sigma_B$, where $A=\{s,f,p\}$, and $B=\{f,t,m\}$. Thus,  Alan  prefers $C_1$ over $C_2$, and Betty $C_2$ over $C_1$, since ${C_2}_A<{C_1}_A$ and ${C_1}_B<{C_2}_B$, while neither $e_A$ nor $e_B$ leads to a decision (hence, neither Alan nor Betty would have been able to decide  if the winning rule were total dominance). 
Note that any agenda $e^\Sigma_Y$ for  $Y\subseteq X$ s.t.~$|Y| = 3$  would lead to a decision between  $C_1$ and $C_2$.    
\end{remark}

Abusing notation, for any $y\in X$, we write $e^\Sigma_y$ instead of $e^\Sigma_{\{y\}}$. As discussed earlier on,  $e^\Sigma_x$  is the equivalence relation associated with the question `what is the given candidate's sum-score on  $x$?', for any  $x\in X$. 
As we score on a binary scale, the sum-scores over singletons coincide with binary scores; hence, $e^\Sigma_x = e_x$ as defined in the previous section, and any such equivalence relation corresponds to a binary partition of $W$. Therefore,   $e^\Sigma_x = e_x$ is a  meet-irreducible element of the lattice $E(W)$ (cf.~Proposition \ref{interrogative:prop:charact meet-irr}).

This means that  $e^\Sigma_Y$ cannot be expressed as $\bigsqcap_{y \in Y} e^\Sigma_y$. Indeed, $\bigsqcap_{y \in Y} e^\Sigma_y = \bigsqcap_{y \in Y} e_y = e_Y$ is the equivalence relation which identifies any two profiles which have the same score  on each $y\in  Y$. Hence, for any $Y\subseteq X$, all profiles $u, w\in W$ which are $e_{Y}$-equivalent  must also have  the same sum-score on $Y$, i.e.~they are $e^\Sigma_Y$-equivalent. However, the converse inequality does not need to hold,
witnessed e.g.~by $w =(0,1,1,0,0)$ and $u=(0,0,1,1,0)$ which are $e^\Sigma_Y$-equivalent for $Y\coloneqq\{f, p, t\}$ but are not $e_{Y}$-equivalent, 
nor are they $e_x$-equivalent for $x\in \{f, t\}$.  Thus, $ e^\Sigma_Y  \nleq e_Y =\bigsqcap_{y \in Y} {e}_y $, which implies $e^\Sigma_Y  \nleq e_y$ for some $y\in Y$. 

In fact, a stronger claim holds. Each $R\subseteq W\times W$ induces a function $R: \mathcal{P}(W)\to \mathcal{P}(W\times W)$ defined by the assignment $U\mapsto R_{| U\times U}$. The fact that $e_Y = \bigsqcap_{y\in Y}e_y$ implies that, for any finite nonempty $Y\subseteq X$, the map $e_Y: \mathcal{P}(W)\to \mathcal{P}(W\times W)$ is such that 
\[e_Y(U) = T(e_{y_1}(U),\ldots, e_{y_n}(U))\]
for some $|Y|$-ary term function $T$ on $\mathcal{P}(W\times W)$, namely the map $T$ obtained by composing the operation $\sqcap$ on $\mathcal{P}(W\times W)$ with itself $|Y|-1$ times. However, 
\begin{proposition}
\label{interrogative:prop:sum-broken-down}
For any $Y\subseteq X$, if $2\leq |Y|$, then there is no term function $T$ on $\mathcal{P}(W\times W)$ such that for all $U\subseteq W$,
\[e^\Sigma_Y(U) = T(e_{y_1}(U),\ldots, e_{y_n}(U)).\] 
\end{proposition}
\begin{proof}
For the sake of simplicity, we will prove the statement for $Y\coloneqq\{s,f\}$, but the proof can easily be generalized to any $Y\subseteq X$ with two or more elements.
Let $T$ be an $n$-ary term function on $\mathcal{P}(W\times W)$ which is built out of the operations $\sqcap, \sqcup, \neg,  (-)^{-1}, \circ, \mathsf{RST}(-)$. It  is not difficult to show, by induction on the shape of $T$, that \begin{equation}
\label{interrogative:eq:tg = gt}
    g(T(R_1, \ldots, R_n))=T(g(R_1), \ldots, g(R_n))
\end{equation} for any  bijective map $g:U \to U'$ and all $U, U' \subseteq W$, where $g(R)=\{(g(u), g(v)) \mid (u,v) \in R\}$ for any  $R\subseteq W\times W$. 

Let $U=\{w,u\}$ and $U'=\{w',u'\}$
for 
$w =(1,0,0,0,0)$, $u =(0,1,0,0,0)$, $w' =(0,0,0,0,0)$, $u' =(1,1,0,0,0)$. Let $g$ be the bijection s.t.~$g(w)=w'$ and $g(u)=u'$. Then, for all $y \in Y$,
\[g({e_y}(U))=g(\{(w, w), (u, u)\}) = \{(w', w'), (u', u')\} ={e_y}(U').\] 
Suppose for contradiction that $e^\Sigma_Y= T(e_s, e_f)$ for some term function $T$. Then, 
\[T(e_s(U), e_f(U)) =e^\Sigma_Y(U)= U \times U \]
and 
 \[T(e_s(U'), e_f(U'))=e^\Sigma_Y (U') =\{(w',w'), (u',u')\}.\]
 Hence, $g(T(e_s(U), e_f(U))) = g( U \times U)= U' \times U'$, whereas  \[T(g(e_s(U)), g(e_f(U)))= T(e_s(U'), e_f(U')) =\{(w',w'), (u',u')\},\] which shows that  $g(T(e_s,e_f)) \neq T(g(e_s),g(e_f))$, contradicting \eqref{interrogative:eq:tg = gt}.  
\end{proof}

Thus, as anticipated in Footnote \ref{interrogative:ftn:issues and parameters}, unlike the cases in which  total dominance is the winning rule, in the case of the sum-rule,   agendas associated with {\em sets} of parameters cannot be described in terms of agendas  associated with  {\em individual} parameters.

However,  agendas of the form $e^\Sigma_Y$ can still be described as meets of meet-irreducible elements (which can be understood as their {\em issues}), as follows. 
For any $Y\subseteq X$ and any $ 0 \leq k < |Y|$, let us consider the question: `is the sum-score of the given candidate on $Y$ less than equal to $k$?'.\footnote{For the purpose of obtaining the identity of Lemma \ref{interrogative:lem:sum_decomposition},  different yes/no questions would also work; for instance, questions of the form `is the sum-score of the given candidate on $Y$  equal to $k$?', which  gives rise to the bi-partition $\mathcal{E}_{Y, k_=}\coloneqq\{\{w\in W\mid w_Y = k\}, \{w\in W\mid w_Y\neq k\}\}$. However, all alternatives concern the parameter $k$, while  $Y$ remains constant.} This question can be associated with  the following  bipartition of $W$:  \[\mathcal{E}_{Y, k_\leq}\coloneqq\{\{w\in W\mid w_Y\leq k\}, \{w\in W\mid w_Y> k\}\}.\] 
Let $e^\Sigma_{Y, k_\leq}$  be the equivalence relation corresponding to $\mathcal{E}_{Y, k_\leq}$. By Proposition \ref{interrogative:prop:charact meet-irr}, any such equivalence relation is a meet-irreducible element of $E(W)$.  

\begin{lemma}\label{interrogative:lem:sum_decomposition}
    For any $Y\subseteq X$, \[e^\Sigma_Y = \bigsqcap_{0 \leq k < |Y|} e^\Sigma_{Y,k_\leq}.\]
\end{lemma}
\begin{proof}
For any $Y\subseteq X$ and any $0< k<|Y|$, the equivalence relation $e^\Sigma_{Y,k} \coloneqq   e^\Sigma_{Y,(k-1)_\leq}\rand e^\Sigma_{Y,k_\leq}$   defines the following tripartition  of $W$:
\[\{\{w\in W \mid w_Y \leq k-1 \}, \{w\in W \mid w_Y  = k \}, \{w\in W \mid w_Y > k \}\}.\] 
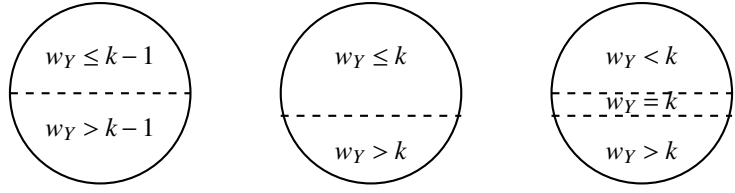
\begin{figure}[H]
    \centering
\begin{tikzpicture}[scale=0.6]
    \draw[thick] (0,0) circle(2); 
    \draw[thick, dashed] (-2,0) -- (2,0); 
    \node at (0, 0.8) {$w_Y \leq k-1$}; 
    \node at (0, -0.8) {$w_Y > k-1$}; 

    \draw[thick] (6,0) circle(2); 
    \draw[thick, dashed] (4,-0.5) -- (8,-0.5); 
    \node at (6, 0.8) {$w_Y \leq k$}; 
    \node at (6, -1.3) {$w_Y > k$}; 

    \draw[thick] (12,0) circle(2); 
    \draw[thick, dashed] (10,0) -- (14,0); 
    \draw[thick, dashed] (10,-0.5) -- (14,-0.5); 
    \node at (12, 0.8) {$w_Y < k$}; 
    \node at (12, -0.2) { $w_Y =k$}; 
    \node at (12, -1.3) {$w_Y > k$}; 
\end{tikzpicture}
   \caption{From left to  right, partitions associated with  $e^\Sigma_{Y,(k-1)_\leq}$, $e^\Sigma_{Y,k_\leq}$, and  $e^\Sigma_{Y,k}$.}
    \label{fig: meet-irreducible sum}
\end{figure}
Therefore, $\bigsqcap\{ e^\Sigma_{Y,k_\leq} \mid 0 \leq k < |Y|  \}$ induces a partition of $W$ into $|Y|+1$ cells, each of which   corresponds to a value of $w_Y$ between $0$ and $|Y|$. Since this is the same partition induced by $e^\Sigma_Y$, the two equivalence relations coincide.
\end{proof}

Thanks to the lemma above, for any $Y\subseteq X$,  the agenda $e^\Sigma_Y$  can be described  as a meet of $|Y|$ meet-irreducible elements, which can be understood as the {\em issues} supported by this agenda.\footnote{Hence, under the sum-rule, the issues and the parameters are decoupled.}   Thus, the total number of bipartitions needed to meet-generate all the agendas of the form $e^\Sigma_Y$ is bounded by $\sum_{0 \leq k \leq n} k {n \choose k}= O(n2^n)$, where $n = |X|$ is the total number of parameters under consideration\footnote{This number is anyway much smaller than the number of all  bipartitions, which is $O(2^{2^{n-1}})$.}. Thus, as we did in the previous section, we let the algebra  of interrogative agendas 
be represented by the lattice  $\mathbb{D}$, i.e.~the complete sub-meet-semilattice   of   $E(W)$ 
meet-generated by  the equivalence relations of the form $ e^\Sigma_{Y,k_\leq}$ for any $Y\subseteq X$. Clearly, being a complete meet-semilattice, $\mathbb{D}$ is also a lattice.

\begin{lemma}
    If $|X|\geq 2$, then the lattice $\mathbb{D}$ of interrogative agendas
 is not distributive.
\end{lemma} \begin{proof}
    Let us preliminarily observe that $e^\Sigma_{\{x\},0_\leq} =e_x$ for every $x\in X$; indeed, by definition, $e^\Sigma_{\{x\},0_\leq}$ is the equivalence relation defined by the bi-partition $\{\{w\in W\mid \pi_x(w) = 0\},\{w\in W\mid \pi_x(w) = 1\} \}$, which is the same partition induced by $e_x$.
   Therefore, for each $x \in X$, $e_x$ is an element of $\mathbb{D}$. Hence, so is the identity relation $\Delta_X = \bigsqcap_{x \in X} e_x$, which is the bottom element of  $\mathbb{D}$. Let $ e=e^\Sigma_{Y,1_\leq}$ for some $Y \subseteq X$ s.t.~$|Y| = 2$. Then $e$ induces the bipartition $\{\{w\in W\mid w_{y_1}+w_{y_2}\leq 1\}, \{w\in W\mid w_{y_1}+w_{y_2}> 1\}\}$ of $W$, and hence  $e_x \neq e $ for any $x\in X$. Thus,  $e_x \ror e=\top$, where $\top$ is the relation $W \times W$. Therefore, $\bigsqcap_{x \in X} (e_x \ror e)=\top$. However, $(\bigsqcap_{x \in X} e_x) \ror e= \Delta_X \ror e = \bot\ror e = e\neq \top$, which shows the failure of distributivity. 
\end{proof}

There are several  ways in which Alan and Betty can agree to a common agenda; here we only mention the two most straightforward ones, each of which leads to a decision.

(1) If Alan and Betty agree to include only  the parameters considered relevant by  {\em both} of them, their common agenda  is  $e^\Sigma_f$, which prefers $C_2$ over $C_1$.

(2) If Alan and Betty agree to include all the parameters which {\em either} of them considers relevant, their common agenda is $e^\Sigma_X$, which prefers $C_1$ over $C_2$.   This is in contrast with the case of total dominance as the winning rule, where including all parameters considered relevant by some  agents can never lead to a decision whenever the individual agents' interrogative agendas induce diverging outcomes. 

Summing up, if the winning rule in a deliberation is the sum-rule, the  interrogative agendas induced by  subsets of relevant parameters form a lattice, which is not in general distributive, let alone a Boolean algebra. However, this lattice can still be described in terms of the meet-irreducible elements of $E(W)$ generating it.

\paragraph{Sum rule and  many-valued  scores. } 
The current setting straightforwardly generalizes to one in which every parameter $x\in X$ is associated with a linearly ordered range of values $\mathbb{L}_x = (L_x, \leq_x)$ such that $L_x\subseteq \mathbb{Q}\cap [0, 1]$ for every $x\in X$, 
and hence, as observed in the previous section,  the set of  profiles is the bounded distributive lattice $W := \prod_{x\in X}\mathbb{L}_x$. In this setting,  the preorder $(W,\leq_Y)$ associated with each $Y \subseteq X$ can be defined verbatim as $w \leq_Y u$ iff $w_Y=\sum_{y \in Y}w_y \leq \sum_{y \in Y}u_y = u_Y$ for all $w,u\in Y$, and  $e^\Sigma_Y: = e_{\leq_Y}$ is the equivalence relation induced by $\leq_Y$.  Each preorder $\leq_Y$ has $\sum_{y \in Y}|L_y|$ $e^\Sigma_Y$-equivalence classes, each corresponding to a value of sum-score on $Y$. The statement and proof of Lemma \ref{interrogative:lemma:preference-sum} hold verbatim in the  present setting. Similarly, the proof of  Lemma \ref{interrogative:lem:sum_decomposition} straightforwardly generalizes to the proof that, for any $Y \subseteq X$, 
\[
e^\Sigma_Y = \bigsqcap_{0 \leq k \leq \sum_{y \in Y}|L_y|}
e^\Sigma_{Y, k_\leq}. 
\]
Thus, $e^\Sigma_Y$ is still generated by agendas of the form $e^\Sigma_{Y, k_\leq}$; however, the range of $k$ is now determined by the exact scale used for scoring each parameter in $Y$.

We do not consider  non-linear scales under the sum rule, because if $x\in Y\subseteq X$ and $(L_x, \leq_x)$  is nonlinear,  there is no uniform way to define the sum-score of a profile $w$ over $Y$ as $w_Y = \sum_{y\in Y} w_y$, and omitting $w_x$ is equivalent to scoring it  $0$ on a  linear scale.

\section{Interrogative agendas and coalitions}
\label{interrogative:sec:interrogative_agendas_and_coalitions}
In this section,  a multi-type
algebraic  framework is introduced and studied for modelling the scenarios discussed in the previous sections. 

\paragraph{Multi-type framework.} We consider terms of two types:  $\mathsf{IA}$ (interrogative agendas) and $\mathsf{C}$ (coalitions). Terms of type  $\mathsf{C}$ are interpreted over finite Boolean algebras with operators (BAOs) $\mathbb{C}$ based on the powerset algebras of the given  set of agents. Individual agents can then be represented as the   join-irreducible elements (atoms) $\nomj\in \jty(\mathbb{C})$ of these Boolean algebras.  Hence, for any $c\in \mathbb{C}$ and $\nomj\in \jty(\mathbb{C})$, the inequality $\nomj\leq c$ can be interpreted as `agent $\nomj$ is a member of coalition $c$'.

Terms of type  $\mathsf{IA}$ are interpreted over some algebra  $\mathbb{D}$ of interrogative agendas which is completely meet-generated by some proper subset $X\subseteq\mty(E(W))$, where $W$ is the  feature space under consideration. Each element of $X$ represents a basic parameter or {\em issue}.   Hence, for any $e\in \mathbb{D}$ and $\cnomm\in X$, the inequality $e\leq \cnomm$ can be understood as `agenda $e$ supports issue $\cnomm$'.  If each $\cnomm\in X$ is the kernel of a projection map to a single  dimension in the feature space $W$,  the lattice  $\mathbb{D}$ is isomorphic to a Boolean algebra (cf.~Footnote \ref{footn: D is BA}).

\paragraph{Unary homogeneous connectives. } Consider the following relation:
\[I\subseteq \jty(\mathbb{C})\times \jty(\mathbb{C}) \quad\quad \nomi I \nomj\; \mbox{ iff\; agent } \nomi \mbox{  {\em influences}  agent } \nomj.\] The relation $I$ induces  the modal operators $\diamdot, \diamdotb$  on $\mathbb{C}$ in the standard way.\footnote{ Similarly, we could model other  relevant relationships between agents (hierarchy, trust, etc).}
That is, for every $\nomj\in \jty(\mathbb{C})$ we let ${\diamdot}\nomj: = I^{-1}[\nomj] = \{\nomi\mid \nomi I\nomj\}$ and ${\diamdotb}\nomj: = I[\nomj] = \{\nomi\mid \nomj I\nomi\}$. Then, these assignments can be extended to unary operations on $\mathbb{C}$ by letting, for every $c\in \mathbb{C}$,
\[{\diamdot} c: = \bigcup \{{\diamdot}\nomj\mid \nomj\leq c\} \quad \quad {\diamdotb} c: = \bigcup \{{\diamdotb}\nomj\mid \nomj\leq c\}.\]
So for any set $c$ of agents, ${\diamdot} c$ denotes the set of  agents who influence some member of $c$ (the largest set of influencers of $c$), and $\diamdotb c$  the set of agents who are influenced by some member of $c$ (the largest audience of $c$). The dual operator  $\boxdot c := \neg {\diamdot} \neg c$  denotes the set of  agents who can influence only members of $c$, if any, and ${\boxdotb} c: = \neg {\diamdotb} \neg c$  the set of  agents who are only  influenced by  members of $c$, if any. 
 Hence, for any $\nomj\in \jty(\mathbb{C})$,  ${\diamdot} \nomj$ represents the influencers of $\nomj$; $\diamdotb \nomj$  the audience of $\nomj$; $ \boxdot \nomj$  the set of  agents who  influence at most $\nomj$, and $ \boxdotb \nomj$  the set of  agents who  are  influenced by $\nomj$ at most.


\paragraph{Unary heterogeneous connectives. } Consider the following relation:
\[R\subseteq \mty(\mathbb{D})\times \jty(\mathbb{C}) \quad\quad \cnomm R \nomj\; \mbox{ iff\; issue } \cnomm \mbox{ is {\em relevant} to agent } \nomj.\]
The relation $R$ induces the  operations $\Diamond, {\rhd}: \mathbb{C}\to \mathbb{D}$ defined as follows: for every agent $\nomj$, let $\Diamond \nomj = {\rhd}\nomj: = \bigsqcap R^{-1}[\nomj]$, where $R^{-1}[\nomj] := \{\cnomm\mid \cnomm R \nomj\}$. Then, for every $c\in \mathbb{C}$,
\[\Diamond c: = \bigsqcup \{\Diamond \nomj\mid \nomj\leq c\} \quad\quad {\rhd}c: = \bigsqcap \{{\rhd} \nomj\mid \nomj\leq c\}.\]
Under the intended interpretation of $R$,
for every coalition $c$, the interrogative agenda ${\rhd} c$  denotes  the {\em distributed agenda} of $c$ (i.e.~${\rhd} c$ is the interrogative agenda supporting exactly those issues which are supported by {\em at least one} member of $c$), while
$\Diamond c$ is the {\em common agenda} of $c$ (i.e.~$\Diamond c$ is the interrogative agenda supporting exactly those issues which are supported by {\em all} members of $c$). Algebraically: \[{\rhd}c\leq \cnomm\quad \mbox{ iff }\quad {\rhd}\nomj\leq\cnomm\; \mbox{ for some }\; \nomj\leq c \quad\quad\quad
\Diamond c\nleq \cnomm\quad \mbox{ iff }\quad \Diamond\nomj\nleq\cnomm\; \mbox{ for some }\; \nomj\leq c.\]

\begin{proposition}
\label{prop:diamond and triangle normal}
For every $\mathcal{C}\subseteq \mathbb{C}$,
\begin{enumerate}
\item $\Diamond \bot_{\mathbb{C}} = \bot_{\mathbb{D}}$ and $ \Diamond(\bigcup\mathcal{C}) = \bigsqcup\{\Diamond c\mid  c\in \mathcal{C} \}$;
\item ${\rhd} \bot = \tau$ and ${\rhd}$ is antitone;
\item If the issues meet-generating $\mathbb{D}$ are completely meet-prime in $\mathbb{D}$, then $ {\rhd}(\bigcup\mathcal{C}) = \bigsqcap\{{\rhd} c \mid  c\in \mathcal{C}$ for every $\mathcal{C}\subseteq \mathbb{C}\}$.
\end{enumerate}
\end{proposition}
\begin{proof}
1.~By definition, $\Diamond \bot = \bigsqcup\{\Diamond \nomj\mid \nomj\leq \bot\} = \bigsqcup \varnothing = \bot_{\mathbb{D}}$. If $c_1\leq c_2$ then $\{\Diamond \nomj\mid \nomj\leq c_1\}\subseteq \{\Diamond \nomj\mid \nomj\leq c_2\}$, and hence $\Diamond c_1: = \bigsqcup\{\Diamond \nomj\mid \nomj\leq c_1\}\leq \bigsqcup\{\Diamond \nomj\mid \nomj\leq c_2\}: = \Diamond c_2$. This implies that $ \bigsqcup\{\Diamond c\mid c\in \mathcal{C}\}\leq \Diamond (\bigcup \mathcal{C})$ for any $\mathcal{C}\subseteq \mathbb{C}$.
%
For the  converse inequality, it is enough to show that  if $\cnomm\in \mty(E(W))$ and  $\bigsqcup\{\Diamond c\mid c\in \mathcal{C}\}\leq \cnomm$,  then $\Diamond (\bigcup \mathcal{C})\leq \cnomm$. By definition, $\Diamond (\bigcup \mathcal{C}): = \bigsqcup\{\Diamond \nomj\mid \nomj\leq \bigcup \mathcal{C}\}$. Hence, it is enough to show that if $\nomj\in J^\infty(\mathbb{C})$ s.t.~$\nomj\leq \bigcup \mathcal{C}$, then $\Diamond\nomj\leq \cnomm$. Since $\mathbb{C}$ is a Boolean algebra, $\nomj$ is completely meet-prime, hence $\nomj\leq \bigcup \mathcal{C}$ implies $\nomj\leq c$ for some $c\in \mathcal{C}$. By monotonicity, this implies that $\Diamond \nomj\leq \Diamond c$. 
The assumption $\bigsqcup\{\Diamond c\mid c\in \mathcal{C}\}\leq \cnomm$
is equivalent to $\Diamond c \leq \cnomm$ for every $c\in\mathcal{C}$, hence $\Diamond\nomj\leq \Diamond c\leq \cnomm$, as required. 

2.~By definition, ${\rhd} \bot = \bigsqcap\{{\rhd} \nomj\mid \nomj\leq \bot\} = \bigsqcap \varnothing = \top_{\mathbb{D}} = \tau$. If $c_1\leq c_2$ then $\{{\rhd} \nomj\mid \nomj\leq c_1\}\subseteq \{{\rhd} \nomj\mid \nomj\leq c_2\}$, and hence ${\rhd} c_2: = \bigsqcap\{\Diamond \nomj\mid \nomj\leq c_2\}\leq \bigsqcap\{{\rhd} \nomj\mid \nomj\leq c_1\}: = {\rhd} c_1$. 

3.~The antitonicity of $\rhd$ implies that  ${\rhd} (\bigcup \mathcal{C}) \leq \bigsqcap \{{\rhd} c\mid c\in\mathcal{C}\}$. For the converse inequality, it is enough to show that if $\cnomm\in \mty(E(W))$ and ${\rhd} (\bigcup \mathcal{C})\leq \cnomm$, then $\bigsqcap \{{\rhd} c\mid c\in\mathcal{C}\}\leq \cnomm$. Since $\cnomm$ is completely meet-prime, the assumption $\bigsqcap \{{\rhd}\nomj\mid \nomj\leq \bigcup \mathcal{C}\} = {\rhd} (\bigcup \mathcal{C})\leq \cnomm$ implies ${\rhd}\nomj\leq \cnomm$ for some $\nomj\leq \bigcup \mathcal{C}$. Since $\nomj$ is completely join-prime, $\nomj\leq \bigcup \mathcal{C}$ implies $\nomj\leq c$ for some $c\in\mathcal{C}$. Hence, by antitonicity, $\bigsqcap \{{\rhd} c\mid c\in\mathcal{C}\}\leq {\rhd}c\leq {\rhd}\nomj \leq \cnomm$, as required.
\end{proof}
The fact that $\Diamond$ and $ {\rhd}$ are completely join-preserving and completely join-reversing, respectively, implies the existence of their adjoints $\blacksquare, {\blacktriangleright}: \mathbb{D}\to \mathbb{C}$ defined as follows:
for every interrogative agenda $e$,
 \[\blacksquare e: = \bigcup\{c\mid \Diamond c\leq e\}\quad \quad {\blacktriangleright} e: = \bigcup\{c\mid e\leq {\rhd} c\}.\]
  That is, $\blacksquare e$ is the largest coalition $c$ such that each member of $c$ considers relevant all issues of $e$; the coalition ${\blacktriangleright} e$ is the largest coalition $c$ such that  each issue of $e$ is considered relevant by at least one member of $c$.

When $\mathbb{D}$ is a Boolean algebra, we can define $\Box c: = \neg \Diamond \neg c$. Notice that $\Diamond \neg c$ is the interrogative agenda which supports all the issues that are considered relevant by all the agents who are not members  of coalition $c$. Hence, \[\Box c: = \bigsqcap \{\cnomm \mid \exists \nomj (\nomj\nleq c \mbox{ and } \Diamond \nomj\nleq \cnomm)\}.\]  

\paragraph{Binary heterogeneous connectives.} Consider the following relation:
\[S\subseteq \mty(\mathbb{D})\times \jty(\mathbb{C})\times \mty(\mathbb{D}) \quad\quad S(\cnomn, \nomj, \cnomm)\; \mbox{ iff\;  agent $\nomj$ would {\em substitute}  issue $\cnomm$ with issue $\cnomn$. } \] 
The relation $S$ induces the  operation $\pdra: \mathbb{C}\times \mathbb{D}\to \mathbb{D}$ defined as follows: for every agent $\nomj$ and issue $\cnomm$, let $\nomj\pdra \cnomm   : = \bigsqcap S^{(0)}[\nomj, \cnomm]$, where $S^{(0)}[\nomj, \cnomm] := \{\cnomn\mid S(\cnomn, \nomj, \cnomm)\}$. Intuitively, $\nomj\pdra \cnomm$ is the interrogative agenda supporting exactly the issues that agent $\nomj$ prefers to issue $\cnomm$.

Then, for every $c\in \mathbb{C}$ and $e\in \mathbb{D}$,
\[ c\pdra e: = \bigsqcup \{\nomj\pdra \cnomm\mid \nomj\leq c \mbox{ and } e\leq \cnomm\}.\]
Intuitively, $c\pdra e$ is the agenda representing the shared view among the members of $c$ of how the issues in $e$ should be modified.
\begin{proposition}
\label{prop:pdra normal}
For every $e\in \mathbb{D}$, $\mathcal{D}\subseteq \mathbb{D}$, $c\in\mathbb{C}$ and $\mathcal{C}\subseteq \mathbb{C}$,
\begin{enumerate}
\item $ \bot_\mathbb{C} \pdra e = \bot_{\mathbb{D}}$ and $c\pdra \tau = \bot_{\mathbb{D}}$;

\item  $\pdra$ is monotone in its first coordinate and is antitone in its second coordinate; 
\item $ (\bigcup \mathcal{C}) \pdra e = \bigsqcup\{ c \pdra e \mid   c\in\mathcal{C} \}$;

\item If the issues meet-generating $\mathbb{D}$ are completely meet-prime in $\mathbb{D}$, then $c \pdra (\bigsqcap \mathcal{D}) =  \bigsqcup\{c \pdra e \mid e\in\mathcal{D}\}$.
\end{enumerate}
\end{proposition}
\begin{proof}
1. By definition, $ \bot_{\mathbb{C}} \pdra e = \bigsqcup\{ \nomj\pdra \cnomm \mid \nomj\leq \bot_{\mathbb{C}} \mbox{ and } e\leq \cnomm\} = \bigsqcup \varnothing = \bot_{\mathbb{D}}$. Likewise, $ c \pdra \tau = \bigsqcup\{ \nomj\pdra \cnomm \mid \nomj\leq c \mbox{ and } \tau\leq \cnomm\} = \bigsqcup \varnothing = \bot_{\mathbb{D}}$.

2. If $c_1\leq c_2$ then $\{\nomj\pdra \cnomm \mid \nomj\leq c_1 \mbox{ and } e\leq \cnomm\}\subseteq \{\nomj\pdra \cnomm \mid \nomj\leq c_2 \mbox{ and } e\leq \cnomm\}$, and hence $ c_1 \pdra e: = \bigsqcup\{\nomj\pdra \cnomm \mid \nomj\leq c_1 \mbox{ and } e\leq \cnomm\}\leq \bigsqcup\{\nomj\pdra \cnomm \mid \nomj\leq c_2 \mbox{ and } e\leq \cnomm\} = c_2 \pdra e$. 

If $e_1\leq e_2$ then $\{\nomj\pdra \cnomm \mid \nomj\leq c \mbox{ and } e_2\leq \cnomm\}\subseteq \{\nomj\pdra \cnomm \mid \nomj\leq c \mbox{ and } e_1\leq \cnomm\}$, and hence $ c \pdra e_2: = \bigsqcup\{\nomj\pdra \cnomm \mid \nomj\leq c \mbox{ and } e_2\leq \cnomm\}\leq \bigsqcup\{\nomj\pdra \cnomm \mid \nomj\leq c \mbox{ and } e_1\leq \cnomm\} = c \pdra e_1$.

3. The monotonicity of $\pdra$ in its first coordinate implies that $\bigsqcup \{ c\pdra e\mid c\in \mathcal{C}\}\leq (\bigcup \mathcal{C}) \pdra e$. For the converse inequality, it is enough to show that  if $\cnomn\in \mty(E(W))$ and  $\bigsqcup \{ c\pdra e\mid c\in \mathcal{C}\}\leq \cnomn$,  then $(\bigcup \mathcal{C}) \pdra e \leq \cnomn$. By definition, $(\bigcup \mathcal{C}) \pdra e : = \bigsqcup \{\nomj\pdra \cnomm\mid \nomj\leq \bigcup\mathcal{C} \text{ and } e\leq\cnomm\}$. Hence, it is enough to show that if $\nomj\in J^\infty(\mathbb{C})$ s.t.~$\nomj\leq \bigcup \mathcal{C}$ and  $\cnomm\in M^\infty(\mathbb{D})$  s.t.~$e\leq \cnomm$, then $\nomj\pdra\cnomm\leq \cnomn$. Since $\mathbb{C}$ is a Boolean algebra, $\nomj$ is completely meet-prime, hence $\nomj\leq \bigcup \mathcal{C}$ implies $\nomj\leq c$ for some $c\in \mathcal{C}$. By monotonicity, this implies that $ \nomj\pdra \cnomm\leq  c\pdra \cnomm$. The assumption $\bigsqcup \{ c\pdra e\mid c\in \mathcal{C}\}\leq \cnomn$ is equivalent to $c\pdra e\leq \cnomn$ for all $c\in \mathcal{C}$. Hence, by the antitonicity of $\pdra$ in its second coordinate, $\nomj\pdra\cnomm\leq c\pdra \cnomm\leq c\pdra e\leq \cnomn$, as required.

4. The antitonicity of $\pdra$ in its second coordinate implies that $\bigsqcup \{c \pdra e\mid  e\in\mathcal{D}\} \leq c  \pdra (\bigsqcap\mathcal{D})$. For the converse inequality, it is enough to show that if $\cnomn\in \mty(E(W))$ and $\bigsqcup \{c \pdra e\mid  e\in\mathcal{D}\} \leq \cnomn$ then $c \pdra (\bigsqcap\mathcal{D})\leq \cnomn$. By definition, $c \pdra (\bigsqcap\mathcal{D}): = \bigsqcup \{\nomj\pdra \cnomm\mid \nomj\leq c\text{ and }\bigsqcap\mathcal{D}\leq \cnomm\}$. Hence, it is enough to show that  if $\nomj\in J^\infty(\mathbb{C})$ s.t.~$\nomj\leq c$ and  $\cnomm\in M^\infty(\mathbb{D})$  s.t.~$\bigsqcap\mathcal{D}\leq \cnomm$, then $\nomj\pdra\cnomm\leq \cnomn$. Since by assumption $\cnomm$ is completely meet-prime, 
$\bigsqcap\mathcal{D}\leq \cnomm$ implies that $e\leq \cnomm$ for some $e\in\mathcal{D}$. Hence by monotonicity, $\nomj\pdra \cnomm\leq c\pdra e\leq \cnomn$, as required.
\end{proof}
The fact that $\pdra$ is completely join-preserving in its first coordinate and completely meet-reversing in its second coordinate   implies the existence of its residuals $\star: \mathbb{D}\times \mathbb{D}\to \mathbb{C}$ and $\apdra: \mathbb{C}\times \mathbb{D}\to \mathbb{D}$ such that for all  $e_1, e_2\in \mathbb{D}$ and any $c\in \mathbb{C}$,
\[c\pdra e_1\leq e_2 \quad \mbox{ iff }\quad c\leq e_1 \star e_2 \quad \mbox{ iff }\quad c \apdra e_2 \leq e_1\]
Hence by construction,
 \[ e_1 \star e_2: = \bigcup\{c\mid c\pdra e_1\leq e_2 \}\quad \mbox{ and }\quad c\apdra e: = \bigsqcap\{e'\mid c\pdra e'\leq e\},\]

\noindent and by well known order-theoretic facts, $\star$ is completely meet-preserving in each coordinate, while $\apdra$ is completely join-preserving in the first coordinate and completely meet-reversing in the second coordinate. 

The relation $S$ also induces the  operation $\br: \mathbb{C}\times \mathbb{D}\to \mathbb{D}$ defined as follows: for every agent $\nomj$ and issue $\cnomm$, let $\nomj\br \cnomm   : = \nomj\pdra \cnomm =  \bigsqcap S^{(0)}[\nomj, \cnomm]$, where $S^{(0)}[\nomj, \cnomm]$ has been defined above.
Then, for every $c\in \mathbb{C}$ and $e\in \mathbb{D}$,
\[ c\br e: = \bigsqcap \{\nomj\br \cnomm\mid \nomj\leq c \mbox{ and } e\leq \cnomm\}.\]
Intuitively, $c\br e$ is the agenda representing the distributed view among the members of $c$ of how the issues in $e$ should be modified.
\begin{proposition}
\label{prop:br normal}
For every $e\in \mathbb{D}$, $\mathcal{D}\subseteq \mathbb{D}$, $c\in\mathbb{C}$ and $\mathcal{C}\subseteq \mathbb{C}$,
\begin{enumerate}
\item $ \bot_{\mathbb{C}} \br e = \tau$ and $c\br \tau = \tau$;
\item  $\br$ is antitone in its first coordinate and  monotone in its second coordinate;
\item If the issues meet-generating $\mathbb{D}$ are completely meet-prime in $\mathbb{D}$, then $ (\bigcup \mathcal{C}) \br e = \bigsqcap\{c\br e\mid c\in\mathcal{C}\}$ and $c \br (\bigsqcap \mathcal{D}) =  \bigsqcap \{c \br e\mid  e\in\mathcal{D}\}$.
\end{enumerate}
\end{proposition}
\begin{proof}
1. By definition, $ \bot \br e = \bigsqcap\{ \nomj\br \cnomm \mid \nomj\leq \bot \mbox{ and } e\leq \cnomm\} = \bigsqcap \varnothing = \top_{\mathbb{D}} = \tau$. Likewise, $ c \br \tau = \bigsqcap\{ \nomj\br \cnomm \mid \nomj\leq c \mbox{ and } \tau\leq \cnomm\} = \bigsqcap \varnothing = \top_{\mathbb{D}} = \tau$.

2. If $c_1\leq c_2$ then $\{\nomj\br \cnomm \mid \nomj\leq c_1 \mbox{ and } e\leq \cnomm\}\subseteq \{\nomj\br \cnomm \mid \nomj\leq c_2 \mbox{ and } e\leq \cnomm\}$, and hence $ c_2 \br e: = \bigsqcap\{\nomj\br \cnomm \mid \nomj\leq c_2 \mbox{ and } e\leq \cnomm\}\leq \bigsqcap\{\nomj\br \cnomm \mid \nomj\leq c_1 \mbox{ and } e\leq \cnomm\} = c_1 \br e$. 

If $e_1\leq e_2$ then $\{\nomj\br \cnomm \mid \nomj\leq c \mbox{ and } e_2\leq \cnomm\}\subseteq \{\nomj\br \cnomm \mid \nomj\leq c \mbox{ and } e_1\leq \cnomm\}$, and hence $ c \br e_1: = \bigsqcap\{\nomj\br \cnomm \mid \nomj\leq c \mbox{ and } e_1\leq \cnomm\}\leq \bigsqcap\{\nomj\br \cnomm \mid \nomj\leq c \mbox{ and } e_2\leq \cnomm\} = c \br e_2$. 

3. The antitonicity of $\br$ in its first coordinate implies that  $(\bigcup \mathcal{C})  \br e \leq \bigsqcap \{c \br e\mid c\in\mathcal{C}\}$. For the converse inequality, it is enough to show that if $\cnomn\in \mty(E(W))$ and $(\bigcup \mathcal{C})  \br e\leq \cnomn$ then $\bigsqcap \{c \br e\mid c\in\mathcal{C}\} \leq \cnomn$. Since  $\cnomn$ is completely meet-prime, the  assumption $ \bigsqcap\{\nomj\br \cnomm\mid \nomj\leq\bigcup \mathcal{C}\text{ and } e\leq \cnomm\} = (\bigcup \mathcal{C})  \br e \leq \cnomn$ implies that $\nomj\br\cnomm\leq \cnomn$ for some $\nomj\in J^\infty(\mathbb{C})$ s.t.~$\nomj\leq \bigcup\mathcal{C}$ and some $\cnomm\in M^\infty(\mathbb{D})$ s.t.~$e\leq \cnomm$. Since $\nomj$ is completely join-prime, $\nomj\leq \bigcup\mathcal{C}$ implies $\nomj\leq c$ for some $c\in\mathcal{C}$. Hence, by monotonicity and antitonicity, $\bigsqcap\{c\br e\mid c\in\mathcal{C}\} \leq c\br e\leq \nomj\br \cnomm\leq \cnomn$, as required.

As to the second identity, the monotonicity of $\br$ in its second coordinate implies that  $c  \br (\bigsqcap\mathcal{D})\leq \bigsqcap \{ c \br e \mid   e\in\mathcal{D}\} $. For the converse inequality, it is enough to show that if $\cnomn\in \mty(E(W))$ and $c  \br (\bigsqcap\mathcal{D})\leq \cnomn$ then
$\bigsqcap \{ c \br e \mid   e\in\mathcal{D}\}\leq \cnomn$. Since  $\cnomn$ is completely meet-prime, the  assumption $\bigsqcap \{\nomj\br \cnomm\mid \nomj\leq c \mbox{ and } \bigsqcap\mathcal{D}\leq \cnomm\} = c  \br (\bigsqcap\mathcal{D})\leq \cnomn$ implies that $\nomj\br \cnomm\leq \cnomn$ for some $\nomj\in J^\infty(\mathbb{C})$ s.t.~$\nomj\leq c$ and some $\cnomm\in M^\infty(\mathbb{D})$ s.t.~$\bigsqcap\mathcal{D}\leq \cnomm$. Since $\cnomm$ is completely meet-prime, $\bigsqcap\mathcal{D}\leq \cnomm$ implies $e\leq \cnomm$ for some $e\in\mathcal{D}$. Hence, by monotonicity and antitonicity, $\bigsqcap\{c\br e\mid e\in\mathcal{D}\} \leq c\br e\leq \nomj\br \cnomm\leq \cnomn$, as required.
\end{proof}

The fact that $\br$ is completely join-reversing in its first coordinate and completely meet-preserving in its second coordinate implies the existence of its residuals $\abr: \mathbb{D}\times \mathbb{D}\to \mathbb{C}$ and $\mand: \mathbb{C}\times \mathbb{D}\to \mathbb{D}$ such that for all  $e_1, e_2\in \mathbb{D}$ and any $c\in \mathbb{C}$,
\[e_1 \leq c \br e_2\quad \mbox{ iff } \quad c\mand e_1\leq e_2\quad \mbox{ iff } \quad c \leq e_1\abr e_2.\]
Hence by construction,
 \[ e_1 \abr e_2: = \bigcup\{c\mid e_1 \leq c \br e_2 \}\quad \mbox{ and }\quad c\mand e: = \bigsqcap\{e'\mid c\pdra e\leq e'\},\]

 \noindent and by well known order-theoretic facts, $\abr$ is completely join-reversing in its first coordinate and completely meet-preserving in its second coordinate, while $\mand$ is completely join-preserving in each coordinate.

\paragraph{Algebraic/logical framework.} Based on the discussion so far, we  define the  {\em abstract heterogeneous algebras of interrogative agendas and coalitions} as the triples $\mathbb{H} = (\mathbb{C}, \mathbb{D}, \mathcal{H})$ such that

\begin{enumerate}
    \item   $\mathbb{C}$ is a Boolean algebra, the elements of which represent coalitions of agents, expanded with modal operators $\diamdot$, $\diamdotb$, $\boxdot$, $\boxdotb$, capturing various aspects of the social network among the various coalitions;
    \item  $\mathbb{D}$ is a (distributive) lattice, the elements of which represent interrogative agendas;
     \item $\mathcal{H}$ is the set of heterogeneous connectives listed below. The  connectives in each column are adjoints or residuals of each other, as indicated in the discussion above.
\begin{center}
\begin{tabular}{r| r| r| r}
$\Diamond : \mathbb{C} \to \mathbb{D}$ & ${\rhd} :  \mathbb{C}\to \mathbb{D}$ & $\pdra  :   \mathbb{C} \times\mathbb{D}\to \mathbb{D}$ & $\br: \mathbb{D}\times \mathbb{C} \to\mathbb{D}$ \\
$\blacksquare : \mathbb{D} \to  \mathbb{C}$ & ${\blacktriangleright} :  \mathbb{C} \to\mathbb{D}$ & ${\star}:\mathbb{D} \times \mathbb{D}\to  \mathbb{C}$ & ${\mand}: \mathbb{C} \times \mathbb{D} \to \mathbb{D}$\\
&&  $\apdra: \mathbb{C}\times \mathbb{D} \to \mathbb{D}$ & $\abr: \mathbb{D} \times \mathbb{D} \to  \mathbb{C}$ \\
\end{tabular}\end{center}
\end{enumerate}
The {\em canonical extension} of each such $\mathbb{H}$ is $\mathbb{H}^\delta\coloneqq (\mathbb{C}^\delta, \mathbb{D}^\delta, \mathcal{H}^\delta)$ such that $\mathbb{C}^\delta$ and $\mathbb{D}^\delta$ are the canonical extensions of $\mathbb{C} $ and $\mathbb{D}$, respectively (cf.~\cite{jonsson1951boolean, gehrke2004bounded}), and each operation in $\mathcal{H}^\delta$ is the $\sigma$-extension  (resp.~$\pi$-extension) of an operation  $h\in \mathcal{H}$ which is a left (resp.~right) adjoint or residual. 

The following multi-type language is naturally associated with these structures: let $\mathsf{AtC}$ and $\mathsf{AtIA}$ be disjoint denumerable sets of atomic coalition variables and atomic interrogative agenda variables, respectively. 
\begin{center}
\begin{tabular}{cl}
$\mathsf{IA} \ni e:: =$ & $ q\mid \tau\mid \bot \mid e\sqcap e\mid e\sqcup e\mid \Diamond c\mid {\rhd}c\mid {\blacktriangleright} c\mid c\pdra e\mid c\apdra e\mid c\br e\mid c\mand e$\\
$\mathsf{C} \ni c:: =$ & $ t \mid\top\mid \bot_{\mathsf{C}} \mid c\cap c\mid c\cup c\mid \neg c\mid \diamdot c \mid \diamdotb c\mid \boxdot c\mid \boxdotb c\mid \blacksquare e\mid e\star e\mid e{\abr} e$ \\
\end{tabular}
\end{center}
where $t\in \mathsf{AtC}$ and $q\in \mathsf{AtIA}$.
A {\em valuation} on $\mathbb{H}$ is a map $v: \mathsf{AtC}\cup\mathsf{AtIA}\to \mathbb{H}$ such that $v(t)\in \mathbb{C}$ for every $t\in \mathsf{AtC}$ and $v(q)\in \mathbb{D}$ for every $q\in \mathsf{AtIA}$. A {\em model} is a pair $(\mathbb{H}, v)$ such that $\mathbb{H}$ and $v$ are as above. The {\em basic logic of interrogative agendas and coalitions} consists of two sets of sequents $c_1\vdash c_2$ and $e_1\vdash e_2$, one for each type, each of which is closed under the usual axioms and rules for Boolean algebras with operators (in the case of type $\mathsf{C}$) and (distributive) lattices (in the case of type $\mathsf{IA}$) with operators. Additional axioms and rules capture the (finitary versions of) the order-theoretic properties discussed in Propositions \ref{prop:diamond and triangle normal}, \ref{prop:pdra normal}, and \ref{prop:br normal}, as well as the adjunctions/residuations among the various connectives. A model $(\mathbb{H}, v)$ {\em satisfies} a sequent $c_1\vdash c_2$ (resp.~$e_1\vdash e_2$) when $\mathbb{C}\models v(c_1)\leq v(c_2)$ (resp.~$\mathbb{D}\models v(e_1)\leq v(e_2)$). The notion of {\em validity} is defined, as usual, as satisfaction for every valuation, and is denoted $\mathcal{H}\models c_1\vdash c_2$ (resp.~$\mathcal{H}\models e_1\vdash e_2$).  The proof of the completeness of the basic logic w.r.t.~the class of models defined above is a routine Lindenbaum-Tarski construction, and is omitted. Via standard discrete Stone/Priestley-type duality (cf.~), this logic can also be endowed with a complete relational semantics, which we discuss in Section \ref{interrogative:sec:interaction conditions}.
%


\section{Interaction conditions} 
\label{interrogative:sec:interaction conditions}
In the previous section, we have introduced a basic (abstract) algebraic/logical framework for interrogative agendas and coalitions, in which certain heterogeneous (unary and binary) modal operators arise from some relations $I$, $R$, and $S$, linking agents and issues in various ways, the intended interpretation of which captures information on the social structure of the agents, as well as the building blocks of their cognitive stances, in the form of which agents influence who, which issues the agents consider relevant, and which issues they would like to see replaced by which issues. In this section, we  discuss examples of first-order conditions involving these relations which can yield an even richer picture of the agents' attitudes, and adapt and use the (inverse) correspondence-theoretic tools developed e.g.~in \cite{conradie2012algorithmic,conradie2019algorithmic, palmigiano2024unified} to explore whether these conditions can be captured by some axioms in the language of  the basic logic introduced in the previous section.
%
The conditions we are going to consider are listed in   Table \ref{interrogative:tab:Properties and defining conditions}, together with their names. In what follows, for any $\mathbb{H} = (\mathbb{C}, \mathbb{D}, \mathcal{H})$ as above, we let $I\subseteq J^\infty(\mathbb{C}^\delta)\times J^\infty(\mathbb{C}^\delta)$, $R\subseteq M^\infty(\mathbb{D}^\delta)\times J^\infty(\mathbb{C}^\delta)$, and $S\subseteq M^\infty(\mathbb{D}^\delta)\times J^\infty(\mathbb{C}^\delta)\times M^\infty(\mathbb{D}^\delta)$. Moreover, $\nomi, \nomj, \nomk, \nomh$ are variables ranging in $J^\infty(\mathbb{C}^\delta)$ and $\cnomm, \cnomn, \cnomo$  in $M^\infty(\mathbb{D}^\delta)$.

{{\footnotesize
\begin{longtable}
{|m{17em} |m{21em}|}
    \hline
       \textbf{Name} &
       \textbf{Condition} 
       \\
       \hline
    $S$  symmetric\footnote{Symmetry is a form of widespread irrationality.} & $\forall \nomj\forall \cnomm\forall \cnomn[S(\cnomn, \nomj, \cnomm)\Rightarrow S(\cnomm, \nomj, \cnomn)]$
    \\
    \hline
      $S$  anti-symmetric  & $\forall \nomj\forall \cnomm\forall \cnomn[(S(\cnomn, \nomj, \cnomm)\ \&\ S(\cnomm, \nomj, \cnomn))\Rightarrow \cnomm \leq \cnomn]$ \\
    \hline
    $S$  unanimous & $\forall \nomj\forall\nomi\forall \cnomm\forall \cnomn [S( \cnomm, \nomj, \cnomn) \Rightarrow S( \cnomm, \nomi, \cnomn)]$ \\
    \hline
    $S$  reflexive &  $\forall \nomj \forall \cnomm[ S(\cnomm,\nomj,\cnomm)]$ \\
    \hline
    $S$  transitive & $\forall \nomj\forall \cnomm\forall \cnomo\forall \cnomn[(S(\cnomn, \nomj, \cnomm)\ \&\ S(\cnomo, \nomj, \cnomn))\Rightarrow S(\cnomo, \nomj, \cnomm)]$ \\
    \hline 
    $S$   globally indifferent & $\forall \nomj \forall \cnomm \forall \cnomn[ S(\cnomm,\nomj,\cnomn)]$ \\
    \hline 
    $S$  Euclidean & $\forall \nomj \forall \cnomm \forall \cnomn \forall \cnomo [(S( \cnomn,\nomj,\cnomm )\ \&\ S( \cnomo,\nomj,\cnomm )) \Rightarrow S( \cnomn,\nomj,\cnomo ) ]$ \\
    \hline 
    $S$  single-stepped\footnote{If $S$ is single-stepped, then $S$ is transitive.} & $\forall \nomj\forall \cnomm_1\forall \cnomm_2\forall \cnomn[(S(\cnomn, \nomj, \cnomm_1)\ \&\ S(\cnomm_2, \nomj, \cnomn))\Rightarrow \cnomn \leq \cnomm_2]$ \\
    \hline
    $S$  reasonably ductile & $\forall\nomj\forall \nomi\forall \cnomm\forall \cnomn[(\nomj I\nomi\ \&\ \cnomm R\nomj\ \& \ \cnomm R\nomi \ \& \ \cnomn R\nomj )\Rightarrow S(\cnomn, \nomi,\cnomm )]$
    \\
    \hline 
    $S$  positively coherent with $R$ &  $\forall \nomj\forall \cnomm[ \cnomm R\nomj \Rightarrow S(\cnomm, \nomj, \cnomm)]$ \\
    \hline 
     $S$  negatively coherent with $R$ &  $\forall \nomj\forall \cnomm[S(\cnomm, \nomj, \cnomm)\Rightarrow \cnomm R\nomj]$ \\
    \hline 
      $S$  negatively preference-coherent with $R$ &  $\forall\nomj\forall\cnomm\forall\cnomn [S(\cnomn,\nomj,\cnomm) \Rightarrow \cnomn R \nomj]$ \\
    \hline 
     $S$  positively preference-coherent with $R$ & $\forall\nomj\forall\cnomm\forall\cnomn [ \cnomn R \nomj \Rightarrow S(\cnomn,\nomj,\cnomm) ]$ \\
     \hline 
     $S$    $R$-equanimous\footnote{If $S$ is $R$-equanimous, then $S$ is positively coherent with $R$} & $\forall \nomj\forall \cnomm\forall \cnomn[ (\cnomm R\nomj\ \& \ \cnomn R\nomj) \Rightarrow S(\cnomn, \nomj, \cnomm)]$\\
     \hline
     $S$   $R$ bi-coherent & $\forall \nomj\forall \cnomm \forall \cnomn[(\cnomm R \nomj\ \&\  \neg (\cnomn R\nomj) )\Rightarrow S(\cnomm, \nomj, \cnomn)  ]$ 
     \\
     \hline 
  
    $S$  $R$-intransigent & $\forall \nomj\forall \cnomm\forall \cnomn[ (\cnomm R\nomj\ \& \ S(\cnomn, \nomj, \cnomm)) \Rightarrow \cnomm\leq \cnomn]$\\
    \hline 
    $I$  positively coherent with $S$ & $ \forall \nomj \forall \nomi\forall\cnomm\forall\cnomn [(S(\cnomm, \nomj,\cnomn)\ \&\ \nomj I \nomi) \Rightarrow  S(\cnomm,\nomi,\cnomn)]$\\
    \hline 
        $I$  negatively coherent with $S$ & $ \forall \nomj\forall \nomi\forall\cnomm\forall\cnomn [(\neg S(\cnomm, \nomj,\cnomn)\ \& \ \nomj I \nomi) \Rightarrow  \neg S(\cnomm,\nomi,\cnomn)]$\\
    \hline 
     $I$, $R$, $S$  coherent with each other & $ \forall \nomj \forall \nomi \forall \cnomm \forall \cnomn [(\cnomn R \nomj\ \&\ \nomj I \nomi) \Rightarrow S(\cnomn, \nomi,  \cnomm)]$ \\
      \hline 
    \caption{Conditions and their names}
\label{interrogative:tab:Properties and defining conditions}
\end{longtable}

}}

The coherence conditions in Table \ref{interrogative:tab:Properties and defining conditions} can be understood as different `rational explanations' of the attitudes of agents towards issues, captured in terms  of  the relations $I$, $R$, and $S$. 

In the following proposition, some of the conditions above are shown to be equivalent to axioms (in the form of inequalities) in the algebraic/logical framework outlined at the end of the previous section.
In proving it, we will make use of the fact that, with the exception of one axiom only involving constants, the other axioms are of a particularly simple syntactic shape, referred to as {\em primitive} (see \cite[Definition 28]{greco2016unified}).

\begin{proposition}
\label{interrogative:prop:inv:correspondence}
For any $\mathbb{H} = (\mathbb{C}, \mathbb{D}, \mathcal{H})$ and any $I$, $R$, $S$ as above, 
\begin{enumerate}
\item $\mathbb{H}\models S\text{ symmetric }$  iff  $\;\mathbb{H}\models e_1\star e_2\vdash e_2\star  e_1$  iff $\;\mathbb{H}\models c\apdra e\vdash  c\pdra e$.
\item $\mathbb{H}\models S \text{ positively coherent with } R$ iff  $\;\mathbb{H}\models\blacksquare e \vdash e \star e$.
\item $\mathbb{H}\models S \text{ negatively coherent with } R$ iff  $\;\mathbb{H}\models e\star e\vdash \blacksquare e$.
\item $\mathbb{H}\models S \text{ negatively preference-coherent with } R$ iff  $\;\mathbb{H}\models\Diamond c \leq c \pdra \tau$.
\item $\mathbb{H}\models S \text{ positively preference-coherent with } R$ iff  $\;\mathbb{H}\models c \pdra \bot  \vdash  \Diamond c$.
\item $\mathbb{H}\models S\; R$-intransigent iff  $\;\mathbb{H}\models\blacksquare e\vdash  e \abr e $.
\item $\mathbb{H}\models S\; R$-equanimous iff  $\;\mathbb{H}\models\blacksquare(e_1 \sqcap e_2) \vdash e_1 \star e_2$.

\item $\mathbb{H}\models S\text{ globally indifferent }$  iff  $\;\mathbb{H}\models \top\vdash \bot \star \bot$.  
\item   $\mathbb{H}\models I \text{ positively coherent with } S$ iff  $\;\mathbb{H}\models {\diamdotb} c\pdra e \vdash  c\pdra e$.
\item  $\mathbb{H}\models I \text{ negatively coherent with } S$ iff  $\;\mathbb{H}\models{\diamdot} c\pdra e\vdash  c\pdra e $.

\item $\mathbb{H}\models I, R, S$ coherent with each other iff  $\;\mathbb{H}\models\blacksquare e \leq {\boxdot} (e\star \bot)$. 

\end{enumerate}
\end{proposition}

\begin{proof} 
The following computations are justified by unified algorithmic correspondence and  inverse correspondence.

1. \begin{tabular}{c l l}
    & $\forall \nomj\forall \cnomm\forall \cnomn[S(\cnomn, \nomj, \cnomm)\Rightarrow S(\cnomm, \nomj, \cnomn)]$ \\
iff & $\forall \nomj\forall \cnomm\forall \cnomn[ \nomj\pdra\cnomm\leq \cnomn \Rightarrow  \nomj\pdra\cnomn\leq \cnomm]$ & definition of  $\pdra$\\
iff & $\forall \nomj\forall \cnomm\forall \cnomn[\nomj\leq \cnomn\star \cnomm  \Rightarrow \nomj \leq \cnomm\star \cnomn]$ &  $\pdra \dashv \star$ \\
iff & $\forall \cnomm\forall \cnomn[\cnomn\star \cnomm \leq \cnomm\star \cnomn ]$ & \\
iff & $\forall e_1\forall e_2[e_1\star e_2\leq e_2\star  e_1]$. & \cite[Prop.~37(2)]{greco2016unified}\\
\end{tabular}

2. \begin{tabular}{c l l}
    & $\forall \nomj\forall \cnomm[  \cnomm R\nomj  \Rightarrow S(\cnomm, \nomj, \cnomm) ]$ \\
iff & $\forall \nomj\forall \cnomm[ \Diamond \nomj \leq \cnomm \Rightarrow  \nomj\pdra\cnomm\leq \cnomm  ]$ & def.~of $\Diamond$ and $\pdra$\\
iff & $\forall \nomj\forall \cnomm[\nomj \leq \blacksquare \cnomm
\Rightarrow \nomj\leq \cnomm\star \cnomm  ]$ & $\Diamond \dashv \blacksquare$ and $\pdra \dashv \star$ \\
iff & $\forall \cnomm[   \blacksquare \cnomm \leq \cnomm\star \cnomm]$ & \\
iff & $\forall e[\blacksquare e\leq  e\star e ]$. & \cite[Prop.~37(2)]{greco2016unified}\\
\end{tabular}

3. \begin{tabular}{c l l}
    & $\forall \nomj\forall \cnomm[ S(\cnomm, \nomj, \cnomm)  \Rightarrow \cnomm R\nomj ]$ \\
iff & $\forall \nomj\forall \cnomm[ \nomj\pdra\cnomm\leq \cnomm \Rightarrow  \Diamond \nomj \leq \cnomm ]$ & def.~of $\Diamond$ and $\pdra$\\
iff & $\forall \nomj\forall \cnomm[\nomj\leq \cnomm\star \cnomm  \Rightarrow \nomj \leq \blacksquare \cnomm]$ & $\Diamond \dashv \blacksquare$ and $\pdra \dashv \star$ \\
iff & $\forall \cnomm[\cnomm\star \cnomm \leq  \blacksquare \cnomm ]$ & \\
iff & $\forall e[ e\star e\leq \blacksquare e]$. & \cite[Prop.~37(2)]{greco2016unified}\\
\end{tabular}

4. \begin{tabular}{c l l}
& $\forall\nomj\forall\cnomm\forall\cnomn [S(\cnomn,\nomj,\cnomm) \Rightarrow \cnomn R \nomj]$\\
iff & $\forall\nomj\forall\cnomm\forall\cnomn [\nomj\pdra \cnomm \leq \cnomn  \Rightarrow \Diamond \nomj\leq \cnomn ]$ & def.~of $\Diamond$ and $\pdra$\\
iff & $\forall\nomj\forall\cnomm [\Diamond \nomj\leq\nomj\pdra \cnomm  ]$\\
iff & $\forall\nomj[\Diamond \nomj\leq\nomj\pdra \tau  ]$\\
iff & $\forall c[\Diamond c\leq c\pdra \tau  ]$. &  \cite[Prop.~37(1)]{greco2016unified}\\
\end{tabular}

5. \begin{tabular}{c l l}
& $\forall\nomj\forall\cnomm\forall\cnomn [ \cnomn R \nomj \Rightarrow S(\cnomn,\nomj,\cnomm) ]$\\
iff & $\forall\nomj\forall\cnomm\forall\cnomn [ \Diamond \nomj\leq \cnomn\Rightarrow \nomj\pdra \cnomm \leq \cnomn   ]$ & def.~of $\Diamond$ and $\pdra$\\
iff & $\forall\nomj\forall\cnomm [\nomj\pdra \cnomm \leq \Diamond \nomj ]$\\
iff & $\forall\nomj[\nomj\pdra \bot \leq \Diamond \nomj ]$\\
iff & $\forall c[ c\pdra \bot\leq   \Diamond c]$. &  \cite[Prop.~37(1)]{greco2016unified} \\
\end{tabular}

6. \begin{tabular}{c l l}
    & $\forall \nomj\forall \cnomm\forall \cnomn[ (\cnomm R\nomj\ \& \ S(\cnomn, \nomj, \cnomm)) \Rightarrow \cnomm\leq \cnomn]$\\
iff & $\forall \nomj\forall \cnomm\forall \cnomn[ (\Diamond\nomj \leq \cnomm\ \& \  \nomj\br \cnomm \leq \cnomn) \Rightarrow \cnomm\leq \cnomn]$ &  def.~of $\Diamond$ and $\br$\\
iff & $\forall \nomj\forall \cnomm[ \Diamond\nomj \leq \cnomm\Rightarrow \forall \cnomn[  \nomj\br \cnomm \leq \cnomn \Rightarrow \cnomm\leq \cnomn]]$\\
iff & $\forall \nomj\forall \cnomm[ \Diamond\nomj \leq \cnomm\Rightarrow \cnomm\leq  \nomj\br \cnomm ]$\\
iff & $\forall \nomj\forall \cnomm[ \nomj \leq \blacksquare\cnomm\Rightarrow \nomj \leq  \cnomm \abr \cnomm ]$  &  $\Diamond \dashv \blacksquare$ and  $\mand \dashv \br$,  $\mand \dashv \abr$\\
iff & $\forall \cnomm [\blacksquare\cnomm \leq  \cnomm \abr \cnomm ]$ & \\
iff & $\forall e [\blacksquare e\leq  e \abr e ]$. & \cite[Prop.~37(2)]{greco2016unified}\\
\end{tabular}

7. \begin{tabular}{c l l}
    & $\forall \nomj\forall \cnomm  \forall \cnomn[(\cnomm R\nomj\ \& \ \cnomn R\nomj)\Rightarrow S(\cnomn, \nomj, \cnomm)]$ &\\
    iff & $\forall \nomj\forall \cnomm  \forall \cnomn[(\Diamond \nomj\leq \cnomm\ \&\ \Diamond \nomj\leq\cnomn)\Rightarrow \nomj\pdra \cnomm\leq \cnomn]$ &  def.~of $\Diamond$ and $\pdra$\\
iff & $\forall \nomj\forall \cnomm  \forall \cnomn[(\nomj\leq \blacksquare \cnomm\ \&\ \nomj\leq \blacksquare \cnomn)\Rightarrow \nomj\leq \cnomn\star \cnomm]$ & $\Diamond \dashv \blacksquare$ and $\pdra \dashv \star$\\
iff & $\forall \nomj\forall \cnomm  \forall \cnomn[\nomj\leq \blacksquare (\cnomn \sqcap \cnomm)\Rightarrow \nomj\leq \cnomn\star \cnomm]$\\
iff & $\forall e_1 \forall e_2 [\blacksquare(e_1 \sqcap e_2) \leq e_1 \star e_2]$. & \cite[Prop.~37(2)]{greco2016unified} \\
\end{tabular}

8. \begin{tabular}{c l l}
    &$\forall \nomj \forall \cnomm \forall \cnomn[ S(\cnomm,\nomj,\cnomn)]$\\
    iff &$\forall \nomj \forall \cnomm \forall \cnomm[ \nomj\pdra \cnomn \leq \cnomm]$ & def.~of $\pdra$\\
     iff &$\forall \nomj \forall \cnomm\forall \cnomn[\nomj\leq \cnomm \star \cnomn]$ & $\pdra \dashv \star$\\
      iff &$\top\leq \bot \star \bot$.\\
\end{tabular}

9. 
\begin{tabular}{c l l}
    &$\forall \nomj \forall \nomi\forall\cnomm\forall\cnomn [(S(\cnomm, \nomj,\cnomn)\ \&\ \nomj I \nomi) \Rightarrow  S(\cnomm,\nomi,\cnomn)]$\\
    iff &$\forall \nomj \forall \nomi\forall\cnomm\forall\cnomn [( \nomj\pdra \cnomn\leq \cnomm\ \&\ \nomi\leq {\diamdotb}\nomj) \Rightarrow  \nomi\pdra \cnomn\leq \cnomm]$ &   def.~of $\diamdotb$ and $\pdra$\\
 iff &$\forall \nomj \forall\cnomm\forall\cnomn [ \nomj\pdra \cnomn\leq \cnomm \Rightarrow \forall \nomi( \nomi\leq {\diamdotb}\nomj \Rightarrow  \nomi\leq  \cnomm\star \cnomn)]$\\
 iff &$\forall \nomj \forall\cnomm\forall\cnomn [ \nomj \pdra \cnomn\leq \cnomm \Rightarrow  {\diamdotb}\nomj \leq  \cnomm\star \cnomn]$\\
iff &$\forall \nomj \forall\cnomm\forall\cnomn [ \nomj\pdra \cnomn\leq \cnomm \Rightarrow  {\diamdotb}\nomj\pdra \cnomn \leq  \cnomm]$ & $\pdra \dashv \star$\\
iff &$\forall \nomj \forall\cnomn [  {\diamdotb}\nomj\pdra \cnomn \leq  \nomj\pdra \cnomn]$\\
 iff &$\forall c \forall e [  {\diamdotb} c\pdra e \leq  c\pdra e]$. &  \cite[Prop.~37(2)]{greco2016unified}\\
    
\end{tabular}

10. 
\begin{tabular}{c l l}
    &$\forall \nomj \forall \nomi\forall\cnomm\forall\cnomn [(S(\cnomm, \nomi,\cnomn)\ \&\ \nomj I \nomi) \Rightarrow  S(\cnomm,\nomj,\cnomn)]$\\
    iff &$\forall \nomj \forall \nomj\forall\cnomm\forall\cnomn [( \nomi\pdra \cnomn\leq \cnomm\ \&\ \nomj\leq {\diamdot}\nomi) \Rightarrow  \nomj\pdra \cnomn\leq \cnomm]$ &  def.~of  $\diamdot$ and $\pdra$\\
    iff &$\forall \nomi \forall\cnomm\forall\cnomn [ \nomi\pdra \cnomn\leq \cnomm\Rightarrow \forall \nomj(\nomj\leq {\diamdot}\nomi \Rightarrow  \nomj\leq  \cnomm\star \cnomn)]$\\
    iff &$\forall \nomi \forall\cnomm\forall\cnomn [ \nomi\pdra \cnomn\leq \cnomm\Rightarrow  {\diamdot}\nomi \leq  \cnomm\star \cnomn]$\\
    iff &$\forall \nomi \forall\cnomm\forall\cnomn [ \nomi\pdra \cnomn\leq \cnomm\Rightarrow  {\diamdot}\nomi\pdra \cnomn \leq  \cnomm]$ & $\pdra \dashv \star$\\
    iff &$\forall \nomi \forall\cnomn [ {\diamdot}\nomi\pdra \cnomn\leq  \nomi\pdra \cnomn ]$\\
    iff &$\forall c \forall e [ {\diamdot} c\pdra e\leq  c\pdra e ]$. &   \cite[Prop.~37(2)]{greco2016unified}\\
\end{tabular}

11.
\begin{tabular}{c l l}
& $\forall \nomj \forall \nomi 
\forall \cnomm\forall \cnomn [(\cnomn R \nomj\ \&\ \nomj I \nomi) \Rightarrow S(\cnomn, \nomi,  \cnomm)]$\\
iff & $\forall \nomj \forall \nomi \forall \cnomm \forall \cnomn [(\Diamond \nomj\leq \cnomn\ \&\ \nomi \leq {\diamdotb} \nomj) \Rightarrow  \nomi\pdra  \cnomm\leq \cnomn]$ &  def.~of $\Diamond$, $\diamdotb$ and $\pdra$\\
iff & $\forall \nomj \forall \nomi \forall \cnomm \forall \cnomn [( \nomj\leq \blacksquare\cnomn\ \&\ \nomi \leq {\diamdotb} \nomj) \Rightarrow  \nomi\pdra  \cnomm\leq \cnomn]$ & $\Diamond \dashv \blacksquare$\\
iff & $ \forall \nomi \forall \cnomm \forall \cnomn [\nomi \leq {\diamdotb} \blacksquare\cnomn \Rightarrow  \nomi\leq  \cnomn\star \cnomm]$ & $\pdra \dashv \star$\\
iff & $ \forall \cnomm \forall \cnomn [ {\diamdotb} \blacksquare\cnomn \leq  \cnomn\star \cnomm]$\\
iff & $  \forall \cnomn [ {\diamdotb} \blacksquare\cnomn \leq  \cnomn\star \bot]$\\
iff & $  \forall \cnomn [  \blacksquare\cnomn \leq {\boxdot} (\cnomn\star \bot)]$ & $\diamdotb \dashv \boxdot$\\
iff & $  \forall e [  \blacksquare e \leq {\boxdot} (e\star \bot)]$.  &   \cite[Prop.~37(2)]{greco2016unified}\\
\end{tabular}

\end{proof}

\noindent In what follows, we show that the conditions in Table \ref{interrogative:tab:Properties and defining conditions} which are not mentioned in Proposition \ref{interrogative:prop:inv:correspondence} are not definable in the multi-type language introduced in Section \ref{interrogative:sec:interrogative_agendas_and_coalitions}.

To do so, it is enough to consider heterogeneous algebras $\mathbb{H} = (\mathbb{C}, \mathbb{D}, \mathcal{H})$ such that $\mathbb{D}$ is a Boolean algebra; each such $\mathbb{H}$ can be associated  with a relational structure $\mathbb{H}_+ = (J^\infty(\mathbb{C}^\delta), M^\infty(\mathbb{D}^\delta), I_{\diamdot}, R_\Diamond, S_{\pdra})$ such that 
\begin{enumerate}
    \item $I_{\diamdot}\subseteq J^\infty(\mathbb{C}^\delta)\times J^\infty(\mathbb{C}^\delta)$ s.t.~$\nomj I\nomi$ iff $\nomj\leq {\diamdot}\nomi$ for all $\nomj,\nomi\in J^\infty(\mathbb{C}^\delta)$;
    \item $R_{\Diamond}\subseteq M^\infty(\mathbb{D}^\delta)\times J^\infty(\mathbb{C}^\delta)$ s.t.~$\cnomm R\nomj$ iff $\Diamond \nomj\leq \cnomm$ for all $\nomj\in J^\infty(\mathbb{C}^\delta)$ and $\cnomm\in M^\infty(\mathbb{D}^\delta)$;
    \item $S_{\pdra}\subseteq M^\infty(\mathbb{D}^\delta)\times J^\infty(\mathbb{C}^\delta)\times M^\infty(\mathbb{D}^\delta)$ s.t.~$S_{\pdra} (\cnomm,\nomj, \cnomn)$ iff $\nomj\pdra\cnomn\leq \cnomm$ for all $\nomj\in J^\infty(\mathbb{C}^\delta)$ and $\cnomm, \cnomn\in M^\infty(\mathbb{D}^\delta)$.
\end{enumerate}
Conversely, each relational structure $\mathcal{F} = (C, D, I, R, S)$ such that $C$ and $D$ are sets and $I\subseteq C\times C$, $R\subseteq D\times C$, $S\subseteq D\times C\times D$ gives rise to a heterogeneous algebra $\mathcal{F}^+ = (\mathbb{C}, \mathbb{D}, \mathcal{H})$ such that $\mathbb{C} = \mathcal{P}(C)$ and $\mathbb{D} = \mathcal{P}(D)^\partial$, and hence $C\cong J^\infty(\mathbb{C})$ and $D\cong M^\infty(\mathbb{D})$, and the operations in $\mathcal{H}$ are defined as indicated in Section \ref{interrogative:sec:interrogative_agendas_and_coalitions}. The proofs that $(\mathcal{F}^+)_+\cong \mathcal{F}$ and $(\mathbb{H}_+)^+\cong \mathbb{H}^\delta$ for each such $\mathcal{F}$ and $\mathbb{H}$ are routine and are omitted.

Clearly, the notions of modelhood, satisfaction, and validity transfer from the algebraic to the relational setting in the expected way. Moreover, the notions of bisimulation,  p-morphism, disjoint union of frames, and generated subframe straightforwardly generalize from the single-type setting of classical Kripke frames to the multi-type setting, and so do the invariance results of the validity of modal axioms under generated subframes, bounded morphic images and disjoint unions. The next proposition pivots on these invariance results.


\begin{proposition}
\label{interrogative:prop:gt}
The following conditions are not definable:

\begin{tabular}{llc ll}
1. &  $S$ transitive. && 5. & $S$  Euclidean.\\

2. & $S$  reflexive. && 6. & $S$  unanimous.\\
3. &  $S$  antisymmetric. && 7. & $S$  bi-coherent with $R$.\\
4. & $S$  single-stepped.
&& 8. & $S$  reasonably ductile.\\
\end{tabular}
\end{proposition}

\begin{proof}
Consider  
$\mathcal{F}_k=(C_k, D_k,  I_k, R_k, S_k)$, for $k=1,2$. Whenever the conditions in the statement do not involve $I$ or $R$, we let $I_k= R_k = \varnothing$. 

(1) Let
$C_1= \{j_1,j_2\}$,
$D_1= \{m_1,m_2,m_3\}$, $S_1= \{(m_1,j_1,m_2),(m_2,j_2,m_3)\}$, and $C_2= \{i\}$,
$D_2= \{n_1,n_2,n_3\}$, $S_2= \{(n_1,i,n_2),(n_2,i,n_3) \}$. The assignment $m_k\mapsto n_k$  for $1\leq k \leq 3 $, $j\mapsto i$ for each $j\in C_1$ defines a surjective bounded morphism $f:\mathcal{F}_1\to \mathcal{F}_2$. 
However, $\mathcal{F}_1\models S \,  transitive$, but  $\mathcal{F}_2\not\models S\, transitive$, which shows, by the invariance  of  validity under bounded morphic images, that $S \, transitive$  is not definable.

(2) Let $\mathrm{C}_k= \{j_k\}$,
$\mathrm{D}_k= \{m_k\}$,  $S_k= \{(m_k,j_k,m_k) \}$ for $k = 1, 2$. Then $\mathcal{F}_k\models S\,  reflexive$, but $\mathcal{F}_1\uplus \mathcal{F}_2\not\models S\, reflexive$, as  it is not the case that $S_1\uplus S_2(m_1,j_2,m_1)$.  Thus, by the invariance  of validity under disjoint unions,  $S\,  reflexive$ is not definable.

(3) Let $C_1= \{j\}$,
$D_1= \{m_k \mid k \in \mathbb{N}\}$, $S_1= \{(m_k,j,m_{k+1}) \mid k \in  \mathbb{N}\}$, and $C_2= \{i\}$,
$D_2= \{n_0,n_1\}$,  $S_2= \{(n_0,i,n_1),(n_1,i,n_0) \}$. The assignment $j\mapsto i $, and  $m_{2k}\mapsto n_0$,  $m_{2k+1}\mapsto n_1$ for every $k\in\mathbb{N}$ defines  a surjective bounded morphism  $f: \mathcal{F}_1\to \mathcal{F}_2$. 
However, $\mathcal{F}_1\models S\, antisymmetric$ , while $\mathcal{F}_2\not\models S\,antisymmetric $. 

(4) Let $C_1= \{j_1,j_2\}$,
$D_1= \{m_1,m_2,m_3\}$, $S_1= \{(m_1,j_1,m_2),(m_3,j_2,m_1) \}$, and $C_2= \{i\}$,
$D_2= \{n_1,n_2,n_3\}$,  $S_2= \{(n_1,i,n_2),(n_3,i,n_1) \}$. The assignment  $m_k\mapsto n_k$ for $1\leq k\leq 3$, $j\mapsto i$ for every $j\in C_1$  defines  a surjective bounded morphism $f:\mathcal{F}_1\to\mathcal{F}_2$.  However, $\mathcal{F}_1\models S \, single\text{-}stepped$, while $\mathcal{F}_2\not\models S\, single\text{-}stepped$. 

(5) Let  $C_1= \{j_1,j_2\}$,
$D_1= \{m_1, m_2, m_3\}$,  $S_1= \{(m_1,j_1,m_3),(m_2,j_2,m_3) \}$, and $C_2= \{i\}$,
$D_2= \{n_1, n_2,n_3\}$,  $S_2= \{(n_1,i,n_3),(n_2,i,n_3) \}$. The assignment  $m_k\mapsto n_k$ for $1\leq k\leq 3$, $j\mapsto i$ for every $j\in C_1$  defines  a surjective bounded morphism $f:\mathcal{F}_1\to\mathcal{F}_2$. However, $\mathcal{F}_1\models S \,  Euclidean$, while $\mathcal{F}_2\not\models S \, Euclidean$. 

(6) Let $C_k = \{j_k\}$, $D_k = \{m_k, n_k\}$, $S_k = \{(m_k, j_k, n_k)\}$ for $k = 1, 2$. Then $\mathcal{F}_k\models S\, unanimous$; however, $\mathcal{F}_1\uplus \mathcal{F}_2\not\models S\, unanimous$, as it is not the case that $(S_1\uplus S_2)(m_2, j_1, n_2)$. 

(7) Let $C_k = \{j_k\}$, $D_k = \{m_k\}$, $R_k = \{(m_k, j_k)\}$, $S_k = \varnothing$ for $k = 1, 2$. Then $\mathcal{F}_k\models S\, bicoherent$  with $R$; however, $\mathcal{F}_1\uplus \mathcal{F}_2\not\models S \, bicoherent$  with $R$, since $m_1 (R_1\uplus R_2)j_1$ and not $m_1(R_1\uplus R_2)j_2$, yet not $(S_1\uplus S_2)(m_1, j_2, m_2)$.


(8) Let  $C_1= \{j_1,i_1\}$,
$D_1= \{m_1, m_1', n_1\}$, $R_1= \{(m_1, j_1), (n_1, j_1), (m_1', i_1)\}$, $I_1 = \{(j_1, i_1)\}$, $S_1= \varnothing$, and $C_2=  \{j_2,i_2\}$,
$D_2= \{m_2, n_2\}$, $R_2= \{(m_2, j_2), (n_2, j_2), (m_2, i_2)\}$, $I_2 = \{(j_2, i_2)\}$, $S_2= \varnothing$. The assignment $m_1\mapsto m_2$, $m_1'\mapsto m_2$,  $n_1\mapsto n_2$,  $i_1\mapsto i_2$, $j_1\mapsto j_2$ defines a surjective bounded morphism $f: \mathcal{F}_1\to \mathcal{F}_2$. However, $\mathcal{F}_1\models \, reasonably\,\, ductile$, while $\mathcal{F}_2\not\models \ reasonably \,\,ductile$. 
\end{proof}

\begin{figure}[H]
{
    \centering
\begin{tikzpicture}[node distance=1.5cm, every node/.style={draw, circle}, every label/.style={draw=none},scale=0.6]
    \node (m1) at (-4,0) {\(m_1\)};
    \node[right of=m1] (m2) {\(m_2\)};
    \node[right of=m2] (m3) {\(m_3\)};
    \node[above of=m1] (j1) {\(j_1\)};
    \node[above of=m2] (j2) {\(j_2\)};
    
    \draw[->] (m1) -- (m2) node[midway, below, draw=none] {\(S_1\)};
    \draw[->] (m2) -- (m3) node[midway, below, draw=none] {\(S_1\)};
    \draw[->] (m1) -- (j1) node[midway, left, draw=none] {\(S_1\)};
    \draw[->] (m2) -- (j2) node[midway, left, draw=none] {\(S_1\)};
    
    \node (n1) at (4,0) {\(n_1\)};
    \node[right of=n1] (n2) {\(n_2\)};
    \node[right of=n2] (n3) {\(n_3\)};
    \node[above of=n2] (i) {\(i\)};
    
    \draw[->] (n1) -- (n2) node[midway, below, draw=none] {\(S_2\)};
    \draw[->] (n2) -- (n3) node[midway, below, draw=none] {\(S_2\)};
    \draw[->] (n1) -- (i)node[midway, left, draw=none] {\(S_2\)};
    \draw[->] (n2) -- (i)node[midway, left, draw=none] {\(S_2\)};
    
    \draw[dashed,->, bend right=50] (m1) to node[midway, below, draw=none] {\(f\)} (n1);
    \draw[dashed,->, bend right=50] (m2) to node[midway, below, draw=none] {\(f\)} (n2);
    \draw[dashed,->, bend right=50] (m3) to node[midway, below, draw=none] {\(f\)} (n3);
    
    \draw[dashed,->, bend left=50] (j1) to node[midway, above, draw=none] {\(f\)} (i);
    \draw[dashed,->, bend left=50] (j2) to node[midway, above, draw=none] {\(f\)} (i);
\end{tikzpicture}

\begin{tikzpicture}[node distance=1.5cm, every node/.style={draw, circle}, every label/.style={draw=none},scale=0.6]
    \node (m0) at (-6,0) {\(m_0\)};
    \node[right of=m0] (m1) {\(m_1\)};
    \node[right of=m1] (m2) {\(m_2\)};
    \node[right of=m2] (m3) {\(m_3\)};
    \node[above of=m0] (j) {\(j\)};
    
    \draw[->] (m0) -- (j) node[midway, left, draw=none] {\(S_1\)};
    \draw[->] (m0) -- (m1) node[midway, below, draw=none] {\(S_1\)};
    \draw[->] (m1) -- (j) node[midway, left, draw=none] {\(S_1\)};
    \draw[->] (m1) -- (m2) node[midway, below, draw=none] {\(S_1\)};
    \draw[->] (m2) -- (j) node[midway, left, draw=none] {\(S_1\)};
    \draw[->] (m2) -- (m3) node[midway, below, draw=none] {\(S_1\)};
    
    \draw[<->] (m3) -- (j) node[midway, left, draw=none] {\(S_1\)};
    
    \draw[->, dotted] (m3) -- ++(2, 0) node[midway, above, draw=none] {\(S_1\)}; 

    \node (n0) at (8,0) {\(n_0\)};
    \node[right of=n0] (n1) {\(n_1\)};
    \node[above of=n0] (i) {\(i\)};
    
    \draw[->] (n0) -- (n1)node[midway, below, draw=none] {\(S_2\)};
    \draw[->] (n0) -- (i)node[midway, left, draw=none] {\(S_2\)};
    \draw[->] (n1) -- (i)node[midway, right, draw=none] {\(S_2\)};
    \draw[->] (n1) -- (n0)node[midway, above, draw=none] {\(S_2\)};
    
    \draw[dashed,->, bend right=50] (m0) to node[midway, above, draw=none] {\(f\)} (n0);
    \draw[dashed,->, bend right=50] (m2) to node[midway, above, draw=none] {\(f\)} (n0);
    \draw[dashed,->, bend left=50] (m1) to node[midway, below, draw=none] {\(f\)} (n1);
    \draw[dashed,->, bend left=50] (m3) to node[midway, below, draw=none] {\(f\)} (n1);
    \draw[dashed,->, bend left=50] (j) to node[midway, below, draw=none] {\(f\)} (i);
\end{tikzpicture}

    \begin{tikzpicture}[node distance=1.5cm, every node/.style={draw, circle}, every label/.style={draw=none},,scale=0.6]
    \node (m1) at (-6,0) {\(m_1\)};
    \node[right of=m1] (m1p) {\(m_1'\)};
    \node[right of=m1p] (n1) {\(n_1\)};
    \node[above of=m1] (j1) {\(j_1\)};
    \node[above of=m1p] (i1) {\(i_1\)};
    
    \draw[->] (m1) -- (j1)node[midway, left, draw=none] {\({R}_1\)};
    \draw[->] (n1) -- (j1)node[midway, left, draw=none] {\({R}_1\)};
    \draw[->] (m1p) -- (i1)node[midway, right, draw=none] {\({R}_1\)};
    \draw[->] (j1) -- (i1)node[midway, above, draw=none] {\({R}_1\)};
    
    \node (m2) at (6,0) {\(m_2\)};
    \node[right of=m2] (n2) {\(n_2\)};
    \node[above of=m2] (j2) {\(j_2\)};
    \node[above of=n2] (i2) {\(i_2\)};
    
    \draw[->] (m2) -- (j2)node[midway, left, draw=none] {\({R}_2\)};
    \draw[->] (n2) -- (j2)node[midway, below, draw=none] {\({R}_2\)};
    \draw[->] (m2) -- (i2)node[midway, right, draw=none] {\({R}_2\)};
    \draw[->] (j2) -- (i2)node[midway, above, draw=none] {\({R}_2\)};
    
    \draw[dashed,->, bend right=50] (m1) to node[midway, above, draw=none] {\(f\)} (m2);
    \draw[dashed,->, bend right=50] (m1p) to node[midway, above, draw=none] {\(f\)} (m2);
    \draw[dashed,->, bend right=50] (n1) to node[midway, above, draw=none] {\(f\)} (n2);
    \draw[dashed,->, bend left=50] (i1) to node[midway, above, draw=none] {\(f\)} (i2);
    \draw[dashed,->, bend left=50] (j1) to node[midway, above, draw=none] {\(f\)} (j2);
\end{tikzpicture}
\caption{
\small From top to bottom, the diagrams represent models  discussed in  (1), (3), and (4). We represent the ternary relation $S(m_1,j,m_2)$ with an arrow from $m_1$ to $j$ and an arrow from $m_1$ to $m_2$. Similar diagrams can be drawn to depict other examples. 
}
}
\end{figure}
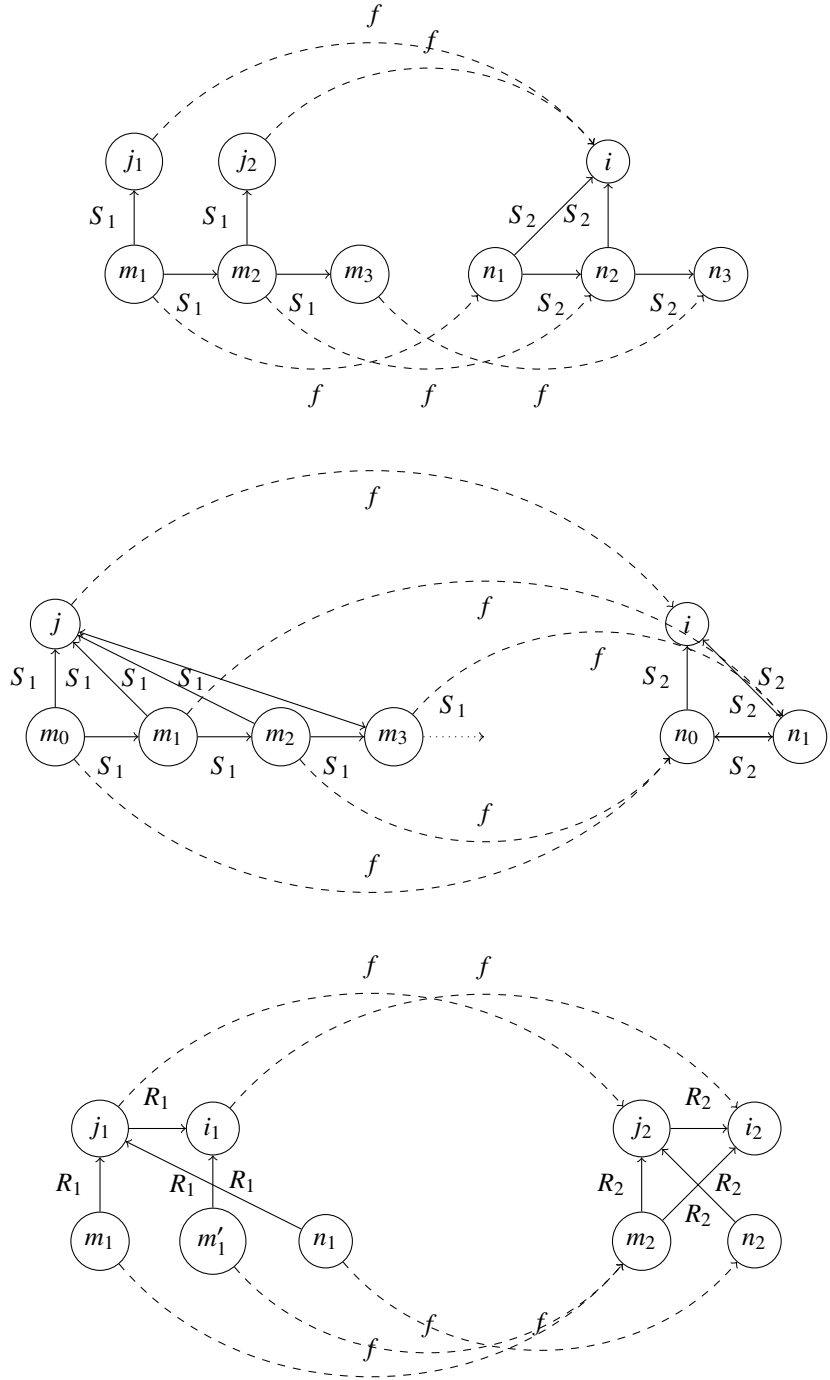

\section{Formalizing the deliberation scenarios}\label{interrogative:sec:Formalizing deliberation}

In this section, we use the framework introduced in Sections \ref{interrogative:sec:interrogative_agendas_and_coalitions} and \ref{interrogative:sec:interaction conditions} to formalize the scenarios described in Sections \ref{interrogative:sec: hiring committee} and  \ref{interrogative:sec:car}. 

\paragraph{Hiring committee.} The  algebra of coalitions  is the two-atom Boolean algebra $\mathbb{C}$ of the picture below. We let $\jty(\mathbb{C}) := \{\aga, \agb\}$. Since $W$ is the 8-element poset described in Section  \ref{interrogative:sec: hiring committee}, by Propositions \ref{interrogative:prop:charact meet-irr} and \ref{interrogative:prop:charact join-irr}, the lattice $E(W)$ has 127 completely meet-irreducible elements (coatoms) and 28 completely join-irreducible elements (atoms); 
the  algebra $\mathbb{D}$ of interrogative agendas  is the sub meet-semilattice  of $E(W)$ meet-generated by the coatoms  $\mathtt{p}: = \{(w, u)\mid w_r =u_r \text{ and } w_l =u_l\}$, $\mathtt{r}: = \{(w, u)\mid w_p =u_p \text{ and } w_l =u_l \}$, and $\mathtt{l}: = \{(w, u)\mid w_p =u_p\text{ and } w_r =u_r\}$. As discussed in Section  \ref{interrogative:sec: hiring committee},  $\mathbb{D}$ is the three-coatom Boolean algebra (see Footnote \ref{footn: D is BA}) represented in the picture below.

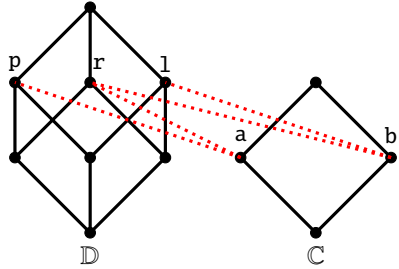
\begin{center}
\begin{tikzpicture}
\draw[very thick] (-1, 0) -- (-1, 1) --
	(0, 0) -- (1, 1) -- (1, 0) -- (0, 1) -- (-1, 0);
	\draw[very thick] (0, 2) -- (-1, 1);
\draw[very thick] (0, 2) -- (0, 1);
\draw[very thick] (0, 2) -- (1, 1);
	\draw[very thick] (0, -1) -- (-1, 0);
\draw[very thick] (0, -1) -- (0, 0);
\draw[very thick] (0, -1) -- (1, 0);
	\filldraw[black] (0,-1) circle (2 pt);
	\filldraw[black] (0, 2) circle (2 pt);
    \filldraw[black] (-1, 1) circle (2 pt);
	\filldraw[black] (1, 1) circle (2 pt);
	\filldraw[black] (0, 1) circle (2 pt);
	\filldraw[black] (-1, 0) circle (2 pt);
	\filldraw[black] (1, 0) circle (2 pt);
	\filldraw[black] (0, 0) circle (2 pt);
		
	\draw (0, -1.3) node {$\mathbb{D}$};

\draw (-1, 1.25) node {{\small{$\mathtt{p}$}}};
    \draw (0.12, 1.25) node {{\small{$\mathtt{r}$}}};
    \draw (1, 1.25) node {{\small{$\mathtt{l}$}}};

\draw[very thick] (3, -1) -- (2, 0) --
	(3, 1) -- (4, 0) -- (3, -1);
\filldraw[black] (3,-1) circle (2 pt);
\filldraw[black] (2,0) circle (2 pt);
\filldraw[black] (3,1) circle (2 pt);
\filldraw[black] (4,0) circle (2 pt);
\draw (2, 0.3) node {{\small{$\mathtt{a}$}}};
    \draw (4, 0.3) node {{\small{$\mathtt{b}$}}};
    \draw (3, -1.3) node {$\mathbb{C}$};
    \draw [very thick, red, dotted] (2, 0) -- (-1, 1);
     \draw [very thick, red, dotted] (2, 0) -- (0, 1);
      \draw [very thick, red, dotted] (4, 0) -- (0, 1);
        \draw [very thick, red, dotted] (4, 0) -- (1, 1);
\end{tikzpicture}
\end{center}
The cognitive attitudes of agents $\aga$ and $\agb$ towards issues are encoded in the  relevance relation $R\subseteq \mty(\mathbb{D})\times \jty(\mathbb{C})$ represented in the picture:
\[R: = \{(\mathtt{p}, \aga), (\mathtt{r}, \aga), (\mathtt{r}, \agb), (\mathtt{l}, \agb)\},\]
which gives rise to the maps $\Diamond, {\rhd}: \mathbb{C}\to\mathbb{D}$ defined as follows: for every $c\in \mathbb{C}$,
\[\Diamond c = \begin{cases}
\mathtt{r} & \text{if } c =  \aga\cup \agb\\
\mathtt{p}\sqcap \mathtt{r} & \text{if } c = \aga\\
\mathtt{r}\sqcap \mathtt{l} & \text{if } c = \agb\\
\bot & \text{if } c = \varnothing.\\
\end{cases} \quad\quad\quad
{\rhd} c = \begin{cases}
\bot & \text{if } c =  \aga\cup\agb\\
\mathtt{p}\sqcap \mathtt{r} & \text{if } c = \aga\\
\mathtt{r}\sqcap \mathtt{l} & \text{if } c = \agb\\
\tau & \text{if } c = \varnothing.\\
\end{cases}
\]

As discussed in Section  \ref{interrogative:sec: hiring committee}, the agenda $\Diamond \aga = \mathtt{p}  \rand \mathtt{r}$ (Alan's agenda) leads to choosing John, while  $\Diamond \agb = \mathtt{r} \rand \mathtt{l}$ (Betty's agenda) gives no unequivocal decision. We can consider different possible deliberation scenarios. 

The outcome of the decision-making process will be determined by the agenda chosen by the coalition  $c = \aga \cup \agb$. Two primary possible choices are the {\em common} agenda $\Diamond c =  \mathtt{r} $, and the {\em distributed} agenda $\rhd c = \mathtt{p}\sqcap \mathtt{r} \sqcap \mathtt{l} $. The common agenda leads to choosing John over Mary, since John totally dominates  Mary on the reference letter parameter, while the distributed agenda yields no decision, given that Alan and Betty's individual agendas do not lead to the same decision (cf.~Remark \ref{interrogative:rem:distributed unanimity}). 

Intuitively, $\Diamond c$ and ${\rhd} c$ arise from the attempt to form a `group agenda' by aggregating the individual preferences of the agents using elementary aggregation rules (namely, intersection and union of issues, respectively). As such, decision procedures yielding $\Diamond c$ and ${\rhd} c$ are conceptually more akin to  voting procedures than to  authentic deliberations. Indeed, the purpose of  deliberation is precisely to let alternatives emerge which would fall out of the scope of a voting procedure. 
A minimal way in which  these specific deliberation-dynamics can be captured is to include the substitution relation $S\subseteq \mty(\mathbb{D})\times \jty(\mathbb{C})\times \mty(\mathbb{D})$ and its associated operation $\pdra: \mathbb{C}\times \mathbb{D}\to\mathbb{D}$ in the present framework. While  $R$ encodes  the cognitive attitudes of the various agents towards issues  `in principle', the relation $S$ captures a more  operational type of information; namely,   which parameters each agent would or would not wish to see in the final outcome.\footnote{Alternatively, the substitution relation $S$ can be understood as encoding information on which parameters each agent would or would not consider negotiable; for instance, $S(\texttt{m}, \texttt{i}, \texttt{m})$ and not $S(\texttt{n}, \texttt{i}, \texttt{m})$ for any $\texttt{n}\neq \texttt{m}$ indicates that agent \texttt{i} considers issue \texttt{m} non negotiable.}

Then $\mathtt{x}\pdra \Diamond \mathtt{y}$ represents  the agenda into which agent $\mathtt{x}$ would like to transform agent $\mathtt{y}$'s agenda. Using these terms, one  candidate for the agenda  which   determines the final outcome according to coalition $c$ is $\bigsqcap_{\texttt{i},\texttt{j}\leq c, \texttt{i}\neq\texttt{j}}\mathtt{x}\pdra \Diamond \mathtt{y}$.


Different outcomes of the decision-making will arise, depending on the different concrete choices of the relation $S$.
To illustrate this point, let  the cognitive attitudes of agents towards other agents' interrogative agendas be encoded in the following substitution relation:
\[S: = \{(\mathtt{p}, \aga, \mathtt{p}), (\mathtt{r}, \aga, \mathtt{r}), (\mathtt{p}, \aga, \mathtt{l}), (\mathtt{r}, \aga, \mathtt{l}),
(\mathtt{r}, \agb, \mathtt{r}), (\mathtt{l}, \agb, \mathtt{l}), (\mathtt{r}, \agb, \mathtt{p}), (\mathtt{l}, \agb, \mathtt{p})\}.\]
This relation   reflects a situation in which each agent would substitute each issue which is not supported by his/her own agenda with any issue in his/her agenda. Note that  $S$  is antisymmetric,  single-stepped hence transitive, and  both positively and negatively coherent with $R$.  The operation  $\pdra$ induced by $S$ is such that
\[\aga\pdra \Diamond \agb = \aga \pdra (\mathtt{r}\sqcap \mathtt{l}) = (\aga\pdra \mathtt{r})\sqcup (\aga\pdra \mathtt{l}) = \mathtt{r}\sqcup (\mathtt{r}\sqcap\mathtt{p})  = \mathtt{r} \]
\[\agb\pdra \Diamond \aga = \agb \pdra (\mathtt{r}\sqcap \mathtt{p}) = (\agb\pdra \mathtt{r})\sqcup (\agb\pdra \mathtt{p}) = \mathtt{r}\sqcup (\mathtt{r}\sqcap\mathtt{l})  = \mathtt{r} \]
Hence,   $(\aga\pdra \Diamond \agb)\sqcap (\agb\pdra \Diamond \aga) = \mathtt{r}$, which yields John over Mary.

The following substitution relation:
\[S = \{(\mathtt{p}, \aga, \mathtt{p}), (\mathtt{r}, \aga, \mathtt{r}), (\mathtt{l}, \aga, \mathtt{r}), (\mathtt{l}, \aga, \mathtt{l}),
(\mathtt{l}, \agb, \mathtt{l}), (\mathtt{r}, \agb, \mathtt{r}),  (\mathtt{p}, \agb, \mathtt{r}), (\mathtt{p}, \agb, \mathtt{p})\}\]
 reflects a situation in which not only each agent leaves undisturbed the issue(s) which are not supported by his/her own agenda, but is willing to let common issues be substituted by issues in the other agent's agenda. This instantiation of $S$
 is antisymmetric,  single-stepped hence transitive, and  positively but not  negatively coherent with $R$.
 The operation $\pdra$ associated with it is such that
\[\aga\pdra \Diamond \agb = \aga \pdra (\mathtt{r}\sqcap \mathtt{l}) = (\aga\pdra \mathtt{r})\sqcup (\aga\pdra \mathtt{l}) = (\mathtt{r}\sqcap \mathtt{l})\sqcup \mathtt{l}  = \mathtt{l} \]
\[\agb\pdra \Diamond \aga = \agb \pdra (\mathtt{r}\sqcap \mathtt{p}) = (\agb\pdra \mathtt{r})\sqcup (\agb\pdra \mathtt{p}) = (\mathtt{r}\sqcap \mathtt{p})\sqcup \mathtt{p}  = \mathtt{p} \]
Hence,   $(\aga\pdra \Diamond \agb)\sqcap (\agb\pdra \Diamond \aga)=\mathtt{p}\sqcap \mathtt{l}$, which yields Mary over John.

\paragraph{Choosing a car.} The algebra of coalitions is the same two-atom Boolean algebra $\mathbb{C}$ as in the hiring committee case study, with $\jty(\mathbb{C}) := \{\aga, \agb\}$.  
As discussed in Section \ref{interrogative:sec:car}, the algebra $\mathbb{D}$ of interrogative   agendas is the complete sub-meet semilattice of $E(W)$ meet-generated by the meet-irreducible elements (coatoms) of the form $e^\Sigma_{Y, k_\leq}$,  corresponding to the bi-partitions \[\mathcal{E}_{Y, k_\leq}\coloneqq\{\{w\in W\mid w_Y\leq k\}, \{w\in W\mid w_Y> k\}\}\] for each $Y \subseteq X$, $0 \leq k \leq |Y|$, where $X = \{s, f, p, t, m\}$ and $W = 2^X$. Hence, $M^\infty(\mathbb{D}) = \{e^\Sigma_{Y, k_{\leq}}\mid Y\subseteq X, 0\leq k\leq|Y|\}$.

Recall (cf.~Section \ref{interrogative:ssec:equivalence relations and preorders}) that any $e\in E(W)$ gives rise to a (pre-)order $\leq_e$ s.t.~$w \leq _e u$ iff for all $w'\in [w]_e$ some $u'\in [u]_e$ exists s.t.~$w'\leq_e u'$, and $e^\Sigma_Y$ prefers  $w$  over $u$  if  $w <_Y u$. By Lemma \ref{interrogative:lemma:preference-sum}, when $e:=e^\Sigma_{Y}$, the order $\leq_{e^\Sigma_{Y}}$ coincides with the sum-order $\leq_Y$. However, only a proper subset of elements of $\mathbb{D}$ are of the form $e^\Sigma_{Y}$ for some $Y\subseteq X$;  by Lemma \ref{interrogative:lem:sum_decomposition}, these elements are exactly those 
which are meet-generated by $\{e^\Sigma_{Y, k\leq}\mid 0\leq k\leq |Y|\}$  for some $Y\subseteq X$. As discussed in Section \ref{interrogative:sec:car}, Alan and Betty's interrogative agendas  are $e_{\aga} := e^\Sigma_{A} = \bigsqcap_{0\leq k\leq 3} e^\Sigma_{A, k\leq}$ and $e_{\agb} := e^\Sigma_{B} = \bigsqcap_{0\leq k\leq 3} e^\Sigma_{B, k\leq}$ respectively, where $A=\{s,f,p\}$, and $B=\{f,t,m\}$. Hence,
%
their cognitive attitudes towards issues can be encoded in the  following relevance relation 
$R\subseteq \mty(\mathbb{D})\times \jty(\mathbb{C})$:
\[R: = \{(e^\Sigma_{A, k_\leq}, \aga), (e^\Sigma_{B, k_\leq}, \agb) \mid 0 \leq k \leq 3 \}.\]
The relation $R$ gives rise to the maps $\Diamond, {\rhd}: \mathbb{C}\to\mathbb{D}$ defined as follows: for every $c\in \mathbb{C}$,
\[\Diamond c = \begin{cases}
e^\Sigma_A \sqcup e^\Sigma_{B} = \tau & \text{if } c =  \aga\cup \agb \\ 
e^\Sigma_{A} & \text{if } c = \aga\\
 e^\Sigma_{B} & \text{if } c = \agb\\
\bot & \text{if } c = \varnothing.\\
\end{cases} \quad\quad\quad
{\rhd} c = \begin{cases}
e^\Sigma_{A} \sqcap e^\Sigma_{B} & \text{if } c =  \aga\cup\agb\\
e^\Sigma_{A} & \text{if } c = \aga\\
 e^\Sigma_{B} & \text{if } c = \agb\\
\tau & \text{if } c = \varnothing.\\
\end{cases}
\]

As discussed in Section \ref{interrogative:sec:car}, a straightforward way for Alan and Betty  to determine a common interrogative agenda under the sum rule is to consider $e^\sum_{A\cap B} = e_f$ or $e^\sum_{A\cup B} = e^\Sigma_{X}$.  Each such choice leads to a decision, and in fact these choices lead to opposite decisions.However, the deliberation dynamics might also yield different candidates for a common interrogative agenda, which are not linked to parameters as directly as the previous ones.  

The next candidates are $\Diamond c  = e_\aga\sqcup e_\agb$ and ${\rhd} c= e_\aga\sqcap e_\agb$ for $c:=\aga\cup\agb$.

The agenda  $\Diamond c = \tau = W\times W$ yields no decision, since all profiles, hence $C_1$ and $C_2$, are equivalent under it (i.e.~$C_1 \leq_{\Diamond c} C_2$ and $C_2 \leq_{\Diamond c} C_1$).  

Let us show that  ${\rhd}c$ yields no decision, as $C_1\not \leq_{\rhd c}C_2$ and $C_2\not\leq_{\rhd c}C_1$. 
By definition, for any $w, u\in W$, $(w, u)\in {\rhd}c$ iff $(w, u)\in e_\aga$ and $(w, u)\in e_\agb$, i.e.~$w_A = u_A$ and $w_B = u_B$. Hence, $[C_1]_{\rhd c} = \{w'\mid w'_A = 2 \ \&\ w'_B = 1\}$ and  $[C_2]_{\rhd c} = \{w'\mid w'_A = 1 \ \&\ w'_B = 2\}$.
Moreover, $w\leq_{\rhd c} u$ iff for every $w'\in W$, if $w_A = w'_A$ and $w_B = w'_B$, then $w'\leq  u'$ for some profile $u'$ s.t.~$u_A = u'_A$ and $u_B = u'_B$, where $\leq$ is the pointwise (total domination) order.\footnote{Recall that the definition of $\leq_e$ for any $e\in E(W)$ depends on the choice of a `baseline' preorder on $W$ (cf.~Section \ref{interrogative:ssec:equivalence relations and preorders}). When  $W = 2^X$, its natural (pre)order is the pointwise order (total domination).}

Hence,  $C_1\leq_{\rhd c}C_2$ iff for every $w\in W$ s.t.~$w_A = 2$ and $w_B = 1$,  some $u\in W$ exists s.t. $w \leq u$ and $u_A = 1 $ and $ u_B = 2$. However, requiring that $w_A=2$ and $u_A=1$ of $w$ and $u$
implies that no such $u$ can dominate $w$ on all parameters in $A$, which means $w \not \leq u$.  
Thus,  $C_1 \not \leq_{\rhd c}C_2$. A similar reasoning based on $w_B=1$ and $u_B=2$ shows that $C_2 \not \leq_{\rhd c}C_1$.

Let  the cognitive attitudes of agents towards other agents' interrogative agendas be encoded in the following substitution relation:\footnote{For the sake of a simpler presentation, this choice of $S$ is not coherent with $R$, and it presupposes that e.g.~Alan knows beforehand the set $B$ of parameters considered by the other agent, and so does Betty. An alternative relation  which captures the same cognitive attitudes of the two agents towards any agendas different from their own is

\begin{center}
\begin{tabular}{cl}
$S: = $ & $ \{(e^\Sigma_{Y, 0_\leq}, \aga, e^\Sigma_{Y, k_\leq}) \mid A\neq Y \subseteq X, 0 \leq k \leq |Y|\} \cup$\\
& $\{ (e^\Sigma_{Y, 0_\leq}, \agb, e^\Sigma_{Y, k_\leq}), (e^\Sigma_{Y, 1_\leq}, \agb, e^\Sigma_{Y, k_\leq}) \mid B\neq Y \subseteq X, 2 \leq k \leq |Y|\} \cup$\\ 
& $\{(e^\Sigma_{Y, k_\leq}, \agb,e^\Sigma_{Y, k_\leq}) \mid B\neq Y \subseteq X, |Y|<2\} $. 
\end{tabular}
\end{center}
}

\begin{center}
$S: =  \{(e^\Sigma_{B, 0_\leq}, \aga, e^\Sigma_{B, k_\leq}) \mid  0 \leq k \leq 3\} \cup\{ (e^\Sigma_{A, 0_\leq}, \agb, e^\Sigma_{A, k_\leq}), (e^\Sigma_{A, 1_\leq}, \agb, e^\Sigma_{A, k_\leq}) \mid  2 \leq k \leq 3\}. $
\end{center}

The relation above reflects a situation in which   Alan is willing to include the question `does the car perform well on some parameter Betty considers relevant?' in the aggregated agenda, while Betty is willing to include the question `does the car perform well on none, one, or more  parameters Alan  considers relevant?' in the aggregated agenda. The operation  $\pdra$ induced by $S$ is such that

\[\aga\pdra \Diamond \agb = \aga \pdra e^\Sigma_{B} = e^\Sigma_{B, 0_\leq}\quad \quad \agb\pdra \Diamond \aga= \aga \pdra e^\Sigma_{A} = e^\Sigma_{A, 0_\leq} \sqcap e^\Sigma_{A, 1_\leq}.\]
Hence, the aggregated agenda    $ (\aga\pdra \Diamond \agb)\sqcap (\agb\pdra \Diamond \aga)$ is $e = e^\Sigma_{B, 0_\leq} \sqcap e^\Sigma_{A, 0_\leq} \sqcap e^\Sigma_{A, 1_\leq} $, which  identifies any two profiles $w, u$ s.t.~$w_Y  = 0$ iff $u_Y  = 0$ for each $Y\in\{A, B\}$, and $w_A\leq 1$ iff $u_A\leq 1$. Therefore, $C_1 = (1, 0, 1, 0, 1)$ is $e$-equivalent to $(1, 1, 1, 1, 1)$, and the latter profile dominates any other profile, and hence dominates any profile equivalent to $C_2$, which shows that $C_2\leq_e C_1$. Moreover, any profile $u$ which is $e$-equivalent to $C_2 = (0, 1, 0, 1, 0)$ must have $u_A\leq 1$; therefore, no such profile can dominate $C_1$ in the pointwise order, which shows that $C_1\nleq_e C_2$. Therefore, $ (\aga\pdra \Diamond \agb)\sqcap (\agb\pdra \Diamond \aga)$ prefers $C_1$ over $C_2$.

 \section{Conclusion and future work}\label{interrogative:sec:conclusion}
In this paper, a logico/algebraic framework is introduced for modelling scenarios of decision-making via deliberation in a multi-agent setting. In this framework,  the different cognitive stances of agents are conceptually modelled based on the notion of  {\em interrogative agenda} (cf. Section \ref{interrogative:ssec:Interrogative agendas and their logical formalizations}).  
Interrogative agendas are formally represented as elements of certain lattices of equivalence relations on the space of profiles. 
This formal framework describes and analyses deliberation scenarios under two {\em winning rules} (the {\em total domination} and the {\em sum} rule), and three domains of {\em scores} (i.e.~linear and 2-valued, linear and many-valued, non-linear).  
The expressivity  of this  framework is preliminarily explored, focusing particularly on the `modal definability' of examples of first-order conditions which encode the various attitudes of agents   towards  the `issues on the table' in the deliberation. 
The contributions of this paper pave the way to several interesting research directions, some of which are listed below.

\paragraph{Multi-type framework.} Similarly to \cite{bilkova2018logic}, the framework introduced in this paper has been introduced as an instance of the {\em multi-type methodology}, a metatheory for the modular development of syntactic and semantic formal tools specifically designed to represent and reason about the interaction of entities of different types. The coalitions and interrogative agendas studied in the present paper form distinct types and interact with one another in complex ways. The syntactic machinery  provided by the multi-type framework makes it possible to express and reason about salient properties of the scenarios that arise from such interactions simply and succinctly. 
The analysis in Section \ref{interrogative:sec:interrogative_agendas_and_coalitions} confirms that the operations interpreting the logical connectives have the required order-theoretic properties to be captured within the framework of logics with algebraic semantics based on normal (distributive) lattice expansions. Such logics have excellent basic semantic and syntactic properties,  such as the availability of cut-eliminable proof calculi for the basic framework and a large class of its axiomatic extensions \cite{greco2016unified}, decidability, and interpolation \cite{andrea-s-thesis}.\footnote{In \cite{greco2016unified, andrea-s-thesis}, these properties have been shown in a single-type (D)LE setting; however, their proofs immediately transfer to any multi-type (D)LE-setting.} 
 Areas of further research that would be valuable for the practical applicability of the present framework include the development of its proof theory and optimal algorithms for model checking and satisfiability.   

The multi-type language of the present paper lends itself to further natural and useful  extensions, in particular to language enrichments along the lines of  hybrid logic \cite{ArecesTenCateChapter}, the most basic of which simply adds  nominal and co-nominal variables, ranging exclusively over individual agents and issues, respectively, to enable direct reference to these. This expressivity can then be further leveraged through the addition of other syntactic mechanisms of hybrid logic, such as $@$-operators and binders, together with (suitable generalizations of) hybrid correspondence theory \cite{ten2005hybrid,conradie:robinson:2017:sahlqvist}, to overcome some of the expressive limitations highlighted by the non-definability results of Section \ref{interrogative:sec:interaction conditions}.

\paragraph{Union and intersection of parameters in agendas for the sum winning-rule.} In the case study of Section \ref{interrogative:sec:car}, we let the agendas of the two agents be  $e^\Sigma_A$ and  $e^\Sigma_B$  for some $A,B \subseteq X$, and discussed that $e^\Sigma_A \sqcup e^\Sigma_B$ and 
  $e^\Sigma_A \sqcap e^\Sigma_B$ are among the candidate aggregated agendas   for the coalition formed by these two agents. However, when the winning rule is the sum-rule, the connection of the candidates above to the parameters is diluted compared to $e^\Sigma_{A\cup B}$ and   $e^\Sigma_{A\cap B}$. 
  Motivated by this, it would be  interesting to also consider the operators $\oplus$ and $\otimes$ defined as follows: for any $A,B \subseteq X$,

\[
e^\Sigma_A \oplus e^\Sigma_B =  e^\Sigma_{A\cup B} \quad \text{and} \quad e^\Sigma_A \otimes e^\Sigma_B =  e^\Sigma_{A\cap B}.
\]
Studying the order-theoretic properties of these operators and their interaction with other operators in our framework is an interesting direction for future research. 

\paragraph{Game-theoretic perspective on deliberation.} In the deliberation processes considered in this paper, collective agendas are formed by aggregating each participant’s agenda either directly or after operations such as coarsening or applying a substitution relation/relations corresponding to different agents. However, we have not examined how agents might act strategically in this aggregation phase. In practice, participants could choose their initial agendas and modify them during deliberation in order to steer the outcome toward their own preferred decisions. It would therefore be valuable to develop a game-theoretic framework for such deliberative games under various rules for agenda formation and decision-making, and to analyze the strategies available to individual agents or coalitions. In particular, one could study how different information structures (complete versus partial information) affect an agent’s ability to enforce specific decisions, with implications for both theory and real-world applications.

\paragraph{Different modal axioms to model interaction between agents and their interrogative agendas.}
 In many practical situations, interactions between $R$, $I$ and $S$ in deliberation  satisfy additional conditions, many of which are modally definable, while other are not, as shown in  Proposition \ref{interrogative:prop:gt}. It would be interesting to  explore extensions with additional relations and operators, also in combination with the hybrid machinery mentioned above, which can make these conditions definable. 

\paragraph{Modelling outcomes internally.} The multi-type framework developed in this paper does not allow us to model the outcomes of deliberation explicitly. In the future, this framework can be extended with another type to model outcomes of deliberation. Formally, this type would contain pre-orders over different possible outcomes, i.e.~a  orders over possible  outcomes. Interrogative agendas would change the preference order between the outcomes, by transforming  pre-orders to  pre-orders. Such a framework would allow us to   describe and reason about the entire decision-making process in a single multi-type environment.

\bibliography{ref}
\bibliographystyle{plain}
\end{document}

%% file: macros.tex
\newcommand{\diamdot}{%
  \tikz[baseline=-0.5ex, scale=0.06]{
    \draw[thick] (0,2) -- (2,0) -- (0,-2) -- (-2,0) -- cycle;
    \fill (0,0) circle (0.6);
  }%
}
\newcommand{\diamdotb}{%
  \tikz[baseline=-0.5ex, scale=0.07]{
    \fill[black] (0,2) -- (2,0) -- (0,-2) -- (-2,0) -- cycle;
    \fill[white] (0,0) circle (0.6);
  }%
}
\newcommand{\boxdotb}{\blacksquare\!\!\!\!{\color{white}{\cdot\ }}}

\newcommand{\apdra}{=\!\!<}
\newcommand{\nomh}{\mathbf{h}}
\newcommand{\nomi}{\mathbf{i}}
\newcommand{\nomj}{\mathbf{j}}
\newcommand{\nomk}{\mathbf{k}}

\newcommand{\cnomm}{\mathbf{m}}
\newcommand{\cnomn}{\mathbf{n}}
\newcommand{\cnomo}{\mathbf{o}}

\renewcommand{\i}{\mathbf{i}}

\newcommand{\jty}{J^{\infty}}
\newcommand{\mty}{M^{\infty}}

\newcommand{\br}{\mbox{$\,\vartriangleright\mkern-15mu\rule[0.51ex]{1ex}{0.12ex}\,$}}

\newcommand{\rand}{\sqcap}
\newcommand{\ror}{\sqcup}

\newcommand{\abr}{\mbox{$\blacktriangleright\mkern-15.5mu\textcolor{white}{\rule[0.51ex]{1.2ex}{0.12ex}}$}}

\def\aga{\texttt{a}}
\def\agb{\texttt{b}}

\newcommand{\mand}{\vartriangle}

\newcommand{\nAND}{%
\mathrel{\ooalign{$\mbox{\TriangleUp}$\cr\kern0pt$\mbox{\rotatebox[origin=c]{180}{\TriangleUp}}$}}}

\newcommand{\nand}{%
\mathrel{\ooalign{$\vartriangle$\cr\kern0pt$\triangledown$}}}

\newcommand{\pdra}{-\!\!\!{<}}
\newcommand{\pdla}{{ >}\!\!\!-}

\DeclareRobustCommand 
\Compactldots{\mathinner{\ldotp\mkern-3mu\ldotp\mkern-3mu\ldotp}}
\renewcommand{\ldots}{%
  \Compactldots
}

\theoremstyle{plain}
\newtheorem{thm}{Theorem}

\newtheorem{proposition}[thm]{Proposition}

\newtheorem{lemma}[thm]{Lemma}

\theoremstyle{definition}

\newtheorem{definition}[thm]{Definition}

\newtheorem{remark}[thm]{Remark}